\documentclass[12pt,reqno,oneside]{amsart}
\usepackage[T1]{fontenc}
\usepackage{srcltx} 
\usepackage[left=1.5in,top=1.5in]{geometry}
\usepackage{har2nat}
\usepackage{natbib}
\usepackage{graphicx}
\allowdisplaybreaks[3]
\usepackage{calc}
\usepackage{setspace}
  \usepackage{times}

\usepackage{booktabs, cellspace, hhline}
\usepackage[dvipsnames]{xcolor}
\usepackage[
  colorlinks=true,linkcolor=blue, urlcolor=red]{hyperref}
  
\AtBeginDocument{\hypersetup{citecolor=Blue,linkcolor=black}}
\usepackage{color, colortbl}
\usepackage{booktabs}
\usepackage{threeparttable}
\usepackage{multirow,array}
\usepackage[dvipsnames]{xcolor}

\usepackage{tikz}
\usepackage{pgfplots}
\pgfplotsset{compat=newest,
samples=500,
every axis/.append style={
                    axis line style={semithick},
                    label style={font=\scriptsize},
                    tick label style={font=\scriptsize},
                    xlabel style={at={(ticklabel* cs:1)},anchor=north west},
                    ylabel style={at={(ticklabel* cs:1)},anchor=south west}
                 }
}

\tikzset{
    myarrow/.style = {-{Triangle[length = 1.5mm, width = 1.5mm]}}
}

\usepackage{graphicx}

\usetikzlibrary{decorations.pathreplacing}
\usetikzlibrary{calc,patterns,positioning}
\usepgfplotslibrary{fillbetween}
   \usetikzlibrary{arrows,snakes,backgrounds,automata,positioning,arrows.meta}
\usepackage[font=footnotesize]{caption}
\usetikzlibrary{arrows.meta}
\usepackage[outline]{contour} 
\tikzset{>=latex} 

\usepackage{adjustbox}
\usetikzlibrary{decorations.pathmorphing}
\tikzset{discont/.style={decoration={zigzag,segment length=12pt, amplitude=3pt},decorate}}
\def\discontarrow(#1)(#2)(#3)(#4);{
  \draw[discont] (#2) -- (#3);
  \draw[thick][->] (#1) -- (#2) (#3) -- (#4);
}
\renewcommand{\thesubsection}{\thesection.\Alph{subsection}}

\usepackage{amssymb}
\usepackage[redeflists]{IEEEtrantools}
\usepackage{multibib}
\theoremstyle{plain}
\newtheoremstyle{defn}%
    {7pt}
    {7pt}
    {}
    {}
    {\scshape}
    {---}
    {.7em}
    {}
\theoremstyle{defn}    
\newtheorem{lemma}{Lemma}

\newtheorem{assm}{Assumption}

\newtheorem{corollary}{Corollary}
\newtheorem{prop}{Proposition}

\usepackage{graphicx} 
\newcommand{\divby}[2]{#1 \mathord{\left/ \vphantom{#1 #2} \right.}
    \kern-\nulldelimiterspace #2}

\renewcommand{\thesection}{\Roman{section}}

\usepackage{graphicx} 
\usepackage{pgf,tikz}
\usepackage{pgfplots}
\usetikzlibrary{babel} 
\usepackage{circuitikz}
\usepackage{siunitx}
\usetikzlibrary{arrows,calc}
\usepackage{relsize}

\usepackage{tcolorbox}
\usepgfplotslibrary{fillbetween}
\usepackage{enumitem}
\usepackage{istgame}
\usepackage{blkarray}
\tikzset{chance node/.style={hollow node, minimum size=2mm}}

\tikzset{normal node/.style={circle, black, minimum size=2mm}}

\title{
There is power in general equilibrium}

\author{Juan Jacobo*}
\thanks{* Department of Economics, Externado University of Colombia, juan.jacobo@uexternado.edu.co. }%

\begin{document}

\begin{abstract}
The article develops a general equilibrium model where power relations are central in the determination of unemployment, profitability, and income distribution. The paper contributes to the ``market forces versus institutions'' debate by providing a unified model capable of identifying key interrelations between technical and institutional changes in the economy.  Empirically, the model is used to gauge the relative roles of technology and institutions in the behavior of the labor share, the unemployment rate, the capital-output ratio, and business profitability and demonstrates how they complement each other in providing an adequate narrative to  the  structural changes of the  US economy.

\bigskip
\smallskip
\noindent \emph{Keywords.}  Power relations,  unemployment, automation,  labor institutions. 


\end{abstract}

\maketitle


 

\section{Introduction}

Over the past 70 years, the US economy has seen dramatic changes in income distribution, technology adoption, corporate profitability, and  unemployment rates. The years  from the late 1940s to the mid-1970s marked a period with a considerable reduction in income  inequality and a slightly increasing labor share, albeit with a higher ratio of capital to value added, a surge in the rate of unemployment, and a deteriorated profitability of businesses. Most of these patterns reverted in the early 1980s and led to a new era with a sharply uneven distribution in favor of  upper income groups. 

While there have been many discussions about the causes of these macro patterns, there is not a fully compelling explanation. The prevailing theories can be divided into market-driven versus institution-driven stories. The market-driven approach posits that technical  change (particularly automation), globalization, and industrial concentration  have created a bias in favor of high-skilled labor and the owners of capital, which are commonly in the top percentiles of the distribution of income (see, e.g., \citeasnoun{autor2020fall, hemous2022rise, moll2022uneven}). The two main problems with  this approach is that it cannot account for the fact  that not all nations subject to similar technological forces have seen an equal rise of top income shares and that it is hard to reconcile with the behavior of key macro trends like the rate of unemployment and several measures of corporate profitability during the postwar period \cite{stansbury2020declining}.

The institution driven stories postulate that union memberships, minimum wages, tax policy,  preferences for redistribution, and broadly defined organizational practices  in the labor market had a major role in macroeconomic outcomes and the evolution of income inequality (see, e.g., \citeasnoun{piketty2014optimal,stansbury2020declining,farber2021unions}.) The difficulty is that it is generally challenging to represent  the multidimensional character of labor institutions in a tractable model that  highlights the relative role of each specific factor.

The first goal of this paper is to present a comprehensive  general equilibrium model capturing key aspects of the market-driven and institution-driven narratives to assess their relative roles in the evolution of inequality and macroeconomic outcomes. To do so, I merge the task-based formalism of \citeasnoun{acemoglu2018race} with the search and matching models of equilibrium unemployment, while relaxing the unrealistic  assumption that firms can and do include the ``required'' rate of return as a cost of production.  This presents a more realistic model of capitalist economies by explicitly revealing how corporate profitability is determined by power relations between workers and firms, and how these power relations are endogenously formed by norms and organizational practices defining the bargaining protocol of wages. Furthermore, the model explores the dynamic interrelation between technical and institutional changes, and provides  a clear and tractable framework illustrating how unemployment and the functional distribution of income  are affected by automation, labor productivity growth, and specific labor institutions like union membership and real minimum wages.

The second goal of the paper is to gauge the relative  roles that technical and institutional changes had in the US economy over the  postwar period by comparing the predicted paths of the model with their empirical counterparts.   I consider  macro-level time series of economic, political, and institutional data.  

Basing the initial analysis on economic time series evidence, and employing  a parsimonious  calibration  strategy where the only parameters which are directly estimated are the measure of automation  and  the bargaining power of labor  using the theoretical  equilibrium conditions, the model reaches two main results. First,    the rise and fall of worker power before and after the mid-1970s is probably the major structural change responsible for the behavior of the labor share,  corporate profitability, and the unemployment rate. This suggests that an adequate understanding of  macroeconomic trends requires a careful study of the institutional and politico-economic variables determining the bargaining power of labor. Second, technical change (particularly automation) is nonetheless a key factor  determining the behavior of the labor share and the ratio of capital to value added. Altogether, by studying a wide array of macroeconomic variables over the entire postwar period, the evidence shows that the market-driven and institution-driven stories  likely complement, rather than substitute, each other in providing a consistent narrative for the main events of the US economy. 

The time series on labor institutions are used to supplement the previous results in two ways.  First, they illustrate that the predicted paths of worker power derived from the calibration  strategy of the model are consistent with the observed variations in labor institutions in the US. Specifically,  worker power increased between the 1940s to the late 1970s when the institutional support to labor was generally rising, and decreased steadily  thereafter when unions, minimum wages, and top marginal income tax rates simultaneously declined. Second, the data  exhibits  a clear association between  the rise and fall of the  institutional support to labor  with the ``Communist threat'', which refers to the class compromise between capital and labor induced by the fear that communism could replace  the foundations of capitalism \cite{gerstle2022}.  This presents a plausible story explaining  why Democrat and Republican governments alike supported the construction of a welfare state in the US before the mid-1970s,  but dismantled  some of its foundations afterwards.

 Combining these empirical results, the model sheds new light  on widely studied phenomena like the wage-premium and the association of corporate markups with market concentration. The evidence  shows that corporate profitability is highly correlated with the wage-premium since the 1950s, suggesting that similar mechanisms driving up the rate of return of capital  are also raising  the relative wage of high-skilled labor. The data also indicates that  the behavior of  corporate markups is only consistent with the trends  of market concentration after the early 1980s \cite{kwon2023}, while it is  generally well aligned with the behavior of worker power throughout the  postwar period. Thus, given the centrality of labor power in explaining the behavior of business profitability, it is likely that the   relations between capital and labor have been key actors shaping  the behavior of the wage-premium and   corporate markups in the US.

 To the best of my knowledge, this is the first paper to connect---both theoretically and empirically---the growing literature on the political economy of income distribution, labor institutions, and political preferences (see, e.g., \citeasnoun{piketty2003income,piketty2014optimal,farber2021unions}) with the numerous studies on the  trends in the labor share, the unemployment rate, and  the capital-output ratio. Similar to  \citeasnoun{stansbury2020declining}, \citeasnoun{dinardo2000unions},   \citeasnoun{krueger2022theory}, \citeasnoun{taschereau2020union}, and \citeasnoun{acemoglu2022eclipse}, the paper establishes an explicit connection between  worker power,  the distribution of income, and the rents transferred from labor to capital. However, unlike  the cited literature, the model corrects for possible confounding factors by developing a methodology that explicitly distinguishes the relative roles of technological and institutional changes in economic dynamics. 

 The paper also contributes to the growing literature on the effects of technical progress and automation on labor demand and income distribution \cite{aghion1994growth,acemoglu2018race,hemous2022rise,moll2022uneven}. Relative to these papers,  I show how technical change is explicitly associated with technological unemployment in a dynamic setting, and why the effects of automation  always depend on the specific institutional arrangements defining the bargaining power of labor. Furthermore, the model establishes the conditions for a balanced growth path (BGP) with positive growth and reveals how they are associated with the  institutions enabling the existence of sufficiently large profits for firms. 
 
 Finally, this work extends on the literature attempting to explain the trend of key macroeconomic variables in the US economy \cite{goldin2010race,karabarbounis2014global,FARHI_GOURIO2018,autor2020fall,barkai2020declining,
 stansbury2020declining}. Similar to \citeasnoun{stansbury2020declining}, the paper identifies worker power as a major  source of the  structural changes in the US  over the postwar period. However, by revealing the links between technical  and   institutional changes,  the model also supports  the findings in \citeasnoun{bergholt2022decline} and \citeasnoun{moll2022uneven} by showing  that   automation contributed to the fall of the labor share in the mid-1970s and in the early 2000s, and to the rise of the capital-output ratio since the late 1960s.

 The next section describes the basic   environment of the model. Section \ref{sec:barg_protocol} defines the bargaining protocol of wages and its connection with the equilibrium rate of return of capital.  Section \ref{sec:equilibrium_dynamics} reveals the conditions for a general equilibrium with positive growth and derives the key results on transitional dynamics. Section \ref{sec:empirics} presents an approximate calibration to the model and evaluates the roles of technology and institutions in the structural changes of the US economy. Section \ref{sec:extensions} shows some channels through which  worker power is  associated  with the wage-premium  and disentangles the extent to which business markups are related to market concentration. Section \ref{sec:conclusions} concludes.  The main Appendix generalizes the model in  Section \ref{sec:model} and complements the theoretical results in Sections \ref{sec:barg_protocol} and \ref{sec:equilibrium_dynamics}.  The online Appendix presents all the relevant proofs and derivations of the paper, the details of the calibration exercise, and the description of the data    along with  additional robustness tests.

\section{Model}\label{sec:model}

This section presents the   technology and price structure of the model, describes the matching function and the dynamics of aggregate   employment and capital with  the automation and creation of new tasks,  and characterizes  the value functions of capitalists and workers.

\subsection{Environment} The description of the production process follows the  formalism of \citeasnoun{acemoglu2018race} by emphasizing the  role of capital and labor  in the production of tasks $j$ indexed over a normalized space $[M_{t}-1, M_{t}]$. Tasks with $j \in (J_{t}, M_{t}]$ are produced with labor, and have an effective unit cost $W_{t}/A^{l}_{t}(j)$---$W_{t}$  is the nominal wage per worker and $A^{l}_{t}(j)$ is the task-specific labor-augmenting technology. Respectively, tasks $j \in [M_{t}-1, J_{t}]$ are produced with capital at an effective unit cost $\delta P^{k}_{t}/A^{k}_{t}(j)$, where $\delta \in (0,1)$ is the depreciation rate,  $P^{k}_{t}$ is the price of  capital, and $A^{k}_{t}(j)$ is the capital-augmenting technology.

Throughout, the factor augmenting technologies are represented by:
 
 \begin{assm}\label{ass:task_function} $A^{k}_{t}(j) = A^{k}>0$ and $A^{l}_{t}(j)=e^{\alpha j}$, with $\alpha >0$. 
 \end{assm}

 Assumption \ref{ass:task_function} says that labor has a comparative advantage in higher-indexed tasks and guarantees the existence of a threshold $\tilde{J}_{t}$ such that 
 
 \begin{equation*}
e^{\alpha \tilde{J}_{t}}=\frac{W_{t} A^{k}}{\delta P^{k}_{t} }.
 \end{equation*}
 
When $ j \leq \tilde{J}_{t}$, tasks are produced with capital since it has a lower effective cost   than  labor. If $j > \tilde{J}_{t}$, the production of tasks is bounded by the existing technology and firms will only be able to automatize up to $J_{t}$. The unique threshold defining the assignment of tasks is consequently $ J^{*}_{t} = \min\{J_{t}, \tilde{J}_{t}\}$.

 Appendix \ref{subappendix:genmodelsec1} shows that, in this setup, the equilibrium  output can be expressed as an aggregate production function

 \begin{equation}\label{eq:agg_prod_text}
 Y_{t} =  \Bigg[(1-m^{*}_{t})^{1/\sigma}\big(A^{k}\; K_{t} \big)^{\frac{\sigma-1}{\sigma}} + \Big(\int_{0}^{m^{*}_{t}} e^{\alpha j} \;  \mathrm{d}j \Big)^{1/\sigma} \Big(e^{\alpha J^{*}_{t}}   L_{t}\Big)^{\frac{\sigma-1}{\sigma}}\Bigg]^{\frac{\sigma}{\sigma-1}},
  \end{equation}
 
where $K_{t}$ is the aggregate  capital stock, $L_{t} $ is  aggregate employment,  $P^{c}_{t}$ is the price index of costs of production satisfying the \emph{ideal price index condition},  and $m^{*}_{t}=M_{t}-J^{*}_{t}$ is the equilibrium measure of automation.


 \subsubsection{Prices and Growth} The economy-wide price  of the final output is given by

\begin{equation}\label{eq:final_price}
P_{t}=(1+\mu_{t})P^{c}_{t}. 
\end{equation}

The key characteristic of  \eqref{eq:final_price} is that firms  only realize a profit \emph{after} a commodity is produced and sold, meaning that the rate of return of capital, $\mu_{t}$,  cannot be included as a   cost of production. 

In the text, the exposition is simplified by assuming that the economy can convert one unit of output into $q_{t}=q$ units of capital, so that  $P^{k}_{t}/P_{t} = q^{-1}$ at any time $t$. This special case of an economy with \emph{investment-specific technological change} allows the existence of a BGP without the introduction of human capital accumulation or further discussions on the so-called ``capital-skill'' complementarity.\footnote{Appendix \ref{subappendix:genmodelsec1} presents  a generalized model    showing how investment-specific technological change can be incorporated to the analysis.}

Denoting the growth rate of any variable $X$ as $g_{X}$, the next lemma specifies the conditions for a BGP in  the economy  described above. 
 
 \begin{lemma}\label{lemma:management_text} Suppose that Assumption  \ref{ass:task_function} holds. Then in any BGP:  
 
 \begin{equation*}
 g_{K}=g_{Y} =g_{C}=g=  \alpha  \dot{M} 
 \end{equation*}

 \end{lemma} 
 
 Lemma \ref{lemma:management_text} is a simplified version of Lemma \ref{lemma:management} in Appendix \ref{subappendix:genmodelsec1},  used below to study how changes in the rate of automation or  in the pace  labor-augmenting technological progress affect  the economy; see Proposition \ref{prop:comp_stat}. 

\subsection{Matching and State Dynamics}\label{subsection:state_dynamics} Society is made of a unit measure of risk-neutral workers and a continuum of  potential firms (capitalists) with a  common discount rate $\rho$. Lower-case letters represent real  stationary variables, whereas stationary per-capita variables  are  denoted by $\hat{x}_{t}$.\footnote{For example, $w_{t} = W_{t}/\big(P_{t} e^{\alpha  (M_{t}-m^{*}_{t})} \big)$ and $\hat{y}_{t} = Y_{t}/\big(L_{t} e^{\alpha (M_{t}-m^{*}_{t})} \big)$. }

Employed workers are denoted by $L_{t}$ and the remaining $U_{t}=1-L_{t}$ are the unemployed. Vacancies are filled via a matching function $G(U_{t},V_{t})$ which exhibits constant returns to scale in $(U_{t},V_{t})$ and decreasing returns to scale in $V_{t}$ or $U_{t}$ separately. Labor market tightness is defined as the vacancy-unemployment ratio $\theta_{t} = V_{t}/U_{t}$,  the probability of filling a vacancy is $q(\theta_{t})= G(U_{t}, V_{t})/V_{t}$, and the job-finding probability per unit of time is $f(\theta_{t})=G(U_{t}, V_{t})/U_{t}$. 

Introducing changes in the automation and the creation of new tasks, the evolution of employment can be described by

\begin{equation*}
L_{t+\mathrm{d}t}=(1-\lambda_{0}) L_{t} +q(\theta_{t}) V_{t} -  \overbrace{\textcolor{black}{ \Big[ \underbrace{\int_{J^{*}_{t}}^{J^{*}_{t+\mathrm{d}t}} l_{t}(j) dj}_{\text{displacement effect}} - \underbrace{\int_{M_{t}}^{M_{t+\mathrm{d}t}} l_{t}(j) dj}_{\text{reinstatement effect}}}\Big] }^{U^{A}_{t} =\text{technological unemployment}}.
\end{equation*}

As usual, $\lambda_{0}$ is the exogenous  job-separation  rate. An important feature of the employment dynamics is that the displacement and reinstatement effects of the automation and creation of new tasks give rise to  a \emph{technological unemployment} component. Essentially, technological change creates  a displacement effect by replacing labor for capital, and a reinstatement effect by expanding the number of tasks on which labor has a comparative advantage.

In the limit when $\mathrm{d}t \rightarrow 0$, the employment dynamics equation becomes

\begin{equation}\label{eq:employment_dyn}
\dot{L}_{t} =  q(\theta_{t})V_{t} - \lambda_{t} L_{t}
\end{equation}

with $\lambda_{t} = \lambda_{0} + \partial U^{A}_{t}/\partial L_{t}$. The  intuition  of how technological change affects employment is well captured in the following lemma.

\begin{lemma}\label{lemma:tech_unemployment}
Suppose that Assumption \ref{ass:task_function} holds. Then technological unemployment is equal to

\begin{equation}\label{eq:tech_unemployment}
U^{A}_{t} = L_{t} \Big(1-e^{\alpha (\sigma-1)(\dot{M}_{t} - \dot{m}^{*}_{t})}\; \frac{ e^{\alpha (\sigma-1)(m^{*}_{t} + \dot{m}^{*}_{t})} -1}{e^{\alpha (\sigma-1)m^{*}_{t}}-1} \Big),
\end{equation}

and satisfies the relations in Table \ref{table:tech_unemployment} in the steady-state.

\begin{table}[h!]\caption{Scenarios of technological unemployment. }
\begin{center}
\resizebox{0.95\textwidth}{!}{
\renewcommand{\arraystretch}{1.35}
\begin{tabular}{l  c  c  c   c c }
\toprule
& $\frac{\partial U^{A}_{t}}{\partial L_{t}} \; ( \text{if } \dot{M}_{t} > 0)$ & $\frac{\partial U^{A}_{t}}{\partial L_{t}} \; ( \text{if } \dot{M}_{t} < 0)$ & $ \partial U^{A}_{L_{t}}/\partial \dot{M}_{t}  $ &  $\frac{\partial U^{A}_{L_{t}} }{\partial \dot{m}^{*}_{t}} \; ( \text{if } m^{*}_{t}=m_{t}) $ &   $\frac{\partial U^{A}_{L_{t}} }{\partial m^{*}_{t}}$ 
\\\cmidrule{2-6}
$\sigma >1 $ & $ < 0 $ & $>0$ & $<0$ & $<0$ &  $=0$\\
$\sigma \in (0,1) $ & $ > 0 $ & $<0$  & $>0$& $<0$ & $=0$\\
\bottomrule
\end{tabular}}
\end{center}
    \label{table:tech_unemployment}
\end{table}

\end{lemma}

The bottom line in Lemma \ref{lemma:tech_unemployment} is that  the  rate of technological unemployment   will decrease in an expanding economy when $\sigma >1$ and  will  increase with a higher rate of automation if  mechanizing tasks is economically feasible,  regardless of the value of $\sigma$.  

Analogous to the evolution of employment, the dynamics of aggregate capital with task automation can be expressed as 
 
 \begin{equation}\label{eq:capital_dyn}
 \dot{K}_{t} = I_{t} - \Big(\delta + \frac{\dot{m}^{*}_{t}}{1-m^{*}_{t}}  \Big) K_{t} =I_{t} - \delta_{t} K_{t},
 \end{equation}

where $\delta_{t}$ is the the total depreciation rate of capital. In the steady-state, when $\dot{m}^{*}_{t}=0$,  $\delta_{t}=\delta$.

\subsection{Value Functions}\label{subsection:value_functions} The  value function of an unemployed worker satisfies\footnote{With the exception of time, partial derivatives are denoted by subscripts. For instance, $y_{L_{t}}$ is the partial derivative of stationary output with respect to labor in period $t$. }

\begin{equation}\label{eq:un_value}
(\rho +\alpha  \; \dot{m}^{*}_{t} - g) \phi_{U_{t}} - \dot{\phi}_{U_{t}} = b_{t} + f(\theta_{t})\big( \phi_{L_{t}}-\phi_{U_{t}}\big).
\end{equation} 

In equation \eqref{eq:un_value},  the unemployed   receive flow utility $b_{t}$ and transition to employment with a rate $f(\theta_{t})$, in which case they receive a payoff $ \phi_{L_{t}}$ satisfying

\begin{equation}\label{eq:job_value_worker}
(\rho +\alpha\; \dot{m}^{*}_{t} - g)  \phi_{L_{t}} - \dot{\phi}_{L_{t}} = \lambda_{t}(\phi_{U_{t}}-\phi_{L_{t}})+ w_{t}. 
\end{equation} 

The employed worker receives flow utility from real wages, and at rate $\lambda_{t}$ the job is dissolved. An important feature of equation \eqref{eq:job_value_worker} is that the job separation rate is partly determined by technological progress, meaning that firms can reduce the worth of a job to a worker by increasing technological unemployment. In addition,  the effective discount rate  is the sum of two  components: (i) the common  time-preference parameter $\rho$; and (ii) variations  in the automation and creation of new tasks (which affect $\dot{m}_{t}$ and $g$, respectively). 

The  value of a vacancy for the firm is represented by

\begin{equation}\label{eq:value_vacancy_firm}
(\rho +\alpha\; \dot{m}^{*}_{t} - g)   \pi_{V_{t}} - \dot{\pi}_{V_{t}} =                                                                                                              q(\theta_{t})\big(\pi_{L_{t}}-\pi_{V_{t}}) - \xi_{t}
\end{equation}

Here the firm pays the flow cost of opening a vacancy, $\xi_{t}$, and matches with a worker at a rate $q(\theta_{t})$. Correspondingly, the value of a filled job for the firm satisfies

\begin{equation}\label{eq:value_job_firm}
(\rho +\alpha\;  \dot{m}^{*}_{t} - g)  \pi_{L_{t}} -\dot{\pi}_{L_{t}} = \lambda_{t}(\pi_{V_{t}}-\pi_{L_{t}}) +  \hat{y}_{t} - \hat{k}_{t} \hat{y}_{\hat{k}_{t}} -w_{t}, 
\end{equation}

where $ y_{L_{t}}-w_{t}= \hat{y}_{t} - \hat{k}_{t} \hat{y}_{\hat{k}_{t}} -w_{t}$ is the flow utility earned by the firm.

\section{Wage Bargaining and the Return of Capital}\label{sec:barg_protocol}

This section presents the core of the paper by showing  how aggregate employment and rate of return of capital   are  simultaneously determined by the bargaining protocol of wages. 

\subsection{Bargaining Protocol} The bargaining model is summarized in Figure \ref{fig:bargaining_protocol} by dividing wage outcomes in terms of two competing organizational practices. On one side we find the individual bargaining protocol, characterized for allowing employee and employer competition in the determination of wages. The competition process is represented by introducing a  minimum  time delay  affecting the probability that firms and workers will each find new  bargaining partners  to restart  the  negotiation of wages.   The   minimum time delay  is proportional to a parameter $T^{w}$, which plays a key role in the model by  capturing the firms' \emph{relative}  capacity of finding new workers willing to compete for lower wages. For instance,    $T^{w}$ can increase as a result of  policies or economic conditions which effectively  reduce  the  employment options and the  mobility  of workers,  as it is the case with   non-poaching and non-competing clauses  or with a higher monopsony power of firms   \cite{krueger2022theory, azar2020concentration}.\footnote{Throughout, I will  refer to $T^{w}$ as the \emph{hiring capacity of firms} or as the \emph{relative mobility of workers}.} Similarly,  $T^{w}$ can decrease by passing legislative action which mitigates the capacity of firms to lower wages through competition, as it can be expected by setting   higher minimum wages \cite[p. 18]{naidu2022there}.

In the left-hand side  of  Figure \ref{fig:bargaining_protocol}, the model introduces the possibility that workers will choose a collective bargaining process when negotiating wages.

\begin{figure}
 	\begin{center}
    \small
\begin{istgame}[font=\tiny]
\setistEllipseNodeStyle[white]
 \xtdistance{10mm}{22mm}
 \setxtinfosetstyle{dotted}
\istroot(0)(0,0)[normal node]{\textbf{Worker}}+2mm..62mm+
\istb{\text{Collective}}[above,sloped]
 \istb{\text{Individual}}[above,sloped]
  \endist
  
  \istroot(c1)(0-1)<-90>{$\underset{w}{\mathrm{max}} \Big(L_{t}(\phi_{L}-\phi_{U})\Big)^{\Gamma^{u}} \Big(\pi^{L}\Big)^{1-\Gamma^{u}}$}
  \endist
  

  
  \istroot(1)(0-2)[chance node]
  \istB{\frac{\theta }{T^{w}+\theta}}[al]
 \istB{\frac{T^{w}}{T^{w}+\theta}}[ar]
  \endist
    \istroot(2)(1-2)<0>{\text{f.1}}
  \istbm
    \istbm
      \endist
     \istroot(3)(1-1)<180>{\text{f.1}}
  \istbm
    \istbm
      \endist
      \xtInfosetO(2)(3){$t=0$}
            
      \cntmdistance*{9mm}{10mm}
      \cntmApreset[dashed]
      \istrootocntmA(1r)(1-2)  
      \istbA{w}[br] \istbm 
          \endist
                \istrootocntmA(1l)(1-1)  
      \istbA{w}[br] \istbm 
          \endist
          
          \istrooto(2a)(1r-1)<120>{\text{w.1}}
          \istbm \istb{N}[al]    \istb{Y}[r]{ (\mathbb{A^{I}},0)}[right]  \endist
          
                    \istrooto(2b)(1l-1)<120>{\text{w.1}}
     \istb{Y}[l]{ (\mathbb{A^{I}},0)}[left]  \istb{N}[ar]   \istbm \endist
         \istroot(4)(2a-2)<0>{\text{w.1}}
  \istbm
    \istbm
      \endist
        \istroot(5)(2b-2)<180>{\text{w.1}}
  \istbm
    \istbm
      \endist
       \istrootocntmA(2ra)(2a-2)  
      \istbA{w'}[br] \istbm 
          \endist
              \istrootocntmA(2lb)(2b-2)  
      \istbA{w'}[br] \istbm 
          \endist
          
              \xtInfoset*(4)(5){$t=\Delta$}[b]
  
          \istrooto(3a)(2ra-1)<120>{\text{f.1}}
          \istbm \istb{N}[al]    \istb{Y}[r]{ (\mathbb{A^{I}},\Delta)}[right]  \endist
                  \istrooto(3b)(2lb-1)<120>{\text{f.1}}
     \istb{Y}[l]{ (\mathbb{A^{I}},\Delta)}[left]    \istb{N}[ar]        \istbm   \endist

         \istroot(6)(3a-2)[white]<0>{\text{f.1}}
  \istbm
    \istb[dashed]
      \istbm
      \endist
      
        \istroot(7)(3b-2)[white]<180>{\text{f.1}}
  \istbm
    \istb[dashed]
      \istbm
      \endist
                \xtInfoset*(6)(7){$\vdots$}[b]
                
                      \istroot(8)(6-2)<0>{\text{w.1}}
  \istbm
    \istbm
      \endist
      
              \istroot(9)(7-2)<180>{\text{f.1}}
  \istbm
    \istbm
      \endist
                
       \istrootocntmA(Tra)(6-2)  
      \istbA{w'}[br] \istbm 
          \endist
              \istrootocntmA(Tlb)(7-2)  
      \istbA{w}[br] \istbm 
          \endist
          
                    \istrooto(Trf)(Tra-1)<120>{\text{f.1}}
 \istbm      \istbm   \istb{N}[al]   \istb{Y}[r]{ (\mathbb{A^{I}}, \frac{ \Delta}{q(\theta)})}[r] \istbA(3.5)<grow=0>{\text{Switch}}[a]    \endist
    
                    \istrooto(Trw)(Tlb-1)<120>{\text{w.1}}
 \istbA(3.5)<grow=-180>{\text{Switch}}[a]   \istb{Y}[l]{ (\mathbb{A^{I}}, \frac{ \Delta T^{w}}{f(\theta)})}[left]   \istb{N}[ar]     \istbm     \istbm      \endist
          
               \xtInfoset*(8)(9){$
               t=\Delta T(\theta)$}[b]
                
      \istroot(11)(Trf-3)<0>{\text{f.1}}
  \istbm
    \istbm
      \endist
      
       \istrootocntmA(Iof)(Trf-3)  
      \istbA{w}[br] \istbm 
          \endist
          
            \istroot(12)(Trf-5)<0>{\text{f.1}}
  \istbm
    \istbm
      \endist
      
             \istrootocntmA(I2of)(Trf-5)  
      \istbA{w}[br] \istbm 
          \endist
          
          \istrooto(12)(I2of-1)<120>{\text{w.2}}
       \istb{Y}[l]{ (\mathbb{A^{I}},0)}[left]        \istb{N}[ar] \istbm  \endist

          \istrooto(11)(Iof-1)<120>{\text{w.1}}
          \istbm \istb[dashed]     \istbm   \endist

      \istroot(13)(Trw-1)<-180>{\text{f.2}}
  \istbm
    \istbm
      \endist
      
            \istrootocntmA(Iof2)(Trw-1)  
      \istbA{w}[br] \istbm 
          \endist
          
              \istrooto(13)(Iof2-1)<120>{\text{w.1}}
          \istbm \istb{N}[al]    \istb{Y}[r]{ (\mathbb{A^{I}},0)}[right]  \endist

            \istroot(14)(Trw-3)<-180>{\text{w.1}}
  \istbm
    \istbm
      \endist

          \istrootocntmA(Iow1)(Trw-3)  
      \istbA{w'}[br] \istbm 
          \endist
          
                    \istrooto(14)(Iow1-1)<120>{\text{f.1}}
          \istbm \istb[dashed]     \istbm   \endist

              \xtInfoset*(Iow1)(Iof){$\begin{matrix}
               t=\Delta T(\theta)\\
               +\Delta\end{matrix}$}[b]

\end{istgame}
    \end{center}
    \caption{\scriptsize BARGAINING PROTOCOL.  \label{fig:bargaining_protocol}}
       \caption{ \scriptsize \emph{Notes--- The individual agreement payoff is $\mathbb{A^{I}}=(\pi_{L}, \phi_{L}-\phi_{U})$. Under collective bargaining,  the firm's payoff satisfies $(\rho + \alpha\;  \dot{m}^{*}_{t} -g) \pi^{L}=y_{t} - w_{t} L_{t} - \frac{\lambda_{t} \xi_{t} L_{t}}{q(\theta_{t})}$. The response time between offers is $\Delta$. The notation $f.i$ means firm $i$ and $w.j$ means worker $j$.}}
   
  \end{figure}

\subsubsection{Individual Bargaining}  The individual bargaining model has the following structure, shown as an extensive-form game in  Figure \ref{fig:bargaining_protocol}.

\begin{itemize}
\item The  first node in the right-hand side of   Figure \ref{fig:bargaining_protocol} is a chance node defining the type of competitive process between workers and firms. Each  worker takes a random sample $T(\theta) =\mathrm{min}\Big\{ T^{w}(\theta) \sim \mathcal{E}\big(T^{w}/f(\theta) \big),  T^{F}(\theta) \sim \mathcal{E}\big(1/q(\theta)\big)\Big\}$.\footnote{Here $\mathcal{E}\big(T^{w}/f(\theta) \big)$ is an exponential distribution with mean $T^{w}/f(\theta) $. } If  $T(\theta)=T^{F}(\theta)$, the firm will be the first to find a new partner to start bargaining  after a time delay  $\Delta T^{F}(\theta)$, measured by the average duration of a vacant job.  The contrary occurs when $T(\theta)=T^{w}(\theta)$, in which case the average time delay is  given by the mean duration of unemployment, $1/f(\theta)$, multiplied by the hiring capacity of firms, $T^{w}$.   

\item Given the law of large numbers, $T(\theta)=T^{F}(\theta)$ with probability $T^{w}/(T^{w}+\theta)$ and $T(\theta)=T^{w}(\theta)$ with probability $\theta/(\theta+ T^{w})$. 

\item If $T(\theta)=T^{F}(\theta)$, the  game follows the steps described by \citeasnoun{shaked1984involuntary}, which is depicted in the rightmost branch of Figure \ref{fig:bargaining_protocol}. However,  if $T(\theta)=T^{w}(\theta)$, the game replicates the   alternating  offers model of \citeasnoun{rubinstein1982perfect} since   firms are identical by assumption.
\end{itemize}

The following proposition summarizes the main results of the individual bargaining protocol.

\begin{prop}\label{prop:non_union_wage} Suppose that firms always make the first offer, that  $\Delta$ tends to zero, and that  the capitalists' response time is $\Delta_{f} = \gamma^{f}\Delta$, with $\gamma^{f} > 0$.   Applying the law of large numbers, 

\begin{enumerate}[label=(\emph{\roman*})]

\item if $T(\theta)=T^{w}/f(\theta_{t})$, \; $w^{na}_{t}  = b_{t}  + \Psi^{na}_{t}  \big(y_{L_{t}} - b_{t}\big)$, with

\begin{equation*}
 \Psi^{na}_{t}  = \frac{\Gamma^{na} [\rho+ \alpha \;  \dot{m}^{*}_{t}  -g +\lambda_{t} + f(\theta_{t})]}{\rho+ \alpha \;  \dot{m}^{*}_{t}  -g +\lambda_{t} +  \Gamma^{na} f(\theta_{t})},\; \;    \Gamma^{na}  =\frac{\gamma^{f}}{1+\gamma^{f}}.
\end{equation*}

\item  If $T(\theta)=1/q(\theta_{t})$, \; $w^{nb}_{t}  = b_{t}  + \Psi^{nb}_{t}  \big(y_{L_{t}} - b_{t}\big)$, \; with 

\begin{equation*}
\Psi^{nb}_{t}  = \frac{\Gamma^{nb}_{t} [\rho + \alpha\;  \dot{m}^{*}_{t}  -g +\lambda_{t} + f(\theta_{t})]}{\rho + \alpha\;  \dot{m}^{*}_{t}  -g +\lambda_{t}  +  \Gamma^{nb}_{t} f(\theta_{t})},
 \; \Gamma^{nb}_{t} =  \frac{\gamma^{f}(1-q(\theta_{t}))}{1+\gamma^{f} +q(\theta_{t}) (1-\gamma^{f})}.
\end{equation*}

\item The  average wage rate from individual bargaining is

\begin{equation}\label{eq:ind_wage}
w^{n}_{t} = b_{t} + \Psi^{n}_{t}  \big(y_{L_{t}} - b_{t}\big), \quad \text{with  }  \Psi^{n}_{t} =\frac{T^{w} \; \Psi^{nb}_{t} + \theta_{t} \; \Psi^{na}_{t}}{T^{w}+\theta_{t}}.
\end{equation}

\end{enumerate}
\end{prop}

In all cases, $\Gamma^{(\cdot)}_{t}$ and $\Psi^{(\cdot)}_{t}$ depict the \emph{intrinsic} and the \emph{actual} bargaining power of labor, with $ \Psi^{na}_{t} \geq \Psi^{n}_{t} \geq \Psi^{nb}_{t}$ for all $\theta \geq 0$.   The importance of Proposition \ref{prop:non_union_wage} can be well understood by studying how  worker power  changes  with variations in the labor market, the relative mobility of labor ($T^{w}$), the pace of automation, and the labor-augmenting  technical progress.  This is summarized in the following corollary.

\begin{corollary}\label{coro:barg_power} Suppose that the assumptions in Proposition \ref{prop:non_union_wage} hold. 

\begin{enumerate}[label=(\emph{\roman*})]

\item (\emph{Loose labor market}) If $\theta \rightarrow 0$, then
\begin{equation*}
\begin{cases}
\Psi^{na}_{t} \rightarrow \Psi^{n}_{t}  \rightarrow \Psi^{nb}_{t}  \rightarrow 0. \\
    \Gamma^{na} \rightarrow \Gamma^{nb}_{t}  \rightarrow 0.
  \end{cases}
\end{equation*}

\item (\emph{Tight labor market}) If  $\theta \rightarrow \infty$,  then

\begin{equation*}
\begin{cases}
\Psi^{nb}_{t} \rightarrow \Psi^{n}_{t}  \rightarrow \Psi^{na}_{t}  \rightarrow 1. \\
    \Gamma^{nb}_{t} \rightarrow  \Gamma^{na} <0.5 
  \end{cases}
\end{equation*}

\item (\emph{Relative mobility of labor}) A lower relative mobility of labor ($T^{w} \uparrow$) reduces the power of workers. That is, 

\begin{equation*}
\frac{\partial \Psi^{n}_{t}}{\partial T^{w}} = \frac{1}{T^{w} +\theta}\;  \big[\Psi^{nb}_{t}-\Psi^{n}_{t}\big] \leq 0 \quad  \text{for all } \theta \geq 0.
\end{equation*}

\item (\emph{Automation}) Suppose that mechanizing tasks is feasible.   Then 

\begin{equation*}
\frac{\partial \Psi^{n}_{t}}{\partial \dot{m}^{*}_{t}} >0 \;\;\;\;\;  \text{if    }  \left|\frac{\partial \lambda_{t}}{\partial \dot{m}^{*}_{t}}  \right| > \alpha. 
\end{equation*}

\item (\emph{Labor-augmenting  technical progress}) A higher equilibrium rate of growth always increases the bargaining power of labor if $\dot{M}_{t}>0$, i.e., $\partial \Psi^{n}_{t}/\partial  \dot{M}_{t} >0$ for all $\sigma >0$ and $\dot{M}_{t}>0$. Particularly, the following is true:

\begin{equation*}
\left.\frac{\partial \Psi^{n}_{t}}{\partial \dot{M}_{t}}\right\vert_{\sigma >1}  >\; \left.\frac{\partial \Psi^{n}_{t}}{\partial \dot{M}_{t}}\right\vert_{\sigma \in (0,1)}  >   \; \left.\frac{\partial \Psi^{n}_{t}}{\partial \dot{M}_{t}}\right\vert_{\sigma \in (0,1),  \dot{M}_{t}<0} \; \lesseqqgtr \; 0 . 
\end{equation*}

\end{enumerate}

\end{corollary}

The results in Corollary \ref{coro:barg_power} are quite intuitive and easy to understand. For instance, the model makes it clear that loose labor markets work as  endogenous mechanisms that  reduce the bargaining power of labor.  Conversely, a tight labor market empowers workers, though  it has  a limited impact on $\Gamma^{na}$ and $\Gamma^{nb}_{t}$  if firms always make the first offer. A relative reduction in the mobility of labor lowers  $\Psi^{n}_{t}$  by increasing the probability that workers will have to compete for each available vacancy. Finally,  extending on  the results of \citeasnoun{aghion1994growth} and \citeasnoun{acemoglu2018race}, technology can have two opposing effects on worker power.  On one hand,  higher automation is expected  to  weaken workers  when the increase in technological unemployment surpasses the reduction in the effective discount rate generated by the rise in the value of capital per  unit of time.  On the other hand,  labor power will generally benefit from a higher productivity growth  through the well-known  \emph{capitalization effect}

\subsubsection{Collective Bargaining } Similar to \citeasnoun{taschereau2020union}, collective  bargaining is modeled as a Nash bargaining problem between the firm and all its workers. If an agreement is reached, workers receive the net reward from employment and  the firm receives the corresponding equilibrium  value derived from the Hamilton-Jacobi-Bellman equation. Otherwise, the firm loses all its workers and  has to rehire its entire   workforce the following period. 

The next proposition presents the solution of the Nash bargaining problem in the left-hand side of Figure \ref{fig:bargaining_protocol}.

\begin{prop}\label{prop:union_wage}   The real wage under collective bargaining is given by

\begin{equation}\label{eq:union_wages}
w^{u}_{t}  = b_{t}  + \Psi^{u}_{t}  \Big[y_{L_{t}} -b_{t} +  \frac{\rho+  \alpha\;  \dot{m}^{*}_{t} -g +\lambda_{t}}{\rho+\alpha\;  \dot{m}^{*}_{t} -g} \big( \hat{y}_{t} -y_{L_{t}}\big) \Big] 
\end{equation}

with $ \Psi^{u}_{t}  = \frac{\Gamma^{u} [\rho+ \alpha\;  \dot{m}^{*}_{t}  -g +\lambda_{t} + f(\theta_{t})]}{\rho+ \alpha\;  \dot{m}^{*}_{t}  -g +\lambda_{t} +  \Gamma^{u} f(\theta_{t})}$.
\end{prop} 

The  solution in equation \eqref{eq:union_wages} is  similar to  the real wage  under individual bargaining, with the notable difference that the former introduces an additional component representing the  benefit that workers can extract from the  increase in the aggregate surplus.

\subsection{Labor Market Equilibrium}  Appendix  \ref{appendix:protocl_decision} presents a game-theoretic model determining the probability $P(\mathcal{U}=1|\cdot)$ that workers will choose a   collective bargaining strategy in the first node of Figure \ref{fig:bargaining_protocol}. This probability is a function of   the perceptions, attitudes, and biases that workers have when sharing  economic outcomes, and   the  preferences for political support of the government. In the main text, however,   $ P(\mathcal{U}_{t}=1| \cdot )$ is a a known datum,  so that the aggregate  wage can be expressed as

\begin{equation}\label{eq:wage_agg}
w_{t} = w^{n}_{t} + P(\mathcal{U}_{t}=1|\cdot )\big( w^{u}_{t} -w^{n}_{t}\big).
\end{equation}

 This is an average of the individual and collective bargaining solution, weighted by the relative advantages of each bargaining protocol and the social factors influencing the workers' perceptions, attitudes and biases.

Combining \eqref{eq:wage_agg} with equations \eqref{eq:final_price} and \eqref{eq:value_job_firm},  we reach the main result of the section. 

\begin{prop}\label{prop:equil_labor} Suppose that Assumption  \ref{ass:task_function}  holds. If firms reserve the right to manage and aggregate wages satisfy \eqref{eq:wage_agg}, then there exists a unique pair $(\mu^{*}_{t}, \theta^{*}_{t})$ resulting from the labor market equilibrium. 
\end{prop}

            \begin{figure}
\begin{center}
\begin{tikzpicture}
     \begin{axis}[
        axis x line=middle,
                    axis y line=middle,
         name=plot1,
          title={\footnotesize Panel A}, 
              title style={at={(axis description cs:0.1,0.965)}, anchor=south} , 
 width=0.37*\textwidth, height=4.75 cm,
ymin=0.1, ymax = 1.099,
xmin=0, xmax=5.2,
xshift=-0.25cm,
 xlabel = {\scriptsize  $\mu$},
    ylabel = {$w$},
 ytick = \empty, 
  xtick = {1.6},
  xticklabels={$\mu^{*}$},
        every axis x label/.style={
    at={(ticklabel* cs:0.95)},
    anchor=north,
},
every axis y label/.style={
    at={(ticklabel* cs:0.95)},
    anchor=east,
}
]

    \addplot[domain=0.7:4,
  line width=0.75pt,
black
]
{1.5/(1+x) }node at (axis cs:0.95,0.935) {\tiny $w=\frac{y^{*}_{L}}{1+\mu}$};

\addplot[color = black, dotted, thin] coordinates {(1.6, 0) (1.6, 0.58) (8, 0.58)};

\addplot[color = black, mark = *, only marks, mark size = 1.5pt] coordinates {(1.6, 0.58)};

\end{axis}

 \begin{axis}[%
    axis x line=middle,
                    axis y line=middle,
    name=plot2,
     title={\footnotesize Panel B },
                title style={at={(axis description cs:0.1,0.965)}, anchor=south} , 
    at=(plot1.right of south east), anchor=left of south west,
 width=0.37*\textwidth, height=4.75 cm,
ymin=0.1, ymax = 1.1,
xmin=0, xmax=5.2,
xshift=-0.35 cm,
 xlabel = {\scriptsize  $\theta$},
    ylabel = {$w$},
  xtick = {1.5,2.4,3.2},
  xticklabels={\scriptsize $\theta^{u}$,$\theta^{*}$,$\theta^{n}$},
    ytick = {0.45,0.58,0.73},
  yticklabels={$w^{n}$,$w^{*}$,$w^{u}$},
    every axis x label/.style={
    at={(ticklabel* cs:0.95)},
    anchor=north,
},
every axis y label/.style={
    at={(ticklabel* cs:0.95)},
    anchor=east,
}
]

    \addplot[domain=0.7:4.35,
  line width=0.75pt,
black
]
{1.0-0.2*x+0.01*x^2 }node at (axis cs:0.65,0.92) {\tiny L.D.};

    \addplot[domain=0.7:4.35,
      dashdotted  ,
  line width=1pt,
gray
]
{0.6+0.1*x-0.01*x^2 }node [very near end, above] {\tiny $w^{u}$};

    \addplot[domain=0.9:4.35,
  line width=0.75pt,
black
]
{0.35+0.12*x-0.01*x^2 }node [very near end, above] {\tiny $w$};

    \addplot[domain=1.1:4.35,
        dashed,
  line width=0.75pt,
BrickRed
]
{0.15+0.13*x-0.01*x^2 }node [very near end, above] {\tiny $w^{n}$};

\addplot[color = black, dotted, thin] coordinates {(3.2, 0) (3.2, 0.45) (0, 0.45)};

\addplot[color = black, dotted, thin] coordinates {(2.4, 0) (2.4, 0.58) (0, 0.58)};

\addplot[color = black, dotted, thin] coordinates {(1.5, 0) (1.5, 0.73) (0, 0.73)};

\addplot[color = BrickRed, mark = *, only marks, mark size = 1.5pt] coordinates {(3.2, 0.46)};
\addplot[color = black, mark = *, only marks, mark size = 1.5pt] coordinates {(2.4, 0.58)};
\addplot[color = gray, mark = *, only marks, mark size = 1.5pt] coordinates {(1.5, 0.73)};

\draw[thick,MidnightBlue, decoration={brace,mirror, raise=4pt},decorate](2.22,0.4)--node[yshift=-15pt,right =-6 pt] {\scalebox{.5}{ $P(\mathcal{U}=1|\cdot)$}} (2.22,0.58);

    \end{axis}

 \begin{axis}[%
    axis x line=middle,
                    axis y line=middle,
    name=plot3,
     title={\footnotesize Panel C},
                title style={at={(axis description cs:0.4,0.965)}, anchor=south} , 
    at=(plot2.right of south east), anchor=left of south west,
 width=0.37*\textwidth, height=4.75 cm,
ymin=0, ymax = 1.1, 
xmin=-3, xmax=5.2,
xtick={0},
ytick={0},
xshift=+0.1 cm,
 xlabel = {\scriptsize $\mathcal{W}$},
    ylabel = {},
  xtick = {0,1},
  xticklabels={0,\scriptsize  $\mathcal{W}^{*}$},
    every axis x label/.style={
    at={(ticklabel* cs:0.95)},
    anchor=north,
},
every axis y label/.style={
    at={(ticklabel* cs:0.95)},
    anchor=east,
}
]
    
    \addplot[domain=-2.1:3.5,
  line width=0.75pt,
black,
]
{1/(1+exp(-(x-1.2)/0.6)) };

\draw[
 dashed,
 gray
]
(-2.9,1) node[gray,below, xshift=2.9em,font=\tiny]{$1$}--(4.7,1);

\draw[thick,MidnightBlue, decoration={brace,mirror, raise=4pt},decorate](0.7,0.01)--node[right =5 pt] {\tiny $P(\mathcal{U}=1|\cdot)$} (0.7,0.38);

    \end{axis}
    
\end{tikzpicture}
            \caption{ \scriptsize LABOR MARKET EQUILIBRIUM.\label{fig:labor_market_comp_stat}}
          \caption*{ \scriptsize \emph{Notes--- In Panel C, $\mathcal{W}$ is the  measure  of the economic gains  from collective bargaining.  In Panel B, L.D. is equation \eqref{eq:value_job_firm},  $\{w^{u}, w^{n}\}$ is the solution under collective and individual bargaining, and $w^{*}$ is the equilibrium  wage rate.}}
          \end{center}
            \end{figure}

The logic behind Proposition \ref{prop:equil_labor} is  captured in Figure \ref{fig:labor_market_comp_stat}. First, given the model in Appendix \ref{appendix:protocl_decision}, workers combine  their preferences and political views with the relative advantages of collective bargaining, and decide on a vote share $P(\mathcal{U}=1|\cdot)$. From this, the aggregate wage and labor market tightness is determined in Panel B using equations \eqref{eq:value_job_firm}  and \eqref{eq:wage_agg}. Lastly, given the equilibrium in the labor market, equations \eqref{eq:agg_prod_text} and \eqref{eq:final_price}  determine $\mu^{*}$ and $\hat{k}^{*}$ simultaneously.

The next corollary presents a simple expression of the  labor share in terms of the technology and the institutions determining the equilibrium rate of return.  
  
  \begin{corollary}\label{corollary:wage_share} In equilibrium, the labor share satisfies
  
  \begin{equation}\label{eq:labor_share}
  \Omega^{*}_{t} = \frac{1}{1+\mu^{*}_{t}}\times \Bigg[  1+ \Bigg(\frac{(1-m^{*}_{t})\alpha (\sigma-1)}{e^{\alpha (\sigma-1)m^{*}_{t}}-1}\Bigg)^{1/\sigma} \big(\hat{k}^{*}_{t}\big)^{\frac{\sigma-1}{\sigma}}\Bigg]^{-1}
  \end{equation}
 
  \end{corollary}
  
 Given the results in Proposition \ref{prop:equil_labor}, the  first term in the right-hand side of \eqref{eq:labor_share}   provides  a link between  nonmarket mechanisms such as labor institutions and political preferences with worker power, and worker power with the rate of return of  capital.  The second component on the right-hand side  of \eqref{eq:labor_share}   is similar to the expression of the wage share obtained by \citeasnoun{acemoglu2018race}; the difference explained by the fact that here the rate of return  is not a cost of production.   Altogether, equation \eqref{eq:labor_share} can reconcile the literature on labor institutions and technological change by showing how each component can potentially affect the  labor share over time.

\section{Equilibrium and  Dynamics}\label{sec:equilibrium_dynamics}

This  section  presents the equilibrium conditions and the dynamic properties of the model.

  \subsection{Equilibrium Analysis}\label{subsection:equilibrium_analysis} Assuming that $\dot{m}^{*}_{t}$ is determined exogenously, Online Appendix \ref{subappendix:gen_equil}  shows that the equilibrium can be characterized by a system of four differential equations (in terms of $\{L_{t}, \theta_{t}, \hat{k}_{t}, \hat{c}_{t}\}$)  consistent with a BGP  with positive growth. This is summarized in the following result. 
  
  \begin{prop}\label{prop:gen_equilibrium} Suppose that Assumption  \ref{ass:task_function} holds.  The economy admits a unique and locally stable equilibrium BGP with positive growth\footnote{Figure \ref{fig:stability_condition_euler} in online Appendix \ref{appendix:robustness} shows that---with the exception of the early 1980s--- equation \eqref{eq:min_surplus}   is satisfied in the US.} 
  
  \begin{equation}\label{eq:harrod_growth}
  g = s^{*}_{t} (r^{*}_{t} - \chi^{*}_{t})
  \end{equation}
  
  if 
  
    \begin{equation}\label{eq:min_surplus}
\mu^{*}_{t} > \frac{g}{\delta} > \mu^{\text{min}}_{t}.
  \end{equation}
  
  Where $r^{*} = q \hat{y}^{*} \mu^{*}/(\hat{k}^{*}( 1+\mu^{*}))$  is the equilibrium rate of profit, $\chi^{*}= q (\hat{\xi}^{*} \; V^{*} + \hat{\tau}^{*})/\hat{k}^{*}$ is the equilibrium  sum of stationary taxes and  vacancy expenses  per unit of capital,  $s^{*} \in (0,1)$ is the equilibrium savings rate, and $\mu^{\text{min}}$ is the rate of return of capital for which $\hat{c}=0$ (see equation \eqref{eq:min_surplus2} in Online Appendix \ref{subappendix:gen_equil}). 
  \end{prop}

\begin{figure}
\begin{center}
 \begin{tikzpicture}    
  \begin{axis}[
        axis x line=middle,
                    axis y line=middle,
         name=plotramsey,
          title={}, 
              title style={at={(axis description cs:0.1,0.965)}, anchor=south} , 
 width=0.75\textwidth, height=5.5 cm,
ymin=0, ymax = 12,
xmin=0, xmax=18,
xshift=-0.5cm,
 xlabel = {$\hat{k}_{t}$},
    ylabel = {$\hat{c}_{t}$},
  xtick = {6,8.3,11.5,13,15.5},
  xticklabels={$\hat{k}^{*}$,$\hat{k}^{\text{gold}}$,$\textcolor{red}{\hat{k}^{\text{max,u}}}$,$\textcolor{red}{\hat{k}^{*u}_{t}}$,$\hat{k}^{\text{max}}$},
    ytick = {7.3,8.05},
      yticklabels={$\hat{c}^{*}$,$\hat{c}^{\text{gold}}$},
        every axis x label/.style={
    at={(ticklabel* cs:0.98)},
    anchor=north,
},
  dot/.style={circle,fill=black,minimum size=4pt,inner sep=0pt,
            outer sep=-0pt},
every axis y label/.style={
    at={(ticklabel* cs:0.98)},
    anchor=east,
}
]

\draw[thick,<->] (0,12) node[left]{$\hat{c}_{t}$}--(0,0) node[below right]{$0$}--(19,0) node[below]{$\hat{k}_{t}$};

\draw[thick] (6,0)--(6,10.5) node[above]{\scriptsize  $\dot{\hat{c}}_{t}=0$};
\draw[thick,dashed,red] (13,0)--(13,10) node[above,red]{\scriptsize  $\dot{\hat{c}}^{u}_{t}=0$ \tiny};

\draw[thick](1.5,0) ..controls (1.5,0) and (7.0,18) .. (15.5,0);    
\draw[thick,dashed,red](2.5,0) ..controls (2.5,1) and (9,8) .. (11.5,0);

    \draw (13,5) node[right]{\scriptsize  $\dot{\hat{k}}_{t}=0$};
        \draw (10,3.5) node[right,red]{\scriptsize  $\dot{\hat{k}}^{u}_{t}=0$};
      \draw (6,0) node[below]{$\hat{k}^{*}$};
     \draw[dotted] (0,7.3)--(6,7.3);
       \draw[dotted] (8.3,0)--(8.3,8.05);
         \draw (9,0) node[below]{$\hat{k}_{\text{gold}}$};
                  \draw (16,0) node[below]{$\hat{k}^{\text{max}}$};
       \draw (0,7.2) node[left]{$\hat{c}^*$};
         \draw[dotted] (0,8.05)--(8.3,8.05);
         \draw (0,8.6) node[left]{$\hat{c}_{\text{gold}}$};
         
\draw[thick,Aquamarine](3.5,2) ..controls (3.5,2) and (6,7.95) .. (8.5,11) node[color=Aquamarine,
	currarrow,
	pos=0.5, 
	xscale=1,
	sloped,
	scale=1]{};
	
	\draw[thick,Aquamarine](3.5,2) ..controls (3.5,2) and (6,7.95) .. (8.5,11)node[
	currarrow,
	color=Aquamarine,
	pos=0.85, 
	xscale=-1,
	sloped,
	scale=1]{};
	
		\draw[black](3.25,2.5) ..controls (3.25,2.5) and (5,7) .. (1,7.5)node[
	currarrow,
	pos=0.75, 
	xscale=-1,
	sloped,
	scale=1]{};

		\draw[black](3.25,2.5) ..controls (3.25,2.5) and (5,7) .. (1,7.5)node[
	currarrow,
	pos=0.25, 
	xscale=1,
	sloped,
	scale=1]{};
		\draw[black](4,1.5) ..controls (4,1.5) and (5,7) .. (13,0.5)node[
	currarrow,
	pos=0.75, 
	xscale=1,
	sloped,
	scale=1]{};
		\draw[black](4,1.5) ..controls (4,1.5) and (5,7) .. (13,0.5)node[
	currarrow,
	pos=0.25, 
	xscale=1,
	sloped,
	scale=1]{};
	
\addplot[color = Aquamarine, mark = *, only marks, mark size = 2pt] coordinates {(6, 7.26)};
\addplot[color = black, mark = *, only marks, mark size = 2pt] coordinates {(3.5, 2)};
	     
	        \end{axis}
	        
\end{tikzpicture}
\caption{\scriptsize DYNAMIC EQUILIBRIUM.\label{fig:ramsey_df}}
\caption*{\scriptsize \emph{Notes--- $\hat{k}^{\text{max}}$ is the  stationary  value of capital at $\mu^{\text{min}}$, $\hat{k}^{\text{gold}}$ is evaluated at $\mu=g/\delta$, and $\hat{k}^{*}$ is evaluated at $\mu^{*}$. The  dashed  lines represent an economically unstable system where $\mu^{*} < \mu^{\text{min}}$.}}
\end{center}
\end{figure}

The expression in \eqref{eq:harrod_growth} is   analogous to  Solow's fundamental equation under the assumption that all savings are made by firms and that capitalists have to pay taxes and vacancy expenses. The  novelty in Proposition \ref{prop:gen_equilibrium}  is that---because the return of capital is a surplus over costs of production---a BGP equilibrium with positive growth requires specific social and institutional arrangements allowing the existence of sufficiently large profits. For instance, given the structure in Appendix \ref{subappendix:arbitrage},  Figure \ref{fig:ramsey_df} shows that if  $\mu <\mu^{\text{min}}$, capitalists will be incapable of \emph{continuously} increasing capital outlays at a rate consistent with a  BGP, pay  taxes and vacancy expenses, and have a remnant for their own consumption, i.e., it is   an \emph{economically unfeasible growth path}.  From a political economy perspective, this implies that the support to workers is partly limited by the growth requirements of the system: very high growth probably requires a weak bargaining power of labor.\footnote{The causal relation need not hold in reverse order: low bargaining power of labor need not lead to high growth because, in a low productivity environment, the increase in the aggregate surplus acquired by capitalists will find little demand for additional units of productive capital.}

 \subsection{Transitional Dynamics}\label{subsection:comp_dynamics} This subsection studies the interrelations and dynamic implications of unanticipated and permanent changes in parameters related to technology and labor  institutions.

\begin{figure}
\begin{center}
\begin{tikzpicture}
\begin{axis}[
   axis x line=middle,
                    axis y line=middle,
 width=0.92\textwidth, height=7cm,
ymin=0.0, ymax = 1.1,
xmin=0, xmax=1,
xtick={0.32,0.42,0.52,0.62},
xticklabels={\scriptsize $\bar{q}(\mu^{(b)})$, \scriptsize $\textcolor{gray}{\bar{q}(\mu^{\text{min}})}$, \scriptsize $\bar{q}(\mu^{(a)})$, \scriptsize $\bar{q}(\mu^{(a)*})$},
ytick={0.01,1},
yticklabels={0,1},
extra x ticks={0.995}, extra x tick labels={$q$},
 extra y ticks={1.08}, extra y tick labels={$m$},
]

\path[fill=lightgray!30] (0.52,0.01) -- (0.52,1) -- (0.98,1) -- (0.98,0);

\draw[
 dashed,
 thick,
->
]
(0,1) node[above, xshift=9em,font=\scriptsize]{$m=1$}--(1,1);

\node[above, Aquamarine] at (0.15,1) { \scriptsize $q^{\text{min}}$};
\node[above, red] at (0.85,1) { \scriptsize $q^{\text{max}}$};

\addplot[color = black, mark = *, only marks, mark size = 2.5pt] coordinates {(0.47, 0.43)};
\node[below, black, xshift=0.5em] at (0.47, 0.43) { \scriptsize $(a)$};

\addplot[color = black, mark = o, only marks, mark size = 2.5pt] coordinates {(0.47, 0.67)};
\node[below, black] at (0.47, 0.67) { \scriptsize $(b)$};


\addplot[name path=A, domain=0.5195:0.85,
thick,
red,
]
{((x-0.52)^0.75)/(0.33^0.75)}
;

\addplot[name path=B, domain=0.15:0.521,
thick,
Aquamarine,
]
{((0.52-x)^0.6)/(0.37^0.6)}
;

\addplot[name path=C, domain=0.32:0.7,
thick,
dashed,
red,
]
{((x-0.32)^0.7)/(0.38^0.7)}
;

\addplot[name path=D, domain=0.08:0.32,
thick,
Aquamarine,
dashed
]
{((0.32-x)^0.8)/(0.24^0.8)}
node[Aquamarine,above, pos=0.6,xshift=-1.65em,yshift=-0.5em,font=\footnotesize]{$\bar{m}(q)$}
;

\addplot[name path=E, domain=0.6195:0.9,
very thick,
red,
dotted
]
{((x-0.62)^0.7)/(0.28^0.7)}
node[red,above, pos=0.1,xshift=2.35em,yshift=1.25em,font=\footnotesize]{$\tilde{m}(q)$}
;

\addplot[name path=F, domain=0.2:0.62,
very thick,
Aquamarine,
dotted
]
{((0.62-x)^0.65)/(0.42^0.65)}
;

\draw[
thick,
red,
->
]
(0.85,1)--(0.97,1);

\draw[
thick,
red,
dashed,
->
]
(0.7,1)--(0.97,1);

\draw[
thick,
red,
<-
]
(0.02,0.01)--(0.52,0.01);

\draw[
thick,
Aquamarine,
<-
]
(0.02,1)--(0.15,1);

\draw[
thick,
Aquamarine,
->
]
(0.52,0.01)--(0.97,0.01);

\draw[
thick,
Aquamarine,
dashed
]
(0.32,0.01)--(0.52,0.01);

\draw[
very thick,
Aquamarine,
dotted
]
(0.03,1)--(0.2,1);

\draw[
very thick,
red,
dotted
]
(0.52,0.01)--(0.62,0.01);


\node [right,font=\tiny,align=left] at (0.0, 0.12) {Region 1: $m^*=\bar{m}(q)$ \\
 $  w_{J}(m) < \frac{\delta}{A^{k} q}$};

\node [right,font=\tiny,align=left] at (0.65, 0.12) {Region 3: $m^*=m$\\
 $ w_{M}(m) > \frac{\delta}{ A^{k}q}$};

\node [right,font=\tiny,align=left] at (0.3, 0.9) {Region 2: $m^*=m$ \\
 $ w_{J}(m)>  \frac{\delta}{ A^{k} q } >  w_{M}(m) $};

\end{axis}
\end{tikzpicture}
\end{center}
\caption{\scriptsize AUTOMATION REGIONS.  \label{fig:automation_reg}}
\caption*{\scriptsize\emph{Notes---  By definition, $w_{J}(m) = \mathrm{lim}_{t \rightarrow \infty} W_{t}/(P_{t} e^{\alpha  J^{*}_{t}})$ and $w_{M}(m) = \mathrm{lim}_{t \rightarrow \infty} W_{t}/(P_{t} e^{\alpha  M_{t}})$. The solid lines represent the baseline scenario associated with point (a). The dashed lines represent a counterfactual  scenario with a higher $m$. The dotted lines illustrate a scenario similar to (a) but with   a higher rate of return.}}
\end{figure}

   The first important result is presented in Figure \ref{fig:automation_reg}, which depicts the main findings of Lemma \ref{lemma:lemma_2A_AR} in Appendix \ref{subappendix:Automation_regions}. Similar to \citeasnoun{acemoglu2018race}, the objective  is to illustrate how the effects of automation depend on the parameter  space defining  the behavior of the relative costs of labor and capital.  The three regions  in Figure \ref{fig:automation_reg} are determined by a critical value of the relative price of capital, $\bar{q}(\mu^{*})$, which is itself a function of the equilibrium rate of return. To the left of $\bar{q}(\cdot)$, there is a decreasing curve $\bar{m}(q)$ defined over $[q^{\text{min}}, \bar{q}(\cdot)]$ with $\bar{m}(\bar{q})=0$ and $\bar{m}(q^{\text{min}})=1$. Region 1 is the area of values  where  labor is relatively cheap, meaning that not all automated tasks will be produced with capital.  Correspondingly, there is an increasing curve $\tilde{m}(q)$ defined over $[\bar{q}(\cdot),q^{\text{max}}]$ with $\tilde{m}(\bar{q})=0$ and $\tilde{m}(q^{\text{max}})=1$. In the area of values with $m<\tilde{m}(q)$ we have that $w_{M}(m)>\delta/(A^{k} q)$, which implies that new tasks would not be adopted because they result in a reduction of  aggregate output. Finally, region 2 is the space where $m > \text{max}\big\{ \bar{m}(q), \tilde{m}(q)\big\}$, meaning that new tasks will raise aggregate output and will be immediately produced with capital.

  To understand the implications of this setting, consider the following three scenarios.\footnote{The notation $\partial m \rightarrow \partial \mu  \rightarrow \partial m $ reads: changes in the share of automation lead to changes in the rate of return and these lead to changes in the automation regions.}

  \begin{enumerate}[label=(\emph{\roman*})]
  
  \item  ($\partial \mu  \rightarrow \partial m $) Suppose  the economy is initially in point (a) of Figure \ref{fig:automation_reg} and  encounters  policy changes  lowering the power of   workers.  Given Proposition \ref{prop:equil_labor}, this   raises the rate of return to $\mu^{(a)^{*}}> \mu^{(a)}$  and the critical relative price of capital to $\bar{q}(\mu^{(a)*}) > \bar{q}(\mu^{(a)})$.\footnote{This is true because the relative price of capital when $m=0$ is $\bar{q}(\mu^{*})=\delta(1+\mu^{*})/A^{k}$; see equation \eqref{eq:online_marginal} in the main Appendix.} As a result,   the automation regions shift to the right (dotted lines) and we reach  a new equilibrium where  the weakening power of labor made machinery relatively superfluous. In this case not all tasks would be produced with capital  since  $w_{J}(m) < \delta/(A^{k}q)$.\footnote{History has many crude and vivid examples that help illustrate this scenario. For example, during the early stages of capitalism, child labor made the adoption of machinery relatively redundant in some tasks of manufacturing, mining and agriculture; see \citeasnoun[pp. 516-517]{marxcapital1}.}
  
    \item ($\partial m \rightarrow \partial \mu \rightarrow \partial m$) Suppose the  economy is initially in point (a) and moves to point (b) in Figure \ref{fig:automation_reg}. If the rise in $m$ is large enough, Proposition \ref{prop:comp_stat} says that $\mu^{(b)}$ can decrease so much that  $\mu^{(b)}<\mu^{\text{min}}$,  meaning that  the system can become unsustainable by an inadequate adoption of machinery. 
    
      \item ($\partial m \rightarrow \partial \mu \rightarrow \partial m$)  Suppose that the solid lines are now associated with point (b) and that there is an exogenous reduction in $m$ taking the system to point (a)  in Figure \ref{fig:automation_reg}. By Proposition \ref{prop:comp_stat}, this shifts the automation regions to the right (see dotted lines) by increasing the rate of return of capital. Thus, automation can lead to the paradoxical result of making machinery   relatively redundant by effectively reducing the relative cost of labor.
  \end{enumerate}

Given the  conclusions derived from Figure \ref{fig:automation_reg},   the  next proposition characterizes the economic implications  of  small unexpected changes in technology and labor institutions.
 
\begin{prop}\label{prop:comp_stat} Suppose that Assumption  \ref{ass:task_function}  holds and that the economy is initially in a  BGP with positive growth satisfying \eqref{eq:min_surplus}. Then, the dynamic equilibrium path converges  in finite time to a new BGP when there are  small unexpected changes in  technology and labor institutions. Particularly:

\begin{enumerate}[label=(\emph{\roman*})]
\item (\emph{Automation}) for $m > \text{max}\{\bar{m}(q), \tilde{m}(q)\}$ and $|\partial \lambda_{t}/\partial \dot{m}^{*}_{t}|>\alpha$, a small decrease in $m$ induces a two-stage transition.\footnote{The case where $|\partial \lambda_{t}/\partial \dot{m}^{*}_{t}|<\alpha$ is studied in Online Appendix \ref{subappendix:comp_stat} and illustrated in Figure \ref{fig:trans_dynamics}. The case where $m$ is either in region 1 or 3 in Figure \ref{fig:automation_reg} is  studied in \citeasnoun{acemoglu2018race}. } First, there is an initial shock $\dot{m}^{*}_{t}<0$ leading to a rise in $U_{t}$ and $\mu_{t}$, a decrease in  $\hat{k}_{t}/(\hat{y}_{t}q)$ and $\Omega_{t}$, and ambiguous effects on $\theta_{t}$ and $V_{t}$. Before the new steady-state is reached, the economy  moves to a new equilibrium with  $m'<m$ and $\dot{m}_{t}=0$. In the new BGP, $V$, $\theta$ and $\Omega$ are lower, whereas $\mu$, $U$ and $\hat{k}/(\hat{y} q)$ are higher for all $\sigma >0$.

\item (\emph{Labor-augmenting technical change}) a small increase in $\dot{M}$ lowers the asymptotic value of $\mu$,  and  raises the equilibrium labor share and capital-output ratio. If $\theta$ stays relatively constant, a small increase  in $\dot{M}$ raises the asymptotic values of $U$ and $V$ when $\sigma \in (0,1)$, and lowers  the  values of $U$ and $V$  when $\sigma>1$.

\item (\emph{Labor institutions}) a  permanent reduction in the support to labor---represented by, e.g.,  a higher $T^{w}$--- induces a new BGP  with lower asymptotic values of $\Omega$,   $\hat{k}/(q \hat{y})$ and $U$, and higher values of  $\mu$, $\theta$ and $V$,  for all $\sigma >0$.  
\end{enumerate}

\end{prop}

  Figure  \ref{fig:trans_dynamics} illustrates the dynamic responses  associated with the three shocks    in Proposition \ref{prop:comp_stat}.\footnote{The results in Proposition \ref{prop:comp_stat} and Figure \ref{fig:trans_dynamics} are broadly aligned  with---but can also be used to extend---the empirical findings of \citeasnoun{bergholt2022decline}.} Starting with   Figure  \ref{fig:trans_dynamics}, Panel A, the  initial stage of the transition-- represented over the interval $[t', t'']$--- features a decrease in $\dot{m}^{*}_{t}$ that gives rise to a higher rate of unemployment and an ambiguous effect on vacancies. The intuition  is that the automation shock    moves   labor demand \eqref{eq:value_job_firm} and labor supply \eqref{eq:wage_agg}  in the same direction by lowering the effective discount rate and by raising the Poisson probability of unemployment. As a consequence, though it is generally not possible to determine how $\theta$ will change,  it can be deduced that the rate of unemployment will increase given that the Beveridge curve moves outwards with the rise of technological unemployment; see Lemma \ref{lemma:tech_unemployment}. 

  If $|\partial \lambda_{t}/\partial \dot{m}^{*}_{t}| > \alpha$, the increase in $U^{A}_{L_{t}}$  will outweigh  the capitalization effect, which will move  the  labor demand  and labor supply  schedules downwards, and lead to an immediate   increase in the equilibrium rate of return by Proposition \ref{prop:equil_labor}.  Using  \eqref{eq:labor_share}, this translates into a lower labor share, as depicted by the green solid line in the lower panel of Figure \ref{fig:trans_dynamics}. The  polar case is obtained when $|\partial \lambda_{t}/\partial \dot{m}^{*}_{t}| < \alpha$, in which case the dominance of the capitalization effect   moves the  labor share upwards as represented by the orange dashdotted line in Figure \ref{fig:trans_dynamics}, Panel A.  The resulting variations in the equilibrium rate of return explain the different trajectories of the capital-output ratio over $[t', t'']$  in Figure  \ref{fig:trans_dynamics},  Panel A,  since  $\hat{k}/q\hat{y}$ will tend to move in the opposite direction of $\mu$ given the principle of diminishing marginal returns.

            \begin{figure}
\begin{center}
\begin{tikzpicture}
\matrix {
     \begin{axis}[
        axis x line=middle,
                    axis y line=middle,
                       title style={align=left},
          title={\scriptsize Panel A. Higher  automation \\ }, 
              title style={at={(axis description cs:0.3,0.9)}, anchor=south} , 
 width=0.45*\textwidth, height=3.75cm,
ymin=0.1, ymax = 1.099,
xmin=0, xmax=5.2,
xshift=-0.5cm,
    ylabel = {\scriptsize $V$},
 ytick = \empty, 
   xtick = \empty, 
        every axis x label/.style={
    at={(ticklabel* cs:0.95)},
    anchor=north,
},
every axis y label/.style={
    at={(ticklabel* cs:0.95)},
    anchor=east,
}
]

\path[fill=OliveGreen!20] (1,0.5) -- (1.52,0.57) -- (1.52,0.75) -- (1.0,0.65);

\draw [thick,OliveGreen,->] plot [smooth, tension=0.1] coordinates { (1,0.58) (1.5,0.67) (1.6,0.53) (2,0.4) (5,  0.4) };

\addplot[color = gray,dotted, thin] coordinates {(0.0, 0.58) (5.0, 0.58)};

\addplot[color = black, solid, thin] coordinates {(0.0, 0.58) (1.0, 0.58)};

\addplot[color = black, dotted, thin] coordinates {(1.0, 0) (1.0, 0.58)};

\addplot[color = black, dotted, thin] coordinates {(1.6, 0) (1.6, 0.58)};

\addplot[color = black, mark = o, only marks, mark size = 1.5pt] coordinates {(1.0, 0.58)};

;\end{axis}
& 
 \begin{axis}[%
    axis x line=middle,
                    axis y line=middle,
                    title style={align=left},
     title={\scriptsize Panel B. Lower insitutional support  \\ 
     \scriptsize  and higher labor-augmenting t.c.},
                title style={at={(axis description cs:0.4,0.9)}, anchor=south} , 
 width=0.45*\textwidth, height=3.75cm,
ymin=0.1, ymax = 1.1,
xmin=0, xmax=5.2,
xshift=-1 cm,
 ytick = \empty, 
  xtick = \empty, 
    every axis x label/.style={
    at={(ticklabel* cs:0.95)},
    anchor=north,
},
every axis y label/.style={
    at={(ticklabel* cs:0.95)},
    anchor=east,
}
]

\path[fill=Plum!15] (1,0.5) -- (2.5,0.57) -- (2.5,0.75) -- (1.0,0.65);
\path[fill=Plum!15] (2.5,0.57) -- (4.2,0.57) -- (4.2,0.75) -- (2.5,0.75);
\draw [very thick,Plum,dashed,->] plot [smooth, tension=0.4] coordinates { (1,0.58) (2.5,0.65) (4.2,  0.65) } node[Plum,right,font=\tiny]{$\sigma<1$};

\path[fill=BrickRed!15] (1,0.5) -- (2.5,0.42) -- (2.5,0.59) -- (1.0,0.65);
\path[fill=BrickRed!15] (2.5,0.59) -- (4.2,0.59) -- (4.2,0.42) -- (2.5,0.42);
\draw [thick,BrickRed,->] plot [smooth, tension=0.4] coordinates { (1,0.58) (2.5,0.5) (4.2,  0.5) } node[BrickRed,right,font=\tiny]{$\sigma>1$};

\draw [very thick,Plum,dashed,->] plot [smooth, tension=0.4] coordinates { (1,0.58) (2.5,0.65) (4.2,  0.65) } ;

\draw [thick,MidnightBlue,->] plot [smooth, tension=0.4] coordinates { (1,0.58) (1.0,0.95) (1.5,0.9) (2.5,0.85) (5,  0.85) };

\node[pin={[pin edge={Plum,->}, black, pin distance=0.55 cm]-95:{\tiny Labor-augmenting t.c.}}] at (axis cs:3.5,0.7 ) {};

\node[pin={[pin edge={BrickRed,->}, black, pin distance=0.4 cm]-70:{\tiny }}] at (axis cs:1.6,0.57 ) {};

\draw [thick,MidnightBlue,->] plot [smooth, tension=0.4] coordinates { (1,0.58) (1.0,0.95) (1.5,0.9) (2.5,0.85) (5,  0.85) };

\node[pin={[pin edge={MidnightBlue,->}, MidnightBlue, pin distance=0.7 cm]-1:{\tiny Institutions}}] at (axis cs:1.6,0.95 ) {};

\addplot[color = gray,dotted, thin] coordinates {(0.0, 0.58) (5.0, 0.58)};

\addplot[color = black, solid, thin] coordinates {(0.0, 0.58) (1.0, 0.58)};

\addplot[color = black, dotted, thin] coordinates {(1.0, 0) (1.0, 0.58)};

\addplot[color = black, mark = o, only marks, mark size = 1.5pt] coordinates {(1.0, 0.58)};

    \end{axis}
    \\
      \begin{axis}[
        axis x line=middle,
                    axis y line=middle,
              title style={at={(axis description cs:0.1,0.965)}, anchor=south} , 
 width=0.45*\textwidth, height=3.75cm,
ymin=0.1, ymax = 1.099,
xmin=0, xmax=5.2,
xshift=-0.5cm,
    ylabel = {\scriptsize $U$},
 ytick = \empty, 
  xtick = \empty, 
        every axis x label/.style={
    at={(ticklabel* cs:0.95)},
    anchor=north,
},
every axis y label/.style={
    at={(ticklabel* cs:0.95)},
    anchor=east,
}
]



\path[fill=OliveGreen!20] (1.6,0.68) -- (2,0.66) -- (2,0.84);

\path[fill=OliveGreen!20] (2.0,0.66) -- (5,0.66) -- (5,0.84)--(2,0.84);

\draw [thick,OliveGreen,->] plot [smooth, tension=0.1] coordinates { (1,0.58) (1.5,0.67) (1.6,0.68) (2,0.75) (5,  0.75) };

\addplot[color = gray,dotted, thin] coordinates {(0.0, 0.58) (5.0, 0.58)};

\addplot[color = black, solid, thin] coordinates {(0.0, 0.58) (1.0, 0.58)};

\addplot[color = black, dotted, thin] coordinates {(1.0, 0) (1.0, 0.58)};

\addplot[color = black, dotted, thin] coordinates {(1.6, 0) (1.6, 0.58)};

\addplot[color = black, mark = o, only marks, mark size = 1.5pt] coordinates {(1.0, 0.58)};

\end{axis}
& 
 \begin{axis}[%
    axis x line=middle,
                    axis y line=middle,
                title style={at={(axis description cs:0.1,0.965)}, anchor=south} , 
  width=0.45*\textwidth, height=3.75cm,
ymin=0.1, ymax = 1.1,
xmin=0, xmax=5.2,
xshift=-1 cm,
  xtick = \empty, 
   ytick = \empty, 
    every axis x label/.style={
    at={(ticklabel* cs:0.95)},
    anchor=north,
},
every axis y label/.style={
    at={(ticklabel* cs:0.95)},
    anchor=east,
}
]

\path[fill=Plum!15] (1,0.5) -- (2.5,0.57) -- (2.5,0.75) -- (1.0,0.65);
\path[fill=Plum!15] (2.5,0.57) -- (4.2,0.57) -- (4.2,0.75) -- (2.5,0.75);
\draw [very thick,Plum,dashed,->] plot [smooth, tension=0.4] coordinates { (1,0.58) (2.5,0.65) (4.2,  0.65) } node[Plum,right,font=\tiny]{$\sigma<1$};

\path[fill=BrickRed!15] (1,0.5) -- (2.5,0.42) -- (2.5,0.59) -- (1.0,0.65);
\path[fill=BrickRed!15] (2.5,0.59) -- (4.2,0.59) -- (4.2,0.42) -- (2.5,0.42);
\draw [thick,BrickRed,->] plot [smooth, tension=0.4] coordinates { (1,0.58) (2.5,0.5) (4.2,  0.5) } node[BrickRed,right,font=\tiny]{$\sigma>1$};

\draw [very thick,Plum,dashed,->] plot [smooth, tension=0.4] coordinates { (1,0.58) (2.5,0.65) (4.2,  0.65) } ;

\draw [thick,MidnightBlue,->] plot [smooth, tension=0.4] coordinates { (1,0.58)  (1.5,0.44) (2.5,0.35) (5,  0.35) };

\node[pin={[pin edge={BrickRed,->}, black, pin distance=0.55 cm]100:{\tiny Labor-augmenting t.c.}}] at (axis cs:3.5,0.5 ) {};
\node[pin={[pin edge={Plum,->}, black, pin distance=0.35 cm]120:{\tiny }}] at (axis cs:3.5,0.63 ) {};

\node[pin={[pin edge={MidnightBlue,->}, MidnightBlue, pin distance=0.45 cm]190:{\tiny Institutions}}] at (axis cs:2.5,0.35 ) {};

\addplot[color = gray,dotted, thin] coordinates {(0.0, 0.58) (5.0, 0.58)};

\addplot[color = black, solid, thin] coordinates {(0.0, 0.58) (1.0, 0.58)};

\addplot[color = black, dotted, thin] coordinates {(1.0, 0) (1.0, 0.58)};

\addplot[color = black, mark = o, only marks, mark size = 1.5pt] coordinates {(1.0, 0.58)};

    \end{axis}
    \\
      \begin{axis}[
        axis x line=middle,
                    axis y line=middle,
              title style={at={(axis description cs:0.1,0.965)}, anchor=south} , 
 width=0.45*\textwidth, height=3.75cm,
ymin=0.1, ymax = 1.099,
xmin=0, xmax=5.2,
xshift=-0.5cm,
    ylabel = {\footnotesize $\frac{\hat{k}}{q\hat{y}}$},
 ytick = \empty, 
  xtick = \empty, 
        every axis x label/.style={
    at={(ticklabel* cs:0.95)},
    anchor=north,
},
every axis y label/.style={
    at={(ticklabel* cs:0.95)},
    anchor=east,
}
]

\draw [very thick,BurntOrange,dashed,->] plot [smooth, tension=0.1] coordinates { (1,0.58)  (1.6,0.65) (2,0.73) (3,0.8) (5,  0.8) };
\node[pin={[pin edge={BurntOrange,->}, BurntOrange, pin distance=0.6 cm]178:{\tiny $|\frac{\partial \lambda_{t}}{\partial \dot{m}^{*}_{t}}| < \alpha$}}] at (axis cs:3.1,0.81 ) {};

\draw [thick,OliveGreen,->] plot [smooth, tension=0.4] coordinates { (1,0.58) (1.5,0.52) (2,0.6) (3,0.7) (5,  0.7) };

\node[pin={[pin edge={OliveGreen,->}, OliveGreen, pin distance=0.7cm]-17:{\tiny $|\frac{\partial \lambda_{t}}{\partial \dot{m}^{*}_{t}}| > \alpha$}}] at (axis cs:2,0.58) {};

\addplot[color = gray,dotted, thin] coordinates {(0.0, 0.58) (5.0, 0.58)};

\addplot[color = black, solid, thin] coordinates {(0.0, 0.58) (1.0, 0.58)};

\addplot[color = black, dotted, thin] coordinates {(1.0, 0) (1.0, 0.58)};

\addplot[color = black, dotted, thin] coordinates {(1.6, 0) (1.6, 0.58)};

\addplot[color = black, mark = o, only marks, mark size = 1.5pt] coordinates {(1.0, 0.58)};

\end{axis}
& 
 \begin{axis}[%
    axis x line=middle,
                    axis y line=middle,
                title style={at={(axis description cs:0.1,0.965)}, anchor=south} , 
 width=0.45*\textwidth, height=3.75cm,
ymin=0.1, ymax = 1.1,
xmin=0, xmax=5.2,
xshift=-1 cm,
  xtick = \empty, 
   ytick = \empty, 
    every axis x label/.style={
    at={(ticklabel* cs:0.95)},
    anchor=north,
},
every axis y label/.style={
    at={(ticklabel* cs:0.95)},
    anchor=east,
}
]
    
\draw [thick,BrickRed,->] plot [smooth, tension=0.4] coordinates { (1,0.58)  (1.5,0.64) (2.5,0.75) (5,  0.75) };

\draw [thick,MidnightBlue,->] plot [smooth, tension=0.4] coordinates { (1,0.58)  (1.5,0.44) (2.5,0.35) (5,  0.35) };

\node[pin={[pin edge={MidnightBlue,->}, MidnightBlue, pin distance=0.45 cm]190:{\tiny Institutions}}] at (axis cs:2.5,0.35 ) {};
\node[pin={[pin edge={BrickRed,->}, BrickRed, pin distance=0.45 cm]168:{\tiny Labor-augmenting t.c.}}] at (axis cs:4,0.77 ) {};

\addplot[color = gray,dotted, thin] coordinates {(0.0, 0.58) (5.0, 0.58)};

\addplot[color = black, solid, thin] coordinates {(0.0, 0.58) (1.0, 0.58)};

\addplot[color = black, dotted, thin] coordinates {(1.0, 0) (1.0, 0.58)};

\addplot[color = black, mark = o, only marks, mark size = 1.5pt] coordinates {(1.0, 0.58)};

    \end{axis}
    \\
      \begin{axis}[
        axis x line=middle,
                    axis y line=middle,
              title style={at={(axis description cs:0.1,0.965)}, anchor=south} , 
 width=0.45*\textwidth, height=3.75cm,
ymin=0.1, ymax = 1.099,
xmin=0, xmax=5.2,
xshift=-0.5cm,
    ylabel = {\scriptsize $\Omega$},
 ytick = \empty, 
 xlabel = {\scriptsize $t$},
  xtick = {1.0,1.6},
  xticklabels={\scriptsize $t'$,$t''$},
        every axis x label/.style={
    at={(ticklabel* cs:0.95)},
    anchor=north,
},
every axis y label/.style={
    at={(ticklabel* cs:0.95)},
    anchor=east,
}
]

\draw [thick,OliveGreen,->] plot [smooth, tension=0.5] coordinates { (1,0.65) (1,0.5)  (1.5,0.55) (1.6,0.38) (2,0.3) (3,0.23) (4,  0.23) } node[OliveGreen,right,font=\tiny]{$\sigma>1$};

\node[pin={[pin edge={OliveGreen,->}, OliveGreen, pin distance=0.5cm]180:{\tiny $|\frac{\partial \lambda_{t}}{\partial \dot{m}^{*}_{t}}| > \alpha$}}] at (axis cs:2.3,0.25) {};



\draw [very thick,BurntOrange,loosely dashdotted,->] plot [smooth, tension=0.5] coordinates { (1,0.65) (1,0.8)  (1.5,0.9) (1.6,0.4) (2,0.45) (3,0.5) (4,  0.5) } node[BurntOrange,right,font=\tiny]{$\sigma<1$};
\node[pin={[pin edge={BurntOrange,->}, BurntOrange, pin distance=0.6 cm]-1:{\tiny $|\frac{\partial \lambda_{t}}{\partial \dot{m}^{*}_{t}}| < \alpha$}}] at (axis cs:1.7,0.8 ) {};

\addplot[color = gray,dotted, thin] coordinates {(0.0, 0.65) (5.0, 0.65)};

\addplot[color = black, solid, thin] coordinates {(0.0, 0.65) (1.0, 0.65)};

\addplot[color = black, dotted, thin] coordinates {(1.0, 0) (1.0, 0.65)};

\addplot[color = black, dotted, thin] coordinates {(1.6, 0) (1.6, 0.65)};

\addplot[color = black, mark = o, only marks, mark size = 1.5pt] coordinates {(1.0, 0.65)};

\end{axis}
&
 \begin{axis}[%
    axis x line=middle,
                    axis y line=middle,
                title style={at={(axis description cs:0.1,0.965)}, anchor=south} , 
 width=0.45*\textwidth, height=3.75cm,
ymin=0.1, ymax = 1.1,
xmin=0, xmax=5.2,
xshift=-1 cm,
 ytick = \empty, 
 xlabel = {\scriptsize $t$},
  xtick = {1.0},
  xticklabels={\scriptsize $t'$},
    every axis x label/.style={
    at={(ticklabel* cs:0.95)},
    anchor=north,
},
every axis y label/.style={
    at={(ticklabel* cs:0.95)},
    anchor=east,
}
]

\draw [thick,BrickRed,->] plot [smooth, tension=0.4] coordinates { (1,0.58)  (1,0.8)  (1.5,0.74) (2.5,0.68) (4,  0.68) } node[BrickRed,right,font=\tiny]{$\sigma > 1$};

\draw [thick,MidnightBlue,->] plot [smooth, tension=0.4] coordinates { (1,0.58)  (1,0.3)  (1.5,0.35) (2.5,0.4) (4,  0.4) } node[MidnightBlue,right,font=\tiny]{$\sigma > 1$};

\node[pin={[pin edge={MidnightBlue,->}, MidnightBlue, pin distance=0.5 cm]0:{\tiny Institutions}}] at (axis cs:1.1,0.3 ) {};
\node[pin={[pin edge={BrickRed,->}, BrickRed, pin distance=0.5 cm]0:{\tiny Labor-augmenting t.c.}}] at (axis cs:1.1,0.82 ) {};

\addplot[color = gray,dotted, thin] coordinates {(0.0, 0.58) (5.0, 0.58)};

\addplot[color = black, solid, thin] coordinates {(0.0, 0.58) (1.0, 0.58)};

\addplot[color = black, dotted, thin] coordinates {(1.0, 0) (1.0, 0.58)};

\addplot[color = black, mark = o, only marks, mark size = 1.5pt] coordinates {(1.0, 0.58)};

    \end{axis}
    \\
    };
\end{tikzpicture}
            \caption{\scriptsize TRANSITIONAL DYNAMICS. \label{fig:trans_dynamics}}
            \caption*{\scriptsize \emph{Notes--- The graphs illustrate the \emph{qualitative}  changes of economic variables over time. The colored areas represent the cases where the direction of the variable cannot be determined a priori.}}
                \end{center}
            \end{figure}

At $t''$, the effects of $\dot{m}^{*}_{t}$  disappear and the economy moves to a new equilibrium with a lower $m$. Similar to \citeasnoun{acemoglu2018race}, this reduces the effective wage paid in the least complex task produced with labor and lowers the vacancy-unemployment ratio. In time, the reduction in $m$ moves the capital-output ratio upwards because, by assumption, automated tasks raise aggregate output and are immediately produced with capital. Moreover, the negative shock on wages is such that in the long-run the labor share always decreases regardless on the value of $\sigma$ and the strength of the initial capitalization effect.

The effects of a reduction in the   support to labor and of  a permanent rise in productivity growth are shown in Panel B of Figure \ref{fig:trans_dynamics}. Focusing first on the labor-augmenting technological change, we find that---thanks to the  capitalization effect---the labor share increases over time for any  $\sigma >0$. Similarly, since  higher effective wages reduce the equilibrium rate of return of capital,  higher labor productivity growth also raises the capital-output ratio. The  result on vacancies and unemployment is ambiguous and depends on the elasticity of substitution parameter. Particularly, if  $\theta$ remains more or less constant with an increase in $g$, the effects of a higher growth rate  are entirely determined by the relation of technological unemployment with $\dot{M}$. As shown in Lemma \ref{lemma:tech_unemployment}, higher growth reduces $U^{A}_{L_{t}}$ if $\sigma >1$, which explains the behavior of $U$ and $V$ depicted by the red solid lines in Figure \ref{fig:trans_dynamics}. The opposite is expected to happen when $\sigma \in (0,1)$, since in this case $\partial U^{A}_{L_{t}}/\partial \dot{M}_{t} >0$.

Finally,  lower  support to workers  moves the wage-curve in \eqref{eq:wage_agg} downwards. As a result, there is a simultaneous increase in the vacancy-unemployment ratio and in  the equilibrium rate of return of capital, which reduces the labor share and lowers the capital-output ratio over time.

\section{Empirical Analysis}\label{sec:empirics} This section evaluates some of the different channels through which technology and labor institutions have impacted the US economy. To do so,   Section \ref{subsection:long_run_empirical} applies an approximate calibration of the model  and compares the predicted paths   with their empirical counterparts.  Section \ref{subsection:cross_validation} extends the analysis and presents a cross-validation exercise that examines the  consistency of the rolling estimates of the model with  historical information of labor institutions  in the US. 

\subsection{Approximate calibration to the US economy}\label{subsection:long_run_empirical}  To get a sense of the effects of power relations and technical change in the US economy, I employ a parsimonious calibration strategy where  $T^{w}$ and $m^{*}_{t}$ are the only parameters targeting specific data. The relative mobility of workers is set to   match the efficient  unemployment rate of \citeasnoun{michaillat2021beveridgean}, which is the amount minimizing the nonproductive  labor time used in jobseeking and recruiting.\footnote{Particularly, I employ $U^{*}=\sqrt{U_{t} V_{t}}$. Online Appendix \ref{appendix:robustness} shows that similar results are obtained by employing the NAIRU as the equilibrium rate of unemployment.}  The automation measure is estimated using equation \eqref{eq:online_marginal} by  solving
 
 \begin{equation}\label{eq:empirical_automation}
 1-m^{*}_{t}  =  \frac{K_{t}}{q Y_{t}} {A^{k}q }^{1-\sigma} \Big( \delta^{\text{BEA-BLS}} (1+\mu_{t}^{\text{BEA-BLS}})\Big)^{\sigma}.
 \end{equation}

Here $q$, $\sigma$, and  $A^{k}$ are set as in Table \ref{table:calibration_steady_state1}, and $\mu_{t}^{\text{BEA-BLS}}$ and $\delta^{\text{BEA-BLS}}$ are  obtained from the  BEA-BLS integrated data; see Online Appendix \ref{appendix:data} for details. All other parameters are either calibrated to roughly describe some basic facts of the US  or are directly obtained from macro data.

 The first block of numbers in  Table \ref{table:calibration_steady_state1}  presents the time-varying values in the calibration obtained from direct data sources.  The  probability of collective bargaining is measured using union membership data from \citeasnoun{farber2021unions}, the  growth rate of average labor productivity is obtained from the Penn World Table \cite{feenstra2015next}, and the  opportunity cost of employment is calculated based on equation (20) of \citeasnoun{chodorow2016cyclicality}.

  \begin{table}\caption{Baseline calibration}
  \centering
\resizebox{\textwidth}{!}
 { \begin{tabular}{l c l l}
  \hline \hline
  Parameter & Average & Description & Target/source\\
    \hline 
Time-varying values  & & &\\
     $P(\mathcal{U}=1|\cdot)$ &0.25& Union membership: Gallup+BLS  & \citeasnoun{farber2021unions} \\
    $ g$ & 0.17$\%$  & Labor productivity growth  &   2$\%$ annual rate/\citeasnoun{feenstra2015next} \\
   $b$ &0.06& Opportunity cost of employment & \citeasnoun{chodorow2016cyclicality} \\
   $1-m^{*}$ & 0.12 & Automation measure & Equation \eqref{eq:empirical_automation}/ BEA-BLS integrated data\\
     Technology  & & &\\
      $\delta $ & 0.056$\%$ &  Depreciation rate & 7$\%$ annual rate/\citeasnoun{barkai2020declining} \\
  $\sigma $ & 0.6 &  Elasticity of substitution &  Standard calibration  \\
              $A^{k}$ & 0.022 & Capital-augmenting  technology  & $ w  \approx 1.5 \big(\delta/(q A^{k})\big)$/ \citeasnoun{moll2022uneven} \\ 
                          $\alpha$ & 1.4 & Labor-augmenting parameter &  $\Omega \approx 0.63$/ Standard calibration \\ 
                                      $q$ & 0.35 & Relative price of capital  & Annual $K/(q Y) \approx 1.5$/ BEA-BLS integrated data \\ 
                          Preferences  & & &\\
                           $\rho $ & 2.22$\%$ &  Discount rate  &    30$\%$ annual rate/  \citeasnoun[p. 3346] {andreoni2012estimating}\\
                           $\gamma^{f}$ & 0.45 &  Response time of firms &   $\Gamma^{na} \approx 0.31$/ Within  standard calibrations \\
                               Search and matching  & & &\\
          $\iota$ & 1.25 & Matching function parameter & \citeasnoun{petrosky2021unemployment}\\
    $\lambda_{0}$ & 0.02 & Separation rate & $ V \approx 3\%$ \\
        $\xi$ & 8 & Vacancy costs  & \citeasnoun{merz2007labor}\\
    \hline 
  \end{tabular}}
   \caption*{\scriptsize \emph{Notes---  All parameters are calibrated at a monthly frequency.}} 
\label{table:calibration_steady_state1}
  \end{table}

  In the second block,  I set $\delta$  close to the  average of the time-varying depreciation rate in \citeasnoun{barkai2020declining}. The elasticity of substitution parameter follows the literature and is set  at $\sigma=0.6$; Figure \ref{fig:results_G1} in  online  Appendix \ref{appendix:robustness} presents the results with $\sigma =1.2$ and shows that the conclusions are roughly equal. Consistent with \citeasnoun{moll2022uneven}, $A^{k}$  is calibrated so that  labor is about  50$\%$ more costly than capital in automated tasks.\footnote{ \citeasnoun{moll2022uneven} set labor 30$\%$ more costly than capital. The difference is explained by the fact that they include the rate of profit as a cost of production.} The  relative price of capital is fixed at 0.35 so that the equilibrium annual capital-output ratio is on average close to 1.5, which is close to the average in Figure \ref{fig:counterfactual_test}  and  Figure \ref{fig:capital_outputUS} in online Appendix \ref{appendix:data}.\footnote{Using Lemma \ref{lemma:lemma_2A_AR}, Figure \ref{fig:automation_regions_empirical} in online Appendix \ref{appendix:robustness} shows that the automation measure in Table \ref{table:calibration_steady_state1} is in Region 2 of Figure \ref{fig:automation_reg} and $q < \bar{q}(\mu^{*}_{t})$, meaning that automated tasks always raise aggregate output and are immediately produced with capital.}
  
  The monthly subjective discount rate  is consistent with the   experimental data of \citeasnoun{andreoni2012estimating}, who find  annual rates  between 0.2 and 0.4.  The matching parameter $\iota$ is set as in \citeasnoun{petrosky2018endogenous}. The job separation rate is between the estimates of \citeasnoun{shimer2005cyclical} and  \citeasnoun{hobijn2009job}, and is consistent with average vacancy rate of about 3$\%$. Lastly,  the value of $\xi$ implies that vacancy costs are   about 2 quarters of wage payments,  similar to \citeasnoun{merz2007labor}.

\subsubsection*{Results}  Figure \ref{fig:counterfactual_test} depicts the predicted paths of the labor share, capital profitability,  the capital-output ratio, and the measures of automation  along with  their empirical counterparts.\footnote{The model is solved using Julia's NLboxsolve.jl. The code is in the Supplementary Material.}   Figure \ref{fig:counterfactual_test}, Panel A shows that the  predictions of the technical change and institutions-driven stories match remarkably well different  measures of the labor share: the  technical change predicted path match the Penn World Table data, while the predictions based on  changes in labor institutions follow closely  the BEA-BLS data.   Panel B, however, demonstrates that the technical change hypothesis cannot account for the fall in the rate of return before the 1980s and its steady recovery afterwards. Similarly, it shows that the institutions hypothesis alone underestimates the fall in the rate of profit from the 1950s to the late 1970s. These results---as illustrated by the magenta lines in Figure \ref{fig:counterfactual_test}, Panel B--- suggest that an adequate understanding of the behavior of capital profitability    requires combining  the technical change and the institutions-driven stories.

                \begin{figure}
\begin{center}
\pgfplotstableread[col sep=comma,]{steady.csv}\datatable
\pgfplotstableread[col sep=comma,]{dt_mech_model2.csv}\datamechmodel
\pgfplotstableread[col sep=comma,]{dt_mech_Mann.csv}\dataMann
\pgfplotstableread[col sep=comma,]{dt_mech_Dechezlepretre.csv}\dataDechezlepretre
  \begin{tikzpicture}
  \begin{axis}[
  name=plot1,
  width=0.5*\textwidth, height=5.1 cm,
    y label style={at={(axis description cs:-0.065,.5)},rotate=90,anchor=south},
  xmin=1947,
  xmax=2010,
  ymin=0.52,
  ymax=0.7,
  legend to name=widelegend,
    legend style={font=\tiny,legend columns=3},
      axis x line*=bottom,
axis y line*=left,
  tick label style={font=\tiny},
  title={\scriptsize Panel A.  Labor share},
    title style={at={(axis description cs:0.18,1)}, anchor=south} , 
  xticklabel style={/pgf/number format/set thousands separator={}},
  xtick={1950,1960,1970,...,2010},
  ytick={0.50,0.54,...,0.76}  ]  

     \node[pin={[pin edge={gray,->}, gray, pin distance=0.45cm]0:{\tiny BEA-BLS}}] at (axis cs:1983,0.67) {};
       
             \node[pin={[pin edge={gray,->}, gray, pin distance=0.65cm]270:{\tiny Penn World Table}}] at (axis cs:1984,0.60) {};
             
   \addplot [mark=*, mark size = 1pt,mark options={fill=white}, draw=gray,line width= 0.7, smooth]  table[x index = {0}, y index = {1},
  each nth point={2}]{\datatable};
  
     \addplot [mark=*, mark size = 1pt,mark options={fill=gray}, draw=gray,line width= 0.7, smooth]  table[x index = {0}, y index = {2},
  each nth point={2}]{\datatable};

    \addplot [dashed, draw=Periwinkle ,line width= 2, smooth] table[x index = {0}, y index = {3},
  each nth point={2}]{\datatable};
  
    \addplot [dotted, draw=Green,line width= 1.5, smooth] table[x index = {0}, y index = {4},
  each nth point={2}]{\datatable};
    \addplot [draw=Magenta,line width= 1, smooth] table[x index = {0}, y index = {5},
  each nth point={2}]{\datatable};
   \legend{, , Technical change alone,Institutions alone, Technical change and institutions};
   
  \end{axis}

   \begin{axis}[%
  name=plot2,
  width=0.5*\textwidth, height=5.1 cm,
      at=(plot1.right of south east), anchor=left of south west,
    y label style={at={(axis description cs:-0.065,.5)},rotate=90,anchor=south},
  xmin=1947,
  xmax=2010,
  ymin=0.1,
  ymax=0.54,
  xshift=-0.1 cm,
        axis x line*=bottom,
axis y line*=left,
  tick label style={font=\tiny},
  legend pos=south west,
  legend style={draw=none,font=\tiny},
   legend cell align={left}, 
  title={\scriptsize Panel B. Capital profitability},
    title style={at={(axis description cs:0.27,0.975)}, anchor=south} , 
  xticklabel style={/pgf/number format/set thousands separator={}},
  xtick={1950,1960,1970,...,2010},
  ytick={0.12,0.2,...,0.6} ]

    \node[pin={[pin edge={gray,->}, gray, pin distance=0.75cm]0:{\tiny Rate of return $\mu_{t}$}}] at (axis cs:1953,0.47) {};
       
             \node[pin={[pin edge={gray,->}, gray, pin distance=0.2cm]0:{\tiny Rate of profit $r_{t} \equiv \mu_{t} Y_{t}/(P^{k}_{t} K_{t})$}}] at (axis cs:1951,0.265) {};

 \addplot  [mark=*, mark size = 1pt,mark options={fill=white}, draw=gray,line width= 0.7, smooth] table[x index = {0}, y index = {22},
  each nth point={2}]{\datatable};
  
   \addplot  [mark=*, mark size = 0.75pt,mark options={fill=gray}, draw=gray,line width= 0.7, smooth] table[x index = {0}, y index = {26},
  each nth point={2}]{\datatable};

      \addplot [ dashed, draw=Periwinkle,line width= 2, smooth] table[x index = {0}, y index = {23},
  each nth point={2}]{\datatable};
  
        \addplot [ dashed, draw=Periwinkle,line width= 2, smooth] table[x index = {0}, y index = {27},
  each nth point={2}]{\datatable};

    \addplot [dotted, draw=Green,line width= 1.5, smooth] table[x index = {0}, y index = {24},
  each nth point={2}]{\datatable};
      \addplot [dotted, draw=Green,line width= 1.5, smooth] table[x index = {0}, y index = {28},
  each nth point={2}]{\datatable};

    \addplot [ draw=Magenta,line width= 1, smooth] table[x index = {0}, y index = {25},
  each nth point={2}]{\datatable};
      \addplot [ draw= Magenta,line width= 1, smooth] table[x index = {0}, y index = {29},
  each nth point={2}]{\datatable};
 
    \end{axis}

   \begin{axis}[%
  name=plot3,
  width=0.5*\textwidth, height=5.1 cm,
    at=(plot1.below south east), anchor=above north east, 
    y label style={at={(axis description cs:-0.065,.5)},rotate=90,anchor=south},
  xmin=1947,
  xmax=2010,
  ymin=0.8,
  ymax=2,
  xshift=-0.1 cm,
        axis x line*=bottom,
axis y line*=left,
  tick label style={font=\tiny},
  legend pos=south west,
  legend style={draw=none,font=\tiny},
   legend cell align={left}, 
  title={\scriptsize Panel C. Capital-output ratio},
    title style={at={(axis description cs:0.28,0.975)}, anchor=south} , 
  xticklabel style={/pgf/number format/set thousands separator={}},
  xtick={1950,1960,1970,...,2010},
  ytick={0.4,0.6,...,2.2} ]  

 \addplot [mark=*, mark size = 1pt,mark options={fill=white}, draw=gray,line width= 1, smooth] table[x index = {0}, y index = {6},
  each nth point={2}]{\datatable};

      \addplot [dashed, draw=Periwinkle,line width= 2, smooth] table[x index = {0}, y index = {7},
  each nth point={2}]{\datatable};

    \addplot [dotted, draw=Green,line width= 1.5, smooth] table[x index = {0}, y index = {8},
  each nth point={2}]{\datatable};

    \addplot [ draw=Magenta,line width= 1, smooth] table[x index = {0}, y index = {9},
  each nth point={2}]{\datatable};

    \end{axis}

   \begin{axis}[%
  name=plot4,
  width=0.5*\textwidth, height=5.1 cm,
       at=(plot3.right of north east), anchor=left of north west,
    y label style={at={(axis description cs:-0.065,.5)},rotate=90,anchor=south},
  xmin=1947,
  xmax=2010,
ymin=0.08,
  ymax=0.18,
        axis x line*=bottom,
axis y line*=left,
  xshift=-0cm,
  tick label style={font=\tiny},
  legend pos=south west,
  legend style={draw=none,font=\tiny},
   legend cell align={left}, 
  title={\scriptsize Panel D. Automation measures},
    title style={at={(axis description cs:0.28,0.975)}, anchor=south} , 
  xticklabel style={/pgf/number format/set thousands separator={}},
  y tick label style={/pgf/number format/fixed,
      /pgf/number format/1000 sep = \thinspace},
  xtick={1950,1960,1970,...,2010},
  ytick={0.08,0.1,...,0.2} ]

                   \node[pin={[pin edge={gray,->}, gray, pin distance=0.75cm]120:{\tiny \citeasnoun{mann2021benign}}}] at (axis cs:1995,0.125) {};
       
       \node[pin={[pin edge={gray,->}, gray, pin distance=0.3cm]271:{{\fontsize{5.5}{4}\selectfont \citeasnoun{Dechezlepretre2019}}}}] at (axis cs:1970,0.115) {};

   \addplot [dashed, draw=Periwinkle,line width= 2, smooth] table[x index = {0}, y index = {1},
  each nth point={1}]{\datamechmodel};
  
     \addplot[dotted, draw=Green,line width= 1.5, smooth] table[x index = {0}, y index = {2},
  each nth point={1}]{\datamechmodel};

   \addplot [ draw=Magenta,line width= 1, smooth]  table[x index = {0}, y index = {1},
  each nth point={2}]{\datamechmodel};
   \addplot [mark=*, mark size = 1.25pt,mark options={fill=white!20}, draw=gray,line width= 1, smooth] table[x index = {0}, y index = {1},
  each nth point={2}]{\dataMann};
    \addplot [mark=triangle*, mark size = 1.5pt,mark options={fill=gray!10}, draw=gray,line width= 1, smooth] table[x index = {0}, y index = {1},
  each nth point={2}]{\dataDechezlepretre};
    \end{axis}
    
    \end{tikzpicture}
    \ref{widelegend}
  \end{center}
  \caption{\scriptsize  EQUILIBRIUM PATHS. \label{fig:counterfactual_test}} 
  \caption*{ \scriptsize  \emph{Notes---  Real data is represented with gray lines.   Panels A, B, C and D use    BLS-BEA  data; see Online Appendix \ref{appendix:data}. The purple dashed lines represent the paths where only  $g_{t}$ and $m_{t}$  vary in time. The green dotted lines represent the paths where only $T^{w}$, $b_{t}$ and $P(\mathcal{U}_{t}=1|\cdot)$ vary in time. The magenta  lines allow all the previous parameters to change in time. In panel D, the initial value of the data is normalized with respect to the current value of $1-m^{*}_{t}$.}}
  \end{figure}
  
The data of the capital-output ratio in  Figure \ref{fig:counterfactual_test}, Panel C,  is matched completely   by introducing changes in the automation of tasks, but is inconsistent with the predictions of  the   institutions-driven hypothesis.  This conclusion is supported in Figure \ref{fig:counterfactual_test}, Panel D,  by noting that the  estimated value of the automation measure based on  \eqref{eq:empirical_automation} is well aligned with the   time series of the automation share constructed by  \citeasnoun{Dechezlepretre2019}  and  \citeasnoun{mann2021benign} using US patent data.

                \begin{figure}
\begin{center}
\pgfplotstableread[col sep=comma,]{steady.csv}\datatable
  \begin{tikzpicture}
  \begin{axis}[
  name=plot1,
  width=0.5*\textwidth, height=5.5 cm,
    y label style={at={(axis description cs:-0.065,.5)},rotate=90,anchor=south},
  xmin=1947,
  xmax=2010,
  ymin=0,
  ymax=10,
  legend to name=widelegend,
    legend style={font=\tiny,fill=white,draw=black, legend columns=3},
      axis x line*=bottom,
axis y line*=left,
  tick label style={font=\tiny},
  title={\scriptsize Panel A.  Unemployment and vacancies},
    title style={at={(axis description cs:0.35,1)}, anchor=south} , 
  xticklabel style={/pgf/number format/set thousands separator={}},
  xtick={1950,1960,1970,...,2010},
  ytick={0,2,...,12},
 yticklabel={$\pgfmathprintnumber{\tick}\%$} ]  

       \node[pin={[pin edge={gray,->}, gray, pin distance=0.75cm]180:{\tiny Unemployment}}] at (axis cs:1985,9) {};
       
             \node[pin={[pin edge={gray,->}, gray, pin distance=0.3cm]270:{\tiny Vacancy}}] at (axis cs:1991,2.5) {};
             
                 \node[pin={[pin edge={Magenta,->}, Magenta, pin distance=1.1cm]90:{\tiny $U^{*}=\sqrt{U_{t} V_{t}}$}}] at (axis cs:1994,4) {};

   \addplot [mark=*, mark size = 1pt,mark options={fill=white}, draw=gray,line width= 0.7, smooth]  table[x index = {0}, y index = {11},
  each nth point={1}]{\datatable};
  
   \addplot [mark=*, mark size = 0.7 pt,mark options={fill=gray}, draw=gray,line width= 0.7, smooth]  table[x index = {0}, y index = {14},
  each nth point={1}]{\datatable};

    \addplot [loosely dashed, draw=Periwinkle,line width= 2, smooth] table[x index = {0}, y index = {13},
  each nth point={2}]{\datatable};
    \addplot [dotted, draw=Green,line width= 2, smooth] table[x index = {0}, y index = {12},
  each nth point={1}]{\datatable};
    \addplot [ draw=Magenta,line width= 1, smooth]  table[x index = {0}, y index = {12},
  each nth point={2}]{\datatable};
  
      \addplot [loosely dashed, draw=Periwinkle,line width= 2, smooth] table[x index = {0}, y index = {15},
  each nth point={2}]{\datatable};
    \addplot [dotted, draw=Green,line width= 2, smooth] table[x index = {0}, y index = {16},
  each nth point={2}]{\datatable};
    \addplot [ draw=Magenta,line width= 1, smooth]  table[x index = {0}, y index = {17},
  each nth point={2}]{\datatable};
  
   \legend{,,Technical change alone,Institutions alone, Technical change and institutions, , , };
   
  \end{axis}

   \begin{axis}[%
  name=plot2,
  width=0.5*\textwidth, height=5.5 cm,
      at=(plot1.right of south east), anchor=left of south west,
    y label style={at={(axis description cs:-0.065,.5)},rotate=90,anchor=south},
  xmin=1947,
  xmax=2010,
  ymin=0.1,
  ymax=1.6,
  xshift=-0.1 cm,
        axis x line*=bottom,
axis y line*=left,
  tick label style={font=\tiny},
  legend pos=south west,
  legend style={draw=none,font=\tiny},
   legend cell align={left}, 
  title={\scriptsize Panel B. Labor market tightness},
    title style={at={(axis description cs:0.3,0.975)}, anchor=south} , 
  xticklabel style={/pgf/number format/set thousands separator={}},
  xtick={1950,1960,1970,...,2010},
  ytick={0,0.2,...,2} ]  

 \addplot [mark=*, mark size = 1pt,mark options={fill=white}, draw=gray,line width= 0.7, smooth] table[x index = {0}, y index = {18},
  each nth point={1}]{\datatable};
  
     \addplot [loosely dashed, draw=Periwinkle,line width= 2, smooth] table[x index = {0}, y index = {19},
  each nth point={2}]{\datatable};
    \addplot [dotted, draw=Green,line width= 2, smooth] table[x index = {0}, y index = {20},
  each nth point={2}]{\datatable};
    \addplot [ draw=Magenta,line width= 1, smooth]  table[x index = {0}, y index = {21},
  each nth point={2}]{\datatable};
    \end{axis}
    
    \end{tikzpicture}
    \ref{widelegend}
  \end{center}
  \caption{\scriptsize LABOR MARKET EQUILIBRIUM PATHS \label{fig:counterfactual_test2}}
\caption*{ \scriptsize  \emph{Notes--- The data is from \protect \citeasnoun{petrosky2021unemployment}.}}
  \end{figure}

  Figure \ref{fig:counterfactual_test2}, in turn,  reveals  that the variations in the labor market  cannot be matched by changes in the rate of automation or the rate of productivity growth. By contrast,  the  predicted paths associated with changes in labor institutions are perfectly   consistent with the behavior of the efficient unemployment rate (by construction), and with the time series  of the vacancy rate and  labor market tightness.

In sum,  the calibration exercise shows that, while it is unlikely that the trends in the US  economy can all be adequately explained by relying on one hypothesis alone,  the fluctuations in worker power induced by variations in labor institutions are probably the major structural changes  given their capacity to explain the behavior of the labor share, the rate of return of capital, and the dynamics of the labor market  throughout the postwar period. This conclusion finds additional support  in the following subsection by showing  that the predicted paths of worker power are consistent with  the behavior of important  labor institutions  in the US.

  \subsection{Worker Power and Labor Institutions}\label{subsection:cross_validation}    Figure \ref{fig:inverse_inference} compares the  inferred time series of $T^{w}_{t}$  with   popular measures of  the institutional support to labor.\footnote{Figures \ref{fig:results_NAIRU} and \ref{fig:results_G1} in online Appendix \ref{appendix:robustness} show the predicted paths of $T^{w}$ are  robust to alternative model calibrations.}  The rolling estimates of the relative mobility of labor indicate that during the period of the New Deal Order capitalists probably lost power over labor given the increasing difficulty of finding new workers willing to accept lower wages.\footnote{Similar to \citeasnoun{gerstle2022}, a political order is defined as  a time period where governments are compelled to follow a specific set of policy rules irrespective of their own ideological affiliation.}  The rise of the federal real minimum wage and the high levels of union membership over this period are some of the institutional changes which support this hypothesis,  given that---by legislative and political action---they helped mitigate the capacity of firms to lower wages through the force of competition.

               \begin{figure}
\begin{center}
\pgfplotstableread[col sep=comma,]{steady.csv}\datatable
\pgfplotstableread[col sep=comma,]{dt_communist_th.csv}\datacom

  \begin{tikzpicture}

  \begin{axis}[
           axis y line*=left,
scale only axis,
  width=0.87*\textwidth, height=7 cm,
  xmin=1947,
  xmax=2010,
  ymin=0,
    ymax=50,
  xticklabel style={/pgf/number format/set thousands separator={}},
  xtick={1950,1955,1960,...,2015},
  ytick={0,10,...,60},
ylabel={Institutional support to labor},
 yticklabel={$\pgfmathprintnumber{\tick}\%$},
y label style={rotate=-0, at={(ticklabel cs:0.25)}, font=\tiny} ]

\path[fill=lightgray!20] (1949,2) -- (1949,49) -- (1977,49) -- (1977,2);

       \addplot[mark=square*, mark size = 2pt,mark options={solid,fill=BrickRed},  draw=BrickRed!80,line width=1, smooth] table[x index = {0}, y index = {1},  each nth point={2}]{\datacom}node[BrickRed,above,font=\tiny, xshift=2.5em]{Communist threat};
       
       \node[pin={[pin edge={gray,->}, gray, pin distance=0.75cm]90:{\tiny Federal real minimum wage}}] at (axis cs:2001,26) {};

       \node[pin={[pin edge={gray,->}, gray, pin distance=0.4cm]180:{\tiny Top marginal 
       income tax rate}}] at (axis cs:1993,18) {};
       
        \node[pin={[pin edge={lightgray,->}, gray, pin distance=2.6cm]90:{\tiny Union membership}}] at (axis cs:2003,21) {};
       
         \node[pin={[pin edge={white,->}, black, pin distance=0cm]-90:{\footnotesize Neoliberal Order}}] at (axis cs:1991,50) {};
         
                   \node[pin={[pin edge={white,->}, black, pin distance=0cm]-90:{\footnotesize New Deal Order}}] at (axis cs:1962,50) {}; 
       
                  \addplot [mark=square* , mark size = 1pt,mark options={solid,fill=lightgray},  draw=lightgray,line width=1, smooth] table[x index = {0}, y index = {30}, each nth point={1}]{\datatable};
                     \addplot [mark=diamond* , mark size = 2pt,mark options={solid,fill=lightgray},  draw=lightgray,line width=1, smooth] table[x index = {0}, y index = {31}, each nth point={1}]{\datatable};
                           \addplot [mark=* , mark size = 2pt,mark options={solid,fill=white},  draw=lightgray,line width=1, smooth] table[x index = {0}, y index = {32}]{\datatable};

\end{axis}

\begin{axis}[
scale only axis,
  axis y line*=right,
axis x line=none,
  width=0.87*\textwidth, height=7 cm,
  xmin=1947,
  xmax=2010,
  ymin=-0.1,
    ymax=8.1,
ylabel style = {align=center},
legend columns = 3,
    legend style = {at={(0.5, -0.13)}, anchor=north, inner sep=1pt, style={column sep=0.15cm}},       
    legend cell align=left,
  xtick={},
  ytick={0,1,...,9},
  ylabel={$T^{w}$},
y label style={rotate=-90, xshift=-0.35em, font=\tiny} ]

    



        
    

        \addplot [loosely dashed, draw=Periwinkle,line width=2, smooth]table[x index = {0}, y index = {33}]{\datatable};

     \addplot [dotted, draw=Green,line width=1.5, smooth] table[x index = {0}, y index = {34}]{\datatable};

                        \addplot [draw= Magenta,line width=1, smooth] table[x index = {0}, y index = {35}]{\datatable}; 
                        



        \legend{\tiny $T^{w}$: Technical change,\tiny $T^{w}$: Institutions alone, \tiny $T^{w}$: Technical change and  institutions,,,};
        
\end{axis}

    \end{tikzpicture}
  \end{center}
  \caption{ \scriptsize WORKER POWER AND LABOR INSTITUTIONS. \label{fig:inverse_inference}}
  \caption*{\scriptsize \emph{Notes---The top marginal income tax rate, the federal real minimum wage, and union membership are normalized so that in 1980 they are equal to the value of the Communist threat. The top marginal income tax rate is obtained from  \citeasnoun{alvaredo2018world} and the Communist threat is  the relative real GDP per capita between the Soviet Union and the US  obtained from   \citeasnoun{bolt2020maddison}. }}
  \end{figure}

  By the mid-1970s,   the  political order supporting  labor lost momentum  and the US found itself in a new era with  declining  real minimum wages, lower union memberships, and falling top marginal income tax rates.\footnote{Given that production workers are not normally in the top of the income distribution,  lower marginal tax rates probably shift the bargaining scale against labor and in favor of capital. \citeasnoun{piketty2014optimal} reach a  similar conclusion in relation to  the relative  pay of CEOs.} These institutional changes coincide with the  fall of the relative mobility of labor, which can account for  the decline of the labor share in the mid-1970s, the fall of the equilibrium rate of unemployment, and the  steady (or even rising) vacancy rates over the 1970s and 1980s. 
  
  But, what explains the rise and fall of the institutional support to labor? And why is  worker power partly  captured by $T^{w}$?   The answer to the second question is  that $T^{w}$  defines the  probability that firms will match with a new worker in the bargaining process of wages; see Proposition \ref{prop:non_union_wage} above. Thus,  as $T^{w}$ gets bigger, firms gain a hiring advantage  by  increasing  the competition among workers for each available vacancy. In this respect, it is reasonable that $T^{w}$ will decrease with institutional changes like higher real minimum wages or   higher union memberships given that these  restrict the capacity of firms to lower wages through competition.

A tentative  answer to the  first question is found in Figure \ref{fig:inverse_inference} by following Gerstle's \citeyear{gerstle2022} argument that  much of the changes in the institutional support to labor can  be attributed to  the Communist threat---which refers to the class compromise between capital and labor induced by the fear that communism could challenge capitalism as the dominant economic system. By this logic, it was in the interest of capitalists and the government to compromise by enhancing social programs for the poor, putting forward legislative actions favoring a bigger welfare state, and addressing the international embarrassment of white supremacy in the southern states.\footnote{Paradoxically, the need to demonstrate  a commitment to dismantling segregation  in the US was accompanied by a decreasing support among radically conservative whites to the  Democratic Party and the legislation favoring the construction of a welfare state \cite{kuziemko2018did}. Notwithstanding these issues related to racism, Appendix \ref{appendix:protocl_decision} presents a  formalism of  Figure \ref{fig:inverse_inference} using   a simple game-theoretic model linking the Communist threat with the institutional support to labor.}  In the mid-1970s, however,   the political pressure to comply with the requirements of a strong welfare state vanished with the decay of the Soviet Union's economy, as illustrated by the simultaneous decline of the Communist threat and the  institutional support  to labor in the US.\footnote{This analysis can benefit by using  microdata that reveals how the public's perception of communism affected their support to a welfare state, and how this perception  impacted public policies.  This is flagged as a fruitful area for future research.}

\section{Further Economic Implications}\label{sec:extensions} 

Extending on the results of the previous section, next I show some connections of  the  worker power hypothesis with  the wage-premium and the association of  market power with increasing markups.

\subsection{Institutions, profitability,  and the wage-premium} Figure \ref{fig:inst_prof_wagep} depicts two important findings which highlight   the  predictive capacity of the worker power hypothesis. The first is that the equilibrium path of the  rate of return of capital obtained by allowing changes in labor institutions matches remarkably well the   behavior of the average  markup in the US. Particularly, both the model and the data show a declining trend in business profitability between the  1960s and the late  1970s, and a steady recovery since the early 1980s (see Figures   \ref{fig:profitability_US}  and \ref{fig:profitability_wp} in Online Appendix \ref{appendix:data} for additional evidence).  This contrasts with the predicted path obtained by only allowing changes in technology, where the model and the data move in polar directions. Thus,  Figures \ref{fig:inverse_inference} and \ref{fig:inst_prof_wagep} present clear evidence---based on solid theoretical foundations---showing that there is a redistribution of the production surplus from labor to capital with a weakening   power of workers.

               \begin{figure}
\begin{center}
\pgfplotstableread[col sep=comma,]{dt_return.csv}\datareturn
\pgfplotstableread[col sep=comma,]{dt_wagepremium.csv}\dtwp
\pgfplotstableread[col sep=comma,]{dt_deloecker.csv}\dtloecker

  \begin{tikzpicture}

  \begin{axis}[
           axis y line*=left,
scale only axis,
  width=0.85*\textwidth, height=6 cm,
  xmin=1947,
  xmax=2010,
  ymin=0.34,
    ymax=0.75,
  xticklabel style={/pgf/number format/set thousands separator={}},
  xtick={1950,1955,1960,...,2005},
   ytick={0.2,0.25,...,1},
   legend columns = 1,
    legend style = {draw=none, at={(0.34, 0.93)}, anchor=north, inner sep=1pt, style={column sep=0.15cm}},   
    legend cell align={left},    
ylabel={Markup and wage-premium},
y label style={rotate=-0, at={(ticklabel cs:0.25)}, font=\tiny} ]

                         \addplot [mark=diamond* , mark size = 2pt,mark options={solid,fill=white},  draw=lightgray,line width=1, smooth] table[x index = {0}, y index = {1}, each nth point={2}]{\dtwp};

                        \addplot [mark=* , mark size = 2pt,mark options={solid,fill=lightgray!20},  draw=black,line width=0.75, smooth] table[x index = {0}, y index = {5}, each nth point={2}]{\datareturn};
                        
                              \addplot [mark=square* , mark size = 2pt,mark options={solid,fill=white},  draw=black,line width=0.75, smooth] table[x index = {0}, y index = {2}, each nth point={2}]{\dtloecker};

                           \legend{\tiny Wage-premium (CPS March), \tiny  Wage-premium (BEA-BLS), \tiny Average  markup \cite{de2020rise}};
             
\end{axis}

\begin{axis}[
scale only axis,
  axis y line*=right,
axis x line=none,
  width=0.85*\textwidth, height=6 cm,
  xmin=1947,
  xmax=2010,
  ymin=0.25,
    ymax=0.51,
ylabel style = {align=center},
legend columns = 3,
    legend style = {at={(0.5, -0.13)}, anchor=north, inner sep=1pt, style={column sep=0.15cm}},       
    legend cell align=left,
 ytick={0.1,0.13,...,0.7},
    legend style = {at={(0.5, -0.13)}, anchor=north, inner sep=1pt, style={column sep=0.15cm}},       
    legend cell align=left,
  xtick={},
  ylabel={$\mu$},
y label style={rotate=-90, xshift=-0.5em, font=\scriptsize} ]

             
      \addplot [dashed, draw=Periwinkle,line width= 2, smooth] table[x index = {0}, y index = {4},
  each nth point={2}]{\datareturn};
    \addplot [dotted, draw=Green,line width= 1.5, smooth] table[x index = {0}, y index = {2},
  each nth point={2}]{\datareturn};
    \addplot [draw=Magenta,line width= 1, smooth] table[x index = {0}, y index = {3},
  each nth point={2}]{\datareturn};
          \legend{\tiny  $\mu$: Technical change alone  ,\tiny  $\mu$: Institutions alone, \tiny $\mu$: Technical change and  institutions};
       
\end{axis}

    \end{tikzpicture}
  \end{center}
  \caption{\scriptsize  PROFITABILITY AND WAGE-PREMIUM.\label{fig:inst_prof_wagep}}
  \caption*{ \scriptsize \emph{Notes--- The average  markup and the wage-premium (BEA-BLS) are normalized so that in 1981 they are equal to the value of the wage-premium (CPS March) of \protect \citeasnoun{autor2008trends}.}}
  \end{figure}

The second important finding in Figure \ref{fig:inst_prof_wagep} is that the  wage-premium is positively correlated  with business profitability.  A possible interpretation of this result   is that the growing surplus going from labor to  capital can filtrate to different types of workers  depending on the role they play in the production process.  For example, managers and executives---who are high up in the scale of skilled workers \cite[p. 18]{autor2015there}---   have profited from the decline of union membership by removing the influence of production workers on executive pay \cite{dinardo2000unions, rosenfeld2006widening}. Additionally, they  have probably  benefited from lower  minimum  wages and declining top marginal income tax rates  given that part of their pay is directly tied to bonuses and stock options---both of which are not necessarily influenced by their own performance,  but are rather determined  by external circumstances related to the profitability of businesses \cite{piketty2003income,piketty2014optimal, acemoglu2022eclipse}.\footnote{It goes without saying that the association of worker power with the wage-premium does not rule out the possibility that skills and education play an important part in the determination of wages. However,  given that the demand for high-skilled labor is probably associated with business profitability, it is unlikely that the bias in skilled labor  is  an exogenous factor causing the sharp increase in the wage-premium.}

\subsection{Concentration and markups} Many works attempting to explain  the fall in the labor share since the 1980s  base their  analyzes  on the principle that large firms  can pay workers below their marginal productivity, such that (in the text's notation):

               \begin{figure}
\begin{center}
\pgfplotstableread[col sep=comma,]{dt_concentration.csv}\dataconcent
\pgfplotstableread[col sep=comma,]{dt_return.csv}\dtreturn
\pgfplotstableread[col sep=comma,]{dt_deloecker.csv}\dtloecker
\pgfplotstableread[col sep=comma,]{steady.csv}\datatable
  \begin{tikzpicture}

  \begin{axis}[
           axis y line*=left,
scale only axis,
  width=0.85*\textwidth, height=6 cm,
  xmin=1945,
  xmax=2010,
  ymin=0.09,
    ymax=0.6,
  xticklabel style={/pgf/number format/set thousands separator={}},
  legend columns = 1,
    legend style = {draw=none, at={(0.32, 0.93)}, anchor=north, inner sep=1pt, style={column sep=0.15cm}},   
    legend cell align={left},    
  xtick={1950,1955,1960,...,2015},
  ytick={0,0.1,...,1.1},
ylabel={Rate of return of capital},
y label style={rotate=-0, at={(ticklabel cs:0.25)}, font=\tiny} ]

\path[fill=lightgray!20] (1982,0.11) -- (1982,0.58) -- (2008,0.58) -- (2008,0.11);

                  \addplot [mark=square* , mark size = 2pt,mark options={solid,fill=lightgray!20},  draw=gray,line width=1, smooth] table[x index = {0}, y index = {1}, each nth point={2}]{\dtloecker};
             
                                    \addplot [mark=* , mark size = 2pt,mark options={solid,fill=white},  draw=lightgray,line width=1, smooth] table[x index = {0}, y index = {1}, each nth point={2}]{\dtreturn};
                                    
                                       \addplot  [mark=diamond* , mark size = 2pt,mark options={solid,fill=lightgray!20}, draw=black, line width=1, smooth] table[x index = {0}, y index = {26},
  each nth point={2}]{\datatable};
  
                                       \legend{ \tiny Average net markup, \tiny Rate of return (BEA-BLS), \tiny Rate of profit (BEA-BLS)};    
             
\end{axis}

\begin{axis}[
scale only axis,
  axis y line*=right,
axis x line=none,
  width=0.86*\textwidth, height=6 cm,
  xmin=1945,
  xmax=2010,
  ymin=0.35,
    ymax=1,
ylabel style = {align=center},
legend columns = 2,
    legend style = {at={(0.5, -0.13)}, anchor=north, inner sep=1pt, style={column sep=0.15cm}},       
    legend cell align=left,
  xtick={},
  ytick={0,0.2,...,1.2},
  ylabel={Market Concentration},
y label style={rotate=-180,at={(1.05,0.7)}, font=\tiny} ]

         \addplot [mark=* , mark size = 2pt,mark options={solid,fill=white},  draw=black,line width=1, smooth] table[x index = {0}, y index = {1}, each nth point={2}]{\dataconcent};

\end{axis}

    \end{tikzpicture}
  \end{center}
  \caption{\scriptsize PROFITABILITY AND MARKET CONCENTRATION.\label{fig:concentration_hyp}}
   \caption*{\scriptsize  \emph{Notes--- The average net markup is the natural logarithm of the markup in \protect \citeasnoun{de2020rise}.  The black line with white circles is the share of corporate  assets accounted for by the top 0.1$\%$ \protect \cite{kwon2023}.}}
  \end{figure}

\begin{equation}\label{eq:markup_power}
y_{L_{t}} = w_{t} (1+\mu_{t})
\end{equation}

The problem, as previously noted by \citeasnoun{stansbury2020declining}, is that there is essentially no way to distinguish between the rise in $\mu$ as a result of an increasing concentration in markets or  a fall in worker power using  equation \eqref{eq:markup_power} alone. 

Figure \ref{fig:concentration_hyp} helps solve this identification problem by directly comparing different measures of capital profitability  with the concentration of markets on large firms;  Figure \ref{fig:profitability_wp} in online Appendix \ref{appendix:data} presents additional evidence.  The key takeaway   is that the association between market concentration and higher markups is only clear after 1982, which is the   period commonly studied  in the  papers  defending the market power hypothesis \cite[e.g.,][]{autor2020fall,barkai2020declining}. Between the 1950s and the late 1970s,  by contrast, market concentration and  business profitability move in polar directions,  while---as shown in Figure \ref{fig:inst_prof_wagep}---the latter is always   consistent with the behavior of worker power induced by the institutional changes in the US.

\section{Conclusions}\label{sec:conclusions}

The article has proposed a novel approach showing how  politico-economic variables can intervene in macroeconomic outcomes by directly affecting the power of  labor. In this environment, labor institutions define the ``playing field'' in the bargaining process of wages, which is instrumental for determining the equilibrium rate of unemployment and the rate of return of capital. Moreover, the  surplus realized by capitalists in the bargaining process is central in the model by establishing the funds for a continuous reproduction of the   economy at an increasing scale, and by  defining the  regions for which it is profitable for firms to substitute capital for labor.

Empirically, the model offers a plausible explanation for the long-run behavior of the labor share, capital profitability, the capital-output ratio, the rate of unemployment, and the vacancy rate,  based on a combination of institutional and technological changes over the postwar period.  In addition, the analysis helps narrow down the multidimensionality of institution-driven stories of the fall in the labor share over the past half-century to specific policy changes which include---but are not necessarily limited to---the variations in union membership, unemployment benefits,    real minimum wages, and geopolitical threats.  In this respect, the model opens up the traditional  framework  by showing how the political economy of income distribution, labor institutions, and  political preferences is not a mere complement to, but rather a vital part of,   macroeconomic  analysis. 



\appendix

\numberwithin{equation}{section}
\numberwithin{lemma}{section}
\numberwithin{assm}{section}
\numberwithin{defn}{section}
\renewcommand{\theequation}{\thesection\arabic{equation}}
\renewcommand{\thelemma}{\thesection\arabic{lemma}}
\renewcommand{\thedefn}{\thesection\arabic{defn}}
\renewcommand{\theassm}{\thesection\arabic{assm}}
\renewcommand{\thesubsection}{\thesection.\arabic{subsection}}

\section{Main Appendix}\label{appendix:AppendixA}

\subsection{Model with investment-specific technological change}\label{subappendix:genmodelsec1}
The analysis in the text was carried out under  Assumption \ref{ass:task_function} and the principle that $q_{t}=q$ for any $t$. This section introduces a generalization of the model in the text by replacing    Assumption \ref{ass:task_function} for

  \begin{assm}\label{ass:management} $A^{k}_{t}=A^{k} D(h_{t})^{-a_{0}}$ and $A^{l}_{t}(j)=e^{\alpha j} D(h_{t})^{a_{1}}$, with $D'(h_{t})>0$ and $a_{0}, a_{1} >0$.
 \end{assm}

  Assumption \ref{ass:management} follows \citeasnoun{grossman2017balanced} by positing a relation between the management effort of firms---here denoted as $h_{t}$---and the   disembodied technology functions. Intuitively, the assumption says that  firms can raise the productivity of labor at the expense of increasing the relative supply of effective capital, which  tilts the unit isoquants and leads to a technological change which is both labor saving and capital using.

 Using Assumption  \ref{ass:management}  and following similar steps as those outlined in \citeasnoun{acemoglu2018race}, the aggregate output of the economy can be written as

\begin{equation*}
Y_{t} = \Bigg[(1-m^{*}_{t})^{1/\sigma}\big(K_{t} A^{k}D(h_{t})^{-a_{0}}\big)^{\frac{\sigma-1}{\sigma}} + \Big(\frac{e^{c(\sigma-1)m^{*}_{t}} -1}{\alpha (\sigma-1)}\Big)^{1/\sigma} \Big(A^{l}_{t}(J^{*}) L_{t}\Big)^{\frac{\sigma-1}{\sigma}}\Bigg]^{\frac{\sigma}{\sigma-1}}.
\end{equation*}

Given the ideal price index condition, the  partial derivatives of $Y_{t}$ with respect to $K_{t}$ and $L_{t}$ satisfy

\begin{equation}\label{eq:online_marginal}
\begin{split}
Y_{K_{t}} & = \Bigg(\frac{Y_{t}}{K_{t}}\Bigg)^{1/\sigma} (1-m^{*}_{t})^{1/\sigma}{A^{h}}^{\frac{\sigma-1}{\sigma}}D(h_{t})^{-a_{0}\frac{(\sigma-1)}{\sigma}}=\frac{\delta P^{k}_{t}}{P^{c}_{t}}\\
Y_{L_{t}} & =  \Bigg(\frac{Y_{t}}{e^{\alpha J^{*}_{t}} D(h_{t})^{a_{1}} L_{t}}\Bigg)^{1/\sigma} \Big(\frac{e^{\alpha (\sigma-1)m^{*}_{t}} -1}{\alpha (\sigma-1)}\Big)^{1/\sigma}  e^{c\alpha ^{*}_{t}} D(h_{t})^{a_{1}}  = \frac{W_{t}}{P^{c}_{t}}
\end{split}
\end{equation}

 To further simplify the analysis,  let $Y_{t}$ be expressed as

\begin{equation*}
Y_{t} = e^{\alpha J^{*}_{t}}D(h_{t})^{a_{1}} L_{t} \Bigg[(1-m^{*}_{t})^{1/\sigma} Z_{t} + \Bigg(\frac{e^{\alpha (\sigma-1)m^{*}_{t}} -1}{\alpha (\sigma-1)}\Bigg)^{1/\sigma} \Bigg]^{\frac{\sigma}{\sigma-1}}.
\end{equation*}

Where $Z_{t} = \Big( (K_{t}/L_{t}) D(h_{t})^{-(a_{0}+a_{1})} e^{-\alpha J^{*}_{t}}\Big)^{\frac{\sigma-1}{\sigma}}$. Denoting $A= a_{1}/(a_{0}+a_{1})$, the aggregate production function  can be expressed as 

\begin{equation*}
Y_{t} = \Big(L_{t} e^{\alpha  J^{*}_{t}}\Big)^{1-A} K^{A}_{t} Z^{\frac{-A \sigma}{\sigma-1}}_{t} \Bigg[(1-m^{*}_{t})^{1/\sigma} Z_{t} + \Bigg(\frac{e^{\alpha (\sigma-1)m^{*}_{t}} -1}{\alpha(\sigma-1)}\Bigg)^{1/\sigma} \Bigg]^{\frac{\sigma}{\sigma-1}}
\end{equation*}

which is a  Cobb-Douglas function  with possible shifts in the factor share parameters.  The next lemma presents a generalization of Lemma \ref{lemma:management_text} in the text.

   \begin{lemma}\label{lemma:management} Suppose that Assumption \ref{ass:management} holds. If firms choose the management effort to maximize output, then in any BGP:  
 
 \begin{itemize}
 \item $g_{K} = g_{Y} +g_{q}$.
 \item $g_{Y} =g_{C}=g=  \alpha  \dot{M} + a_{1}\; g_{q}/a_{0}$.
 \item $\frac{D'(h)}{D(h)}\dot{h} = g_{q}/a_{0}$.
 \end{itemize}

 \end{lemma} 
 
 The proof of Lemma \ref{lemma:management} is shown in Online Appendix \ref{appendix:proofs}. For now, the main argument is that the model in the text can be easily generalized to incorporate investment-specific technological change.

\section{Auxiliary results} 

 \subsection{Decision over bargaining strategies}\label{appendix:protocl_decision}  Here  I propose a game-theoretic model determining the probability that workers will choose a collective bargaining strategy  in Figure \ref{fig:bargaining_protocol}.  The multidimensionality in the preferences  of workers under collective bargaining is expressed as:

\begin{equation*}
U^{i,1}_{W}=\omega_{i0} + \omega_{1}L^{u} + \omega_{2} w^{u} - \omega^{3}(\mathcal{R}-\bar{\mathcal{R}}_{i} )^{2} - \omega_{4} (\mathcal{Q}-\bar{\mathcal{Q}})^{2}
\end{equation*}

 with   $\omega_{j} \geq 0$  for $j \in \{ i0,1,2,3,4\}$, $\omega_{10}>\omega_{20}$ and $\mathcal{\bar{R}}_{1} <\mathcal{\bar{R}}_{2}$. The first term  $\omega_{0i}$ is a proxy of the government's support to labor. The second term is a Stone-Geary type utility function describing the wage-employment gains associated with participating in a collective bargaining protocol \cite{lee2006racism}. The third term represents the workers' view on identity issues. For example, a higher $\mathcal{R}$ can represent a higher degree of racism among workers; whereas a lower $\bar{\mathcal{R}} $ may represent a greater government support to minorities. The last term is meant to represent the workers' view on ``social justice,'' where $\mathcal{Q}$ is a measure of economic equality and  $\bar{\mathcal{Q}}$ is the perceived ideal level of inequality by the typical worker \cite{alesina2011preferences}.  The utility of workers under individual bargaining is simply  $U^{i,2}_{W}= \omega_{1}L^{n} + \omega_{2} w^{n}$ for $i\in \{1,2\}$.

 \begin{table}
 \centering
\resizebox{0.75\textwidth}{!}{
 \begin{tabular}{ l  c  c | c  }
 \hline \hline 
 & \vline & Collective bargaining & Individual individual \\
 \hline 
High political  support & \vline & $U^{1,1}_{W}, U^{1,1}_{G}$ & $U^{1,2}_{W}, U^{1,2}_{G}$   \\
 \hline 
Low political support & \vline &  $U^{2,1}_{W}, U^{2,1}_{G}$ &  $U^{2,2}_{W}, U^{2,2}_{G}$ \\
 \hline \hline \\
 \end{tabular}}
\caption{Payoff table.}\label{table:conflict_problem}
 \end{table}

For conceptual simplicity,  I  assume that the government is exclusively interested in maximizing its vote share. In each scenario, the government gets

\begin{equation*}
\begin{split}
U^{1,j}_{G}&=\mathcal{V}_{1,j}+ \mathcal{V}_{3}  \varphi, \;\;  \text{with   } \mathcal{V}_{3} >0. \\
U^{2,j}_{G}&=\mathcal{V}_{2,j}, \;\;  \text{with   }  j \in \{1,2\}.
\end{split}
\end{equation*}

Here $ \varphi$ is the measure of the ``Communist threat'' and $ \mathcal{V}_{i,j} $ is an autonomous component capturing the public's preference in each possible scenario.  Surely this is over simplistic, but it helps illustrate  how  the Communist threat can induce  the government to favor a bigger welfare state to avoid losing public support.\footnote{This is well represented in a letter of president Eisenhower to his brother in the early 1950s, where he stated: ``Should any political party attempt to abolish social security, unemployment insurance, and eliminate labor laws... you would not hear of that part again in our political history'' \cite[p. 45]{gerstle2022}.}

                \begin{figure}
\begin{center}
\pgfplotstableread[col sep=comma,]{dt_quantal.csv}\dataquantal
  \begin{tikzpicture}
  \begin{axis}[
  name=plot1,
  width=0.45*\textwidth, height=3.6 cm,
    y label style={at={(axis description cs:-0.065,.5)},rotate=90,anchor=south},
  xmin=-0.16,
  xmax=0.21,
  ymin=0.0,
  ymax=0.6,
   legend pos=north east,
  legend style={draw=none,font=\tiny},
   legend cell align={left}, 
      axis x line*=bottom,
axis y line*=left,
    tick label style={font=\tiny,/pgf/number format/fixed},
  tick label style={font=\tiny},
  title={\tiny Panel A. $P(\mathcal{U}=1|\cdot) P^{G}(\mathcal{S}=1|\cdot)$ },
    title style={at={(axis description cs:0.26,1)}, anchor=south} , 
  xticklabel style={/pgf/number format/set thousands separator={}},
  xtick={-0.2,-0.1,...,0.3},
  ytick={0,0.1,...,1}  ]

   \addplot [mark=square*, mark size = 1pt,mark options={fill=black!30}, draw=black,line width= 1, smooth] table[x index = {0}, y index = {1},
  each nth point={4}]{\dataquantal};
        
  \end{axis}

   \begin{axis}[%
  name=plot2,
  width=0.45*\textwidth, height=3.6 cm,
      at=(plot1.right of south east), anchor=left of south west,
    y label style={at={(axis description cs:-0.065,.5)},rotate=90,anchor=south},
  xmin=-0.16,
  xmax=0.21,
  ymin=0.1,
  ymax=0.4,
  xshift=-0.1 cm,
        axis x line*=bottom,
axis y line*=left,
  tick label style={font=\tiny},
  legend pos=south west,
  legend style={draw=none,font=\tiny},
   legend cell align={left}, 
  title={\tiny Panel B.    $P(\mathcal{U}=2|\cdot) P^{G}(\mathcal{S}=1|\cdot)$},
    tick label style={font=\tiny,/pgf/number format/fixed},
    title style={at={(axis description cs:0.36,0.975)}, anchor=south} , 
  xticklabel style={/pgf/number format/set thousands separator={}},
  xtick={-0.2,-0.1,...,0.3},
  ytick={0,0.1,...,1} ]  
   \addplot [mark=square*, mark size = 1pt,mark options={fill=black!30}, draw=black,line width= 1, smooth] table[x index = {0}, y index = {2},
  each nth point={4}]{\dataquantal};
    \end{axis}
    
    \begin{axis}[
  name=plot3,
  width=0.45*\textwidth, height=3.6 cm,
   at=(plot1.below south east), anchor=above north east, 
    y label style={at={(axis description cs:-0.065,.5)},rotate=90,anchor=south},
  xmin=-0.16,
  xmax=0.21,
  ymin=0.1,
  ymax=0.4,
       axis x line*=bottom,
axis y line*=left,
  tick label style={font=\tiny,/pgf/number format/fixed},
  title={\tiny Panel C. $P(\mathcal{U}=1|\cdot) P^{G}(\mathcal{S}=2|\cdot)$  },
      title style={at={(axis description cs:0.35,1)}, anchor=south} , 
  xticklabel style={/pgf/number format/set thousands separator={}},
  xtick={-0.2,-0.1,...,0.3},
  ytick={0,0.1,...,1}  ]  
  
   \addplot [mark=square*, mark size = 1pt,mark options={fill=black!30}, draw=black,line width= 1, smooth] table[x index = {0}, y index = {3},
  each nth point={4}]{\dataquantal};
    \end{axis}

    \begin{axis}[
  name=plot4,
  width=0.45*\textwidth, height=3.6 cm,
 at=(plot3.right of north east), anchor=left of north west,
    y label style={at={(axis description cs:-0.065,.5)},rotate=90,anchor=south},
  xmin=-0.16,
  xmax=0.21,
  ymin=0,
  ymax=0.5,
    xshift=-0.1cm,
         axis x line*=bottom,
axis y line*=left,
  tick label style={font=\tiny,/pgf/number format/fixed},
  title={\tiny Panel D.  $P(\mathcal{U}=2|\cdot) P^{G}(\mathcal{S}=2|\cdot)$},
      title style={at={(axis description cs:0.3,1)}, anchor=south} , 
  xticklabel style={/pgf/number format/set thousands separator={}},
    ytick={0,0.1,...,1},
  xtick={-0.2,-0.1,...,0.3}  ]  

     \addplot [mark=square*, mark size = 1pt,mark options={fill=black!30}, draw=black,line width= 1, smooth] table[x index = {0}, y index = {4},
  each nth point={4}]{\dataquantal};
    \end{axis}
  
    \end{tikzpicture}
  \end{center}
  \caption{ \scriptsize EQUILIBRIUM POLITICAL STATE AND THE COMMUNIST THREAT.\label{fig:quantal}}
    \caption*{ \scriptsize \emph{Notes--- $U^{1,1}_{W}=1+0.75 \varphi$,   $U^{1,2}_{W}=U^{2,1}_{W}=U^{2,2}_{W}=1$ and $\lambda^{W}=6$. Correspondingly,  $U^{1,1}_{G}=0.5 + 0.5 \varphi$, $U^{1,2}_{G}=0.3 + 0.5 \varphi$, $U^{2,1}_{G}=0.3$, $U^{2,2}_{G}=0.5$ and $\lambda^{G}=11$.}}
  \end{figure}

If  workers and  the government maximize the expected payoff associated with each strategy  in Table \ref{table:conflict_problem}, subject to an entropy constraint, there will exist  (given the appropriate regularity conditions)  a unique   Nash equilibrium with mixed strategies  \cite{mackowiak2023rational}

\begin{equation}\label{eq:quantal_worker}
P(\mathcal{U}=i|\cdot) = \frac{e^{\lambda^{W} \sum_{j=1}^{2} P^{G}(S=j|\cdot) U^{j,i}_{W}} }{\sum_{j'=1}^{2} e^{\lambda^{W} \sum_{j=1}^{2} P^{G}(S=j|\cdot) U^{j,j'}_{W}} }
\end{equation}

and 

\begin{equation}\label{eq:quantal_gov}
P^{G}(S=j|\cdot) = \frac{e^{\lambda^{G} \sum_{i=1}^{2} P(\mathcal{U}=i|\cdot) U^{ji}_{G}} }{\sum_{j'=1}^{2} e^{\lambda^{G} \sum_{j=1}^{2} P(\mathcal{U}=j|\cdot) U^{j'j}_{G}} }
\end{equation}

Here $P(\mathcal{U}=1|\cdot)$ denotes the probability of collective bargaining and $P^{G}(S=1|\cdot) $ is the probability that the government  provides high institutional support to labor.  The key feature of  \eqref{eq:quantal_worker} and \eqref{eq:quantal_gov} is that by introducing some ``randomness'' in the  behavior of workers and the government (represented by $\lambda^{W}$ and $\lambda^{G}$), both equations  capture the  complexity of aggregating over  heterogeneous  individuals with limited information-processing capacities. 

Figure \ref{fig:quantal} illustrates the basic argument of the decision model by associating each equilibrium outcome with  the proxy of the Communist threat. For instance, the model shows that the probability of an equilibrium with high institutional support to labor and  high collective bargaining increases with a rise  in $\varphi$---as illustrated  in the data of Figure \ref{fig:inverse_inference}, where   the surge in the relative real GDP per capita of the Soviet Union was accompanied by a rise in the institutional support to labor in the US. Correspondingly,   Figure \ref{fig:quantal}  shows that a decrease in $\varphi$ can raise the probability of an equilibrium with low institutional support to labor and a higher density of individual bargaining, as it happened in the US following the mid-1970s. 

These results do not substitute, but rather complement the existing studies associating factors like racism and the ``American exceptionalism" with the public's support to welfare \cite{lee2006racism,alesina2011preferences}. In fact, this is a potentially  fruitful area for future research   since it can help disentangle the causes determining the political state of society and thus the factors which shape the  power of labor. 

\subsection{Auxiliary results to Section \ref{sec:equilibrium_dynamics}} This subsection presents the theoretical structure for Figures \ref{fig:ramsey_df} and \ref{fig:automation_reg}.

 \subsubsection{Arbitrage Condition}\label{subappendix:arbitrage} Assume the existence of a representative capitalist consumer  looking to maximize\footnote{To save notation I  assume that $q_{t}=q$ and $a_{0}=a_{1}=0$, as in the text.}

  \begin{equation*}
   \int_{0}^{\infty} e^{ -(\mu^{*}_{t}\delta_{t} -\epsilon g_{t}) t}  \;   \frac{C_{t}^{1-\epsilon} -1}{1-\epsilon} \mathrm{d}t \quad  \text{s.t.  }
\eqref{eq:capital_dyn}
 \end{equation*}
 
  Capitalist consumption is  $C_{t}$,  $\epsilon >0$ is  the intertemporal elasticity of substitution, and $g_{t} = \alpha (\dot{M}_{t} -\dot{m}^{*}_{t}) + \dot{L}_{t}/L_{t}$ is the actual rate of growth.  Here the discount rate, or the  \emph{competitive opportunity cost} faced by the representative consumer,   is divided in two complementary parts. The first  is $\mu^{*}_{t}\delta_{t}$, which---similar to \citeasnoun{abel1989assessing}---represents   the equilibrium marginal net rate of return per unit of capital.\footnote{Intuitively, $ \mu^{*}_{t}\delta_{t} $  can be interpreted as the equilibrium return that a typical capitalist can expect to receive in a competitive environment  from an additional unit of productive capital. This takes the place of the \emph{required rate of return} commonly used in the literature.}   The second element  $(\epsilon g_{t})$  states that regardless of the activity chosen by the capitalist, it will always expect a diminishing marginal utility of consumption resulting  from the  expansion of the economy.

Expressing the results using stationary per-capita variables:
 
 \begin{equation}\label{eq:euler}
\frac{ \dot{\hat{c}}_{t}}{\hat{c}_{t}} = \frac{1}{\epsilon} \Big[ \hat{y}_{\hat{k}_{t}} q -\delta_{t} (1+\mu^{*}_{t}) \Big]
 \end{equation}

   \begin{equation}\label{eq:transversality}
\underset{t \rightarrow \infty}{\mathrm{lim}}\;  \hat{k}_{t} \; e^{-\int_{0}^{t} (\mu_{t'} \delta_{t'}  - g_{t'})\mathrm{d}t'} = 0.
 \end{equation}
 
 Equation \eqref{eq:euler} is meant to create an \emph{analogy} of the social conditions of  arbitrage characterizing the tendency towards an equilibrium rate of return. This is  clear if we use equation \eqref{eq:online_marginal},  in which case \eqref{eq:euler} is reduced to $\dot{\hat{c}}_{t}/\hat{c}_{t} = \delta (\mu_{t} - \mu^{*}_{t})/\epsilon$. By this logic, there is a flat consumption profile, $\dot{\hat{c}}=0$, when $\mu_{t}=\mu^{*}_{t}$,  indicating that there  are no net advantages for changes in the use of capital. However, if $\mu_{t} > \mu^{*}_{t}$, capitalists  will be willing to sacrifice some consumption today for consumption tomorrow given that current capital inflows will be rewarded above its equilibrium level. 
 
 Equation \eqref{eq:transversality} shows that  in a dynamic efficient equilibrium the marginal net return per unit of capital must be greater than the equilibrium growth rate of the economy.

  \subsubsection{Automation Regions}\label{subappendix:Automation_regions}
 
The next lemma is a modified version of Lemma A2 in \citeasnoun{acemoglu2018race}.

\begin{lemma}\label{lemma:lemma_2A_AR} Suppose that Assumption \ref{ass:management} holds and that the economy is initially in a BGP with positive growth satisfying \eqref{eq:min_surplus}. Then, for a given $\mu^{*}$, there exist $q^{\text{min}} < \bar{q} < q^{\text{max}}$ such that:

\begin{enumerate}[label=(\emph{\roman*})]
\item If $q \in [q^{\text{min}}, \bar{q}]$, there is a decreasing function $\bar{m}(q): q \in [q^{\text{min}}, \bar{q}]  \rightarrow (0,1)$ such that for all $ m > \bar{m}(q)$, we have $w_{J}(m) > D(h_{t})^{a_{0}} \delta/(A^{k} q) < w_{M}(m)$ and $D(h_{t})^{a_{0}} \delta/(A^{k} q) = w_{M}(\bar{m}(q))$. Moreover, $\bar{m}(q^{\text{min}})=1$ and $\bar{m}(\bar{q})=0$.  

\item If $q \in [\bar{q},q^{\text{max}}]$, there is an increasing function $\tilde{m}(q): q \in [\bar{q},q^{\text{max}}]  \rightarrow (0,1)$ such that for all $ m > \tilde{m}(q)$, we have $w_{J}(m) > D(h_{t})^{a_{0}}  \delta/(A^{k} q) < w_{M}(m)$ and $D(h_{t})^{a_{0}}  \delta/(A^{k} q) = w_{J}(\tilde{m}(q))$. Moreover, $\tilde{m}(q^{\text{max}})=1$ and $\tilde{m}(\bar{q})=0$.  
\end{enumerate}
 
The case where $q>q^{\text{max}}$ and $q<q^{\text{min}}$ is analogous to cases (iii) and (iv) in \citeasnoun[p. 1531]{acemoglu2018race}.
\end{lemma}

\begin{proof}
See Online Appendix \ref{appendix:proofs}. 
\end{proof}

\bibliographystyle{econometrica}
{\footnotesize\bibliography{references_ciep}}

\newpage

\appendix

\numberwithin{equation}{section}
\numberwithin{lemma}{section}
\numberwithin{assm}{section}
\numberwithin{defn}{section}
\numberwithin{figure}{section}
\numberwithin{table}{section}
\renewcommand{\theequation}{\thesection\arabic{equation}}
\renewcommand{\thelemma}{\thesection\arabic{lemma}}
\renewcommand{\thedefn}{\thesection\arabic{defn}}
\renewcommand{\theassm}{\thesection\arabic{assm}}
\renewcommand{\thefigure}{\thesection\arabic{figure}}
\renewcommand{\thetable}{\thesection\arabic{table}}

\section*{ONLINE APPENDIX FOR:   ``THERE IS POWER IN GENERAL EQUILIBRIUM''}

\begin{center}
by \textsc{Juan Jacobo}
\end{center}

\section{Main Proofs}\label{appendix:proofs}

 \subsection{Section \ref{sec:model}} 
 
\subsubsection{Proof of Lemma \ref{lemma:management}}\label{subappendix:lemma_management} Managers have the option of increasing the productivity of workers by increasing the relative supply of effective capital. This is captured by the constraint that $a_{0},a_{1} >0$. Assuming that $J^{*}=J_{t}$, the management effort which maximizes output is consequently given by\footnote{The proof that the first order conditions give a maximum can be found in \citeasnoun{grossman2017balanced}.}

\begin{equation*}
\begin{split}
\frac{\partial Y_{t}}{\partial h_{t}}&  = \frac{-A \sigma}{\sigma -1 } Z^{\frac{-A \sigma}{\sigma -1} -1}_{t} \Big(L_{t} e^{\alpha J^{*}_{t}}\Big)^{1-A} K^{A}_{t} H(Z_{t}) \frac{\partial Z_{t}}{\partial h_{t}} \\
& + \frac{ \sigma}{\sigma -1 } Z^{\frac{-A \sigma}{\sigma -1}}_{t} \Big(L_{t} e^{\alpha  J^{*}_{t}}\Big)^{1-A} K^{A}_{t} H(Z_{t})^{1/\sigma} \frac{\partial Z_{t}}{\partial h_{t}} = 0 
\end{split}
\end{equation*}

where $H(Z_{t}) = \Bigg[(1-m^{*}_{t})^{1/\sigma} Z_{t} + \Bigg(\frac{e^{\alpha (\sigma-1)m^{*}_{t}} -1}{\alpha (\sigma-1)}\Bigg)^{1/\sigma} \Bigg]^{\frac{\sigma}{\sigma-1}}$ and $A=a_{1}/(a_{0}+a_{1})$. Organizing terms it follows that

\begin{equation*}
Z_{t} \Big[(1-m^{*}_{t})^{1/\sigma} Z_{t} + \Big(\frac{e^{\alpha (\sigma-1)m^{*}_{t}} -1}{\alpha (\sigma-1)} \Big)\Big]^{-1}=A
\end{equation*}

That is, given the BGP condition that $m^{*}_{t}=m$, it follows that $Z_{t}$ is constant if capitalists set an optimal management effort. 

Differentiating $Z_{t}$ with respect to time:

\begin{equation}\label{eq:online_dotZ1}
\frac{\dot{Z}_{t}}{Z_{t}} =  \frac{\sigma-1}{\sigma} \Big[ g_{K} - (a_{0}+a_{1}) \frac{D'(h_{t})}{D(h_{t})}\dot{h}_{t} - \alpha  \dot{J}^{*}_{t}\Big]=0,
\end{equation}

since $\dot{L}_{t}=0$ in a BGP. Now, using the aggregate production function  with $Z_{t}$ fixed and following similar steps as those in \citeasnoun{grossman2017balanced},  it follows that $g_{Y}=A g_{K} + (1-A) \alpha \dot{J}^{*}_{t} $. Thus, 
\begin{equation}\label{eq:online_gy1}
\begin{split}
 & (a_{0}+a_{1}) g_{Y} = a_{1} g_{K} + a_{0} \alpha  \dot{J}^{*}_{t} \\
\Rightarrow & (a_{0}+a_{1}) g_{Y} = a_{1} (g_{Y}+g_{q}) + a c \dot{J}^{*}_{t}\\
\Rightarrow &  g_{Y} =  \alpha  \dot{J}^{*}_{t} + \frac{a_{0}}{a_{1}} g_{q}
\end{split}
\end{equation}

where the second line uses $g_{K}=g_{Y}+g_{q}$. The last line in \eqref{eq:online_gy1} is one of the results of Lemma \ref{lemma:management}. Replacing \eqref{eq:online_gy1} in \eqref{eq:online_dotZ1} it follows that

\begin{equation*}
\frac{D'(h_{t})}{D(h_{t})}\dot{h}_{t}  = g_{q}/a_{0},
\end{equation*}

which is also mentioned in Lemma \ref{lemma:management}.  Lastly, in order to show that $g_{K}=g_{Y}+g_{q}$, we use the budget constraint of capitalists and workers to obtain 

\begin{equation*}
Y_{t} = C^{T}_{t} + \frac{I_{t}}{q_{t}} +\xi V_{t},
\end{equation*}

where $C^{T}_{t} \equiv C_{t} + C^{w}_{t}$ is consumption  of capitalists and workers and $U_{t} B_{t} = T_{t}$.  Integrating $\xi V_{t}$ as a form of consumption such that $\bar{C}_{t} = C^{T}_{t}  + \xi V_{t}$, it follows that 

\begin{equation*}
g_{Y} = g_{C} \frac{\bar{C}_{t} }{Y_{t}} + \frac{I_{t}/q_{t}}{Y_{t}} \big(g_{I}-g_{q}\big).
\end{equation*}

Given \eqref{eq:capital_dyn} in the main text, $g_{I}=g_{K}$ in a BGP, so 

\begin{equation*}
\begin{split}
g_{Y} &= g_{C} \frac{\bar{C}_{t} }{Y_{t}} + \frac{I_{t}/q_{t}}{Y_{t}} \big(g_{K}-g_{q}\big)
= g_{C} \frac{\bar{C}_{t} }{Y_{t}}  + \frac{Y_{t}-\bar{C}_{t}}{Y_{t}} \big(g_{K}-g_{q}\big)\\
 &=  \frac{\bar{C}_{t} }{Y_{t}} \big(g_{C} - g_{K}+g_{q}\big) + \big(g_{K}-g_{q}\big).\\
\Rightarrow g_{Y}-g_{K}+g_{q} &=  \frac{\bar{C}_{t} }{Y_{t}} \big(g_{C} - g_{K}+g_{q}\big)=0
\end{split}
\end{equation*}

This  finishes the proof of Lemma \ref{lemma:management}. Lemma \ref{lemma:management_text} is  a special case where $q_{t}=q$ and $a_{0}=a_{1}=0$.

\subsubsection*{Value Functions} The Hamiltonian associated with the problem of workers is

\begin{equation*}
\mathcal{H}^{w}(L_{t}, U_{t}) = e^{-\rho t} \big[ L_{t} (W_{t}/P_{t}) + U_{t} (B_{t}/P_{t})\big] + \varphi^{w}_{1,t} \big(f(\theta_{t}) U_{t} - \lambda_{t} L_{t}) +  \varphi^{w}_{2,t} \big( \lambda_{t} L_{t} +f(\theta_{t}) U_{t} )
\end{equation*}

 where $\varphi^{w}_{1,t}$ and $\varphi^{w}_{2,t}$ are co-state variables for employment and unemployment, respectively. The following necessary conditions hold for the Hamiltonian:
 
 \begin{equation*}
 \begin{split}
 \frac{\partial \mathcal{H}^{w}(\cdot)}{\partial L_{t}} &:  e^{-\rho t}  \frac{W_{t}}{P_{t}} - \lambda_{t}(\varphi^{w}_{1,t}-\varphi^{w}_{2,t})= - \dot{\varphi}^{c}_{1,t} \\
  \frac{\partial \mathcal{H}^{w}(\cdot)}{\partial U_{t}} &:  e^{-\rho t}  \frac{B_{t}}{P_{t}}  + f(\theta_{t})\big(\varphi^{w}_{1,t}-\varphi^{w}_{2,t}\big) =   - \dot{\varphi}^{c}_{2,t} 
 \end{split}
 \end{equation*}
 
 Dividing both sides of the equations by $e^{-\rho t}  e^{\alpha  J^{*}_{t}} D(h_{t})^{a_{1}}$ and expressing the stationary marginal value of employment as $\phi_{L_{t}} = \varphi^{w}_{1,t}/\big(e^{-\rho t}  e^{\alpha  J^{*}_{t}} D(h_{t})^{a_{1}}\big)$ and of unemployment as $\phi_{U_{t}} = \varphi^{w}_{2,t}/\big(e^{-\rho t}   e^{\alpha  J^{*}_{t}} D(h_{t})^{a_{1}}\big)$, we reach the results in equations  \eqref{eq:un_value} and \eqref{eq:job_value_worker}. 
 
 The Hamiltonian associated with the problem of capitalists is

\begin{equation*}
\mathcal{H}^{c}(L_{t}, U_{t}) = e^{-\rho t} \big[Y_{t} - L_{t} \frac{W_{t}}{P_{t}}  -  V_{t} \frac{\Xi_{t}}{P_{t}} - \frac{T_{t}}{P_{t}}- \frac{I_{t}}{q_{t}}  \big] + \varphi^{c}_{1,t} \big(q(\theta_{t}) V_{t} - \lambda_{t} L_{t}) + \varphi^{c}_{2,t} \big(  \lambda_{t} L_{t}-q(\theta_{t}) V_{t}\big)
\end{equation*}

 where $\varphi_{1,t}^{c}$ and $\varphi_{2,t}^{c}$  are  co-state variables of employment and vacancies, respectively. The necessary conditions are described as follows
 
  \begin{equation*}
 \begin{split}
 \frac{\partial \mathcal{H}^{c}(\cdot)}{\partial L_{t}} &:  e^{-\rho t} \big( Y_{L_{t}}-(W_{t}/P_{t})\big)  -  \lambda_{t}(\varphi^{c}_{1,t}-\varphi^{c}_{2,t})= - \dot{\varphi}^{c}_{1,t} \\
  \frac{\partial \mathcal{H}^{c}(\cdot)}{\partial V_{t}} &:  - e^{-\rho t}  \frac{\Xi_{t}}{P_{t}} + q(\theta_{t})\big(\varphi^{c}_{1,t}-\varphi^{c}_{2,t}\big) =   - \dot{\varphi}^{c}_{2,t} 
 \end{split}
 \end{equation*}

 As before, let $\pi_{L_{t}}= \varphi^{c}_{1,t}/\big(e^{-\rho t}  e^{\alpha  J^{*}_{t}} D(h_{t})^{a_{1}}\big)$ and $\pi_{V_{t}}= \varphi^{c}_{2,t}/\big(e^{-\rho t}  e^{\alpha J^{*}_{t}} D(h_{t})^{a_{1}}\big)$ denote the stationary marginal value of employment and vacancy posting for the firm. Dividing both sides of the first-order conditions with respect to  $e^{-\rho t}  e^{\alpha  J^{*}_{t}} D(h_{t})^{a_{1}}$  we get the result in \eqref{eq:value_vacancy_firm} and \eqref{eq:value_job_firm}.

\subsection{Section \ref{sec:barg_protocol}}\label{appendix:proof_sectionprotocol} This section presents the proofs of Proposition \ref{prop:non_union_wage}, Corollary \ref{coro:barg_power}, Proposition \ref{prop:union_wage}, and Proposition \ref{prop:equil_labor}. 

\subsubsection{Proof Proposition \ref{prop:non_union_wage}}\label{subappendix:proof_indwages} The proof of the individual bargaining solution builds from \citeasnoun{shaked1984involuntary}. Focusing on the stationary value functions, and denoting $w_{f}$ and $w_{e}$ as the wage proposal of firms and workers, respectively, the subgame perfect equilibrium satisfies (time arguments are ignored to save notation)

\begin{equation}\label{eq:subgame}
\begin{split}
\pi_{L}(w_{e}) &= (1-\Delta_{f} \lambda) e^{\rho \Delta_{f}} \pi_{L}(w_{f}) \\
(\phi_{L}(w_{f})-\phi_{U}) &= (1-\Delta \lambda) e^{\rho \Delta}(\phi_{L}(w_{e})-\phi_{U})
\end{split}
\end{equation}

 For convenience in notation, assume that the random draws from the exponential distributions of the waiting times  in the individual bargaining protocol are equal to their corresponding averages. This does not affect the final conclusions given that they rely on the law of large numbers.  

\textbf{Case $\boldsymbol{T(\theta)=1/q(\theta)}$:}  Assuming that $\Delta/q(\theta)$ is an odd number,\footnote{\citeasnoun{shaked1984involuntary} show that the same argument applies when this assumption is dropped. However, the argument is significantly simpler by assuming that $\Delta/q(\theta)$ is  odd.} at time $\Delta T(\theta)$ the firm can switch and start a  bargaining process with a new unemployed worker. Let $\mathcal{G}$ denote the subgame starting at $t=0$ in the rightmost branch of Figure \ref{fig:bargaining_protocol} and let $\mathcal{G}_{0}$ be the subgame starting at $t=\Delta T(\theta)$ when the firm is contemplating switching to a new bargaining partner. Correspondingly, let $\mathcal{M}$ and $\mathcal{M}_{0}$ denote the supremum for the firm in the respective subgame. The supremum of $\mathcal{G}_{0}$ must satisfy

\begin{equation*}
\mathcal{M}_{0} = \mathrm{max}\big\{\delta^{f} \big(S^{n} (1-\delta^{w}) + \delta^{w} \mathcal{M}_{0}\big), \mathcal{M}\big\}
\end{equation*}

where $\delta^{f} = (1-\Delta_{f} \lambda) e^{\rho \Delta_{f}}$ and $\delta^{w} = (1-\Delta \lambda) e^{\rho \Delta}$. Beginning with this condition in period $\Delta T(\theta)$, we can iterate backwards until period $t=0$. For instance, in period $\Delta T(\theta)-\Delta$, when contemplating an offer made by the firm, the worker solves

\begin{equation*}
\phi_{L}(w_{f}) = \phi_{U}(1-\delta^{w}) + \delta^{w} (S^{n} + \phi_{U}  -\mathcal{M}_{0}) 
\end{equation*}

Using the previous equation and the surplus equation for the individual bargaining problem we obtain

\begin{equation*}
\pi_{L}(w_{f}) = S^{n}(1-\delta^{w}) + \delta^{w} \mathcal{M}_{0}. 
\end{equation*}

In period $\Delta T(\theta)-2\Delta$, when contemplating an offer made by the worker, the firm solves

\begin{equation*}
\pi_{L}(w_{e}) = \delta^{f} \Big[(1-\delta^{w}) S^{n} + \delta^{w} \mathcal{M}_{0}\big] = \delta^{f} (1-\delta^{w}) S^{n} + \delta^{f} \delta^{w} \mathcal{M}_{0}.
\end{equation*}

By induction, for any $T(\theta)$:\footnote{An intuitive way of interpreting the result is to set $\mathcal{M} = \pi_{L}(w_{f})$ and $\mathcal{M}_{0} = \pi_{L}(w_{e})$, both of which can be deduced  using the surplus of the individual bargaining problem.}

\begin{equation*}
\mathcal{M}= S^{n}\big( (1-\delta^{w} + \delta^{f} \delta^{w} (1-\delta^{w}) + \cdots + (\delta^{w} \delta^{f})^{\frac{T(\theta)-1}{2}} (1-\delta^{w}) \big) + {\delta^{f}}^{\frac{T(\theta)-1}{2}} {\delta^{w}}^{\frac{T(\theta)+1}{2}} \mathcal{M}_{0}.
\end{equation*}

Denoting $\tilde{\delta}=\delta^{w} \delta^{f}$ and noting that $\sum_{i=0}^{\frac{T(\theta)-1}{2}} \tilde{\delta}^{i} = \sum_{i=0}^{\tilde{T}-1} \tilde{\delta}^{i} = \frac{1-\tilde{\delta}^{\tilde{T}}}{1-\tilde{\delta}}$:

\begin{equation*}
\mathcal{M}= \frac{S^{n} (1-\delta^{w})(1-\tilde{\delta}^{\tilde{T}})}{1-\tilde{\delta}} + {\delta^{f}}^{\frac{T(\theta)-1}{2}} {\delta^{w}}^{\frac{T(\theta)+1}{2}} \mathcal{M}_{0}.
\end{equation*}

Now we ask whether $\mathcal{M}_{0}=\mathcal{M}$ or $\mathcal{M}_{0}= \delta^{f} \big(S^{n} (1-\delta^{w}) + \delta^{w} \mathcal{M}_{0}\big)$. If the latter is the case, then

\begin{equation*}
\mathcal{M}_{0} = \frac{S^{n}\delta^{f} (1-\delta^{w})}{1-\tilde{\delta}},
\end{equation*}

so

\begin{equation*}
\begin{split}
\mathcal{M} &=  \frac{S^{n} (1-\delta^{w})(1-\tilde{\delta}^{\tilde{T}})}{1-\tilde{\delta}} + {\delta^{f}}^{\frac{T(\theta)-1}{2}} {\delta^{w}}^{\frac{T(\theta)+1}{2}} \Bigg[\frac{S^{n}(1-\delta^{w})}{1-\tilde{\delta}}\Bigg]\\
& =  \frac{(1-\delta^{w}) S^{n}}{1-\tilde{\delta}} > \mathcal{M}_{0},
\end{split}
\end{equation*}

 which is a contradiction because  $\delta^{f}<1$. Thus

\begin{equation*}
\mathcal{M}  = \frac{S^{n} (1-\delta^{w})(1-\tilde{\delta}^{\tilde{T}})}{1-\tilde{\delta}} + {\delta^{f}}^{\frac{T(\theta)-1}{2}} {\delta^{w}}^{\frac{T(\theta)+1}{2}} \mathcal{M}
\end{equation*}

Solving the previous equation it is obtained that 

\begin{equation*}
\mathcal{M} = \frac{S^{n}(1-\delta^{w})(1-\tilde{\delta}^{\tilde{T}})}{(1-\tilde{\delta})(1-\tilde{\delta}^{\tilde{T}}/\delta^{f})}
\end{equation*}

Reverting back to the complete notation, the supremum of $\mathcal{G}$ can be written as

\begin{equation*}
\begin{split}
\mathcal{M} & = S^{n} \frac{\big(1-(1-\lambda \Delta) e^{-\rho \Delta}\big)}{\big[ 1-(1-\lambda \Delta) (1-\lambda_{f} \Delta) e^{-(\rho+\rho_{f}) \Delta}\big] }  \times \\
& \frac{1-(1-\lambda \Delta)^{\frac{T(\theta)+1}{2}} (1-\lambda_{f} \Delta)^{\frac{T(\theta)+1}{2}} e^{-(\rho+\rho_{f})\big(\frac{T(\theta)+1}{2}\big) \Delta}}{\big[1-(1-\lambda \Delta)^{\frac{T(\theta)+1}{2}}e^{-\rho\big(\frac{T(\theta)+1}{2}\big) \Delta}  (1-\lambda_{f} \Delta)^{\frac{T(\theta)-1}{2}} e^{-\rho_{f}\big(\frac{T(\theta)-1}{2}\big) \Delta}]} 
\end{split}
\end{equation*}

Where $\lambda_{f} = \lambda \gamma^{f}$ and $\rho_{f} =\rho \gamma^{f}$. Applying L'Hôpital's rule to the last equation, it follows that when $\Delta \rightarrow 0$:

\begin{equation*}
\mathcal{M} = \pi_{L}(w) =  S^{n} \frac{T(\theta)+1}{T(\theta)(1+\gamma^{f}) + (1-\gamma^{f})}
\end{equation*}

This result follows from the fact that when $\Delta \rightarrow 0$, $w_{e} = w_{f}=w$. The previous equation  also shows that when  $\Delta \rightarrow 0$, $\mathcal{M}$ can be associated with the first order conditions of the generalized Nash solution with an intrinsic labor power

\begin{equation*}
\Gamma^{nb}  = 1- \frac{T(\theta)+1}{T(\theta)(1+\gamma^{f}) + (1-\gamma^{f})} = \frac{\gamma^{f} (1-q(\theta))}{1+\gamma^{f} - q(\theta)(1-\gamma^{f})},
\end{equation*}

which corresponds to the labor power under individual bargaining in Proposition \ref{prop:non_union_wage}. Now, given that $\phi_{L}-\phi_{U}=S^{n}/\Gamma^{nb} = \frac{w - (\rho + \alpha  \dot{m}^{*} - g) \phi_{U} }{\lambda +\rho + \alpha  \dot{m}^{*} - g}$,  $(\rho + \alpha  \dot{m}^{*} - g) \phi_{U}  = b + f(\theta) S^{n} \Gamma^{nb} $, and $S^{n} = \frac{y_{L}-(\rho + \alpha  \dot{m}^{*} - g)\phi_{U} }{\lambda +\rho + \alpha \dot{m}^{*} - g}$, it follows that 

\begin{equation}
w^{nb}= b + \Psi^{nb}(y_{L} - b), \;\; \text{with   } \Psi^{nb} = \frac{\Gamma^{nb}\big(\lambda +\rho + c \dot{m}^{*} - g + f(\theta)\big)}{\lambda +\rho + c \dot{m}^{*} - g + \Gamma^{nb} f(\theta)},
\end{equation}

which completes the proof of part (ii)  in Proposition \ref{prop:non_union_wage}. 

\textbf{Case $\boldsymbol{T(\theta)=T^{w}/f(\theta)}$:} Here we assume that $\Delta T^{w}/f(\theta)$ is an even number, such that at time $\Delta T(\theta)$ the worker can switch and start a bargaining process with a new employer. 

The  subgame $\mathcal{G}$ at $t=0$, when $T(\theta)=T^{w}/f(\theta)$,  is identical to the rightmost branch of Figure \ref{fig:bargaining_protocol}. The difference in this case is that at $t=\Delta T(\theta)$, the decision of staying with the same firm or switching  to a new employer is different for the worker than it is for the firm. If the worker decides to stay with the same firm, it will make the wage offer the following period. However, if the worker decides to switch, it will have to hear the offer of the new firm in the upcoming period and will only be able to respond after two periods.

 Denoting $\mathcal{N}_{0}$ as the supremum for the worker in game $\mathcal{G}_{0}$, which corresponds to the subgame starting at $t=\Delta T(\theta)$, it will follow that

\begin{equation*}
\mathcal{N}_{0} = \mathrm{max}\big\{\delta^{w}\big( S^{n}(1-\delta^{f}) + \delta^{f} \mathcal{N}_{0}\big), \mathcal{N}\big\},
\end{equation*}

where $\mathcal{N}$ is the supremum for the worker in subgame $\mathcal{G}$. Following the same argument as in the  previous case, we get by backward induction that in period $t=0$:

\begin{equation*}
\mathcal{N} = S^{n} \frac{\delta^{w} (1-\delta^{f}) (1-\tilde{\delta}^{T(\theta)/2})}{1-\tilde{\delta}} + \tilde{\delta}^{T(\theta)/2} \mathcal{N}_{0}.
\end{equation*}

Suppose that $\mathcal{N}_{0}  = \delta^{w}\big( S^{n}(1-\delta^{f}) + \delta^{f} \mathcal{N}_{0}\big)$, such that

\begin{equation*}
\begin{split}
\mathcal{N} &= S^{n} \frac{\delta^{w} (1-\delta^{f}) (1-\tilde{\delta}^{T(\theta)/2})}{1-\tilde{\delta}} + \tilde{\delta}^{T(\theta)/2} \Bigg[ \frac{S^{n} \delta^{w} (1-\delta^{f})}{1-\tilde{\delta}}  \Bigg] \\
& = \frac{S^{n} \delta^{w} (1-\delta^{f})}{1-\tilde{\delta}}  = \mathcal{N}_{0}.
\end{split}
\end{equation*}

That is, the case where $T(\theta)=T^{w}/f(\theta)$ is formally identical to the original \citeasnoun{rubinstein1982perfect} alternating offers model without switch points. From this basis, we can interpret the intrinsic  bargaining power of labor when $T(\theta)=T^{w}/f(\theta)$ as

\begin{equation*}
\Gamma^{na} = \underset{\theta \rightarrow \infty}{\mathrm{lim}} \; \Gamma^{nb} = \frac{\gamma^{f}}{1+\gamma^{f}}
\end{equation*}

since $q(\theta)=0$ when $\theta \rightarrow \infty$. This result follows from the assumption that all firms are identical and that firms always make the first offer.

Lastly, to prove the final part of Proposition \ref{prop:non_union_wage}, we use the assumption that waiting times follow an exponential distribution. Given  the law of large numbers,

\begin{equation*}
\begin{split}
w^{n} &= \frac{\theta}{\theta + T^{w}} w^{na} +  \frac{T^{w}}{\theta + T^{w}} w^{nb} \\
  & =  \frac{\theta}{\theta + T^{w}} \big[ b + \Psi^{na}\big(y_{L}-b\big)\big] +  \frac{T^{w}}{\theta + T^{w}} \big[ b + \Psi^{nb}\big(y_{L}-b\big)\big] \\
  & = b+ \frac{\theta \Psi^{na} + T^{w} \Psi^{nb}}{\theta + T^{w}} \big(y_{L}-b\big).
\end{split}
\end{equation*}

\subsubsection{Proof of Corollary \ref{coro:barg_power}}\label{subappendix:proof_corollary_bargpower} The proof of (i)-(iii) is  straightforward and requires no special attention. 

In part (iv), if $m^{*}_{t}=m_{t}$, then

\begin{equation*}
\frac{\partial \Psi^{ni}}{\partial \dot{m}} = \frac{\alpha +\partial \lambda/\partial \dot{m}}{\rho +\alpha \dot{m}^{*}-g +\lambda +\Gamma^{ni} f(\theta)} \times \Big[ \Gamma^{ni}-\Psi^{ni}\Big]
\end{equation*}

for $i \in \{a,b\}$. The second term in the  right-hand side of the previous equation is always negative, so the sign is determined by the first term. Particular, given Lemma \ref{lemma:tech_unemployment}, if $\alpha + \partial \lambda/\partial \dot{m} <0$, then $\frac{\partial \Psi^{ni}}{\partial \dot{m}}>0$. Lastly, since $\Psi^{n}$ is a linear combination of $\Psi^{na}$ and $\Psi^{nb}$, we get the result in Corollary  \ref{coro:barg_power} (iv).

In part (v), 

\begin{equation*}
\frac{\partial \Psi^{ni}}{\partial \dot{M}} = \frac{\partial \lambda/\partial \dot{M}- \alpha}{\rho +\alpha \dot{m}^{*}-g +\lambda +\Gamma^{ni} f(\theta)} \times \Big[ \Gamma^{ni}-\Psi^{ni}\Big]
\end{equation*}

for $i \in \{a,b\}$. If $\sigma >1$, $\partial \lambda/\partial \dot{M} <0$ from Lemma \ref{lemma:tech_unemployment} and  $\frac{\partial \Psi^{ni}}{\partial \dot{M}}>0$. If, $\sigma \in (0,1)$, since $(1-\sigma)e^{\alpha(\sigma-1)\dot{M}} <1$ for all $\dot{M}>0$, then  $\frac{\partial \Psi^{ni}}{\partial \dot{M}}>0$.  Again, since $\Psi^{n}$ is a linear combination of $\Psi^{na}$ and $\Psi^{nb}$, we get the result in Corollary  \ref{coro:barg_power} (v).

\subsubsection{Proof of Proposition \ref{prop:union_wage}}\label{subappendix:proof_unionwage} The Hamilton-Jacobi-Bellman (HJB) equation for the firm satisfies\footnote{To save notation I will assume $\dot{m}^{*}=0$. That is, $\rho-\lambda$ is actually $\rho-g+\alpha \dot{m}^{*}_{t}$. }

\begin{equation*}
\begin{split}
(\rho-g )\pi^{L}_{t} &= \frac{ \partial \pi^{L}_{t}}{\partial t} + \underset{V_{t}}{\mathrm{max}} \; \Big[ y_{t} -\xi V_{t} -w_{t} L_{t}  +\hat{\varphi}^{c}_{1t}\big(q(\theta_{t})V_{t} - \lambda_{t} L_{t}\big)\Big]\\
\Rightarrow (\rho-g)\pi^{L}_{t} &=  \frac{ \partial \pi^{L}_{t}}{\partial t} + y_{t} -w_{t}L_{t} - \lambda_{t} \frac{\xi L_{t}}{q(\theta_{t})}
\end{split}
\end{equation*}

Here, for simplicity, I assumed from the start that the marginal value of an unfilled vacancy is zero and used in the second line the result that the marginal value of an employed worker is equal to $\xi/q(\theta)$. 

In a stationary-state, $(\rho-g)\pi^{l} = y - w L - \lambda \pi_{L} L_{t} $. Additionally, since $(\rho-g)\pi_{L}= y_{L} - w  - \lambda \pi_{L}$, we can arrange terms and show that

\begin{equation*}
\pi^{L}=L\Big(\frac{\hat{y}-y_{L}}{\rho-g} + \frac{y_{L}-w}{\rho-g+\lambda}\Big).
\end{equation*}

Similarly, given that in the steady-state $\phi_{L}-\phi_{U}=\frac{w}{\rho-g+\lambda} -\frac{(\rho-g)\phi_{U}}{\rho-g+\lambda}$, the first-order conditions of the Nash bargaining problem satisfy

\begin{equation*}
\begin{split}
\Gamma^{u} (\pi^{L}/L) &= (1-\Gamma^{u}))(\phi_{L}-\phi_{U})\\
\Gamma^{u}\big( \frac{y_{L}-w}{\rho-g +\lambda} + \frac{(\rho-g+\lambda)(\hat{y}-y_{L})}{\rho-g}\big) &= (1-\Gamma^{u})\big( \frac{w}{\rho-g +\lambda} - \frac{(\rho-g) \phi_{U}}{\rho-g +\lambda}\big).
\end{split}
\end{equation*}

Since $(\rho-g)\phi_{U}=b + f(\theta)(\phi_{L}-\phi_{U})$ from \eqref{eq:un_value} and $(\phi_{L}-\phi_{U})= \Gamma^{u}(\pi^{L}/L)/(1-\Gamma^{u})$ from the Nash bargaining rule, then $w^{u}$ satisfies \eqref{eq:union_wages}.

 \subsubsection{Wage premium}\label{subappendix:wage_premium}
Combining \eqref{eq:ind_wage} and \eqref{eq:union_wages}, we can  express the \emph{wage premium} from collective bargaining as follows:

\begin{equation}\label{eq:wage_gain_collective}
w^{u}_{t} - w^{n}_{t} = \big(\Psi^{u}_{t}  - \Psi^{n}_{t}\big)\big( y_{L_{t}}  - b_{t}\big) +  \frac{\Psi^{u}_{t}(\rho +  \alpha\;  \dot{m}^{*}_{t}  -g +\lambda_{t})}{\rho+\alpha\;  \dot{m}^{*}_{t} -g}\big(\hat{y}_{t}- y_{L_{t}}\big) > 0
\end{equation}

The first term in the right-hand side of equation \eqref{eq:wage_gain_collective} is expected to be positive given that the  power of labor will generally be higher under collective bargaining. The second term shows that,  even if $\Psi^{n}_{t}=\Psi^{u}_{t}$, workers can still get higher  real wages as a result of  an increase in the aggregate surplus.

  \subsubsection*{Proof of Proposition \ref{prop:equil_labor} }\label{subappendix:equil_labor} Here I only  present a heuristic proof to Proposition \ref{prop:equil_labor}. Note that  using \eqref{eq:online_marginal},   the rate of return can be defined as map
  
  \begin{equation}\label{eq:contraction_map}
  \mu^{j+1} = \Phi(\mu^{j})= \frac{\hat{y}(\hat{k}(\mu^{j}))- \hat{k}(\mu^{j}) \hat{y}_{\hat{k}}(\mu^{j}) }{w^{*}(\theta(\hat{k}(\mu^{j}))} - 1 . 
  \end{equation}

Use the following recursive argument:

\begin{itemize}
\item Guess some initial $\mu^{0} > 0$. 

\item For $ j \in \{0, ...\}:$
\item Define $\tilde{f}(\hat{k}(\mu))=\hat{y}_{\hat{k}} - \delta (1+\mu) D(h)^{a_{0}}/A^{k}$. From the intermediate value theorem, for any given   $\mu^{j}$, there exists a $\hat{k}(\mu^{j})$ such that  $\tilde{f}(\hat{k}(\mu^{j}))=0$. This is true because $\hat{y}_{\hat{k}} $ is continuously differentiable; see \eqref{eq:online_marginal}. 
\item Suppose that $P(\mathcal{U}=1|\cdot) \in (0,1)$ is given and that there is a finite and non-negative solution to $\{w^{u}, \theta^{u}\}$ and  $\{w^{n}, \theta^{n}\}$. Let $w^{d}$ denote the labor demand equation from \eqref{eq:value_job_firm} and $w^{s}$ be the labor supply equation from \eqref{eq:wage_agg}.  Define  $\tilde{g}(\theta(\hat{k}(\mu)) )=w^{s}-w^{d}$ and note that by the properties of the matching function $q(\theta)$ we can employ the intermediate value theorem to show that there is a $\theta(\hat{k}(\mu^{j})))$ such that $\tilde{g}(\cdot)=0$. 
\item  Set $w^{*}=w^{d}(\theta(\hat{k}(\mu^{j}))))$.

 \item Define $\mu^{j+1}$ using \eqref{eq:contraction_map}  until convergence. 

\end{itemize}

For suitable values of $\mu$,  $\Phi(\cdot)$ is a contraction, so there is a unique $\mu^{*}$ associated with the equilibrium in the labor market.  

\subsection{Section \ref{sec:equilibrium_dynamics}} 

This section presents the  proof of Lemma \ref{lemma:lemma_2A_AR} in Appendix \ref{subappendix:Automation_regions} and the   the proofs of Propositions \ref{prop:gen_equilibrium} and \ref{prop:comp_stat} in Section \ref{sec:equilibrium_dynamics}.

\subsubsection{Proof of Proposition \ref{prop:gen_equilibrium}}\label{subappendix:gen_equil} Suppose that Proposition \ref{prop:equil_labor} holds, such that we can build a recursive proof where $\mu^{*}$ and $\hat{k}^{*}$ can be taken as given. From this we can start with the value functions derived in subsection \ref{subsection:value_functions} and for ease of notation set $\tilde{\rho}=\rho +\alpha \dot{m}^{*} -g$, such that

\begin{equation}\label{eq:value_functions_appendix}
\begin{split}
\tilde{\rho} \phi_{L} &=  w + \lambda (\phi_{U}-\phi_{L}) + \dot{\phi}_{L}\\
\tilde{\rho} \phi_{U} &= b + f(\theta) (\phi_{L}-\phi_{U}) + \dot{\phi}_{U}\\
\tilde{\rho} \pi_{L} &= y_{L} -  w + -\lambda \pi_{L} + \dot{\pi}_{L}\\
\tilde{\rho} \pi^{L} &= \dot{\pi}^{L} + y - wL - \lambda \pi_{L} L.
\end{split}
\end{equation}

The last equation in \eqref{eq:value_functions_appendix} follows the argument described in the Proof of Proposition \ref{prop:non_union_wage}. 

To save notation, denote $P^{U} \equiv P(\mathcal{U}=1|\cdot)$ and define  the total surplus in the bargaining of wages as  

\begin{equation*}
\bar{S} = P^{U} \frac{S^{u}}{L} + (1-P^{U}) S^{n}
\end{equation*}

Recalling that $S^{n} = \pi_{L}+\phi_{L}-\phi_{U}$ and $S^{u} = L(\phi_{L}-\phi_{U}) + \pi^{L}$, we have that

\begin{equation*}
\dot{\bar{S}} = P^{U} \big[\dot{\phi}_{L}-\dot{\phi}_{U} +\dot{\pi}^{L}/L - \frac{\dot{L}}{L} \pi^{L}/L\big] +(1-P^{U}) \big[\dot{\phi}_{L}-\dot{\phi}_{U} +\dot{\pi}_{L}\big]
\end{equation*}

Using the first three lines in \eqref{eq:value_functions_appendix}, 

\begin{equation*}
\dot{\phi}_{L}-\dot{\phi}_{U} +\dot{\pi}_{L} = (\tilde{\rho}+\lambda) S^{n} - y_{L} + \tilde{\rho} \phi_{U} - \dot{\phi}_{U} = \dot{S}^{n}
\end{equation*}

Correspondingly, using the last two lines in \eqref{eq:value_functions_appendix},

\begin{equation*}
\dot{\big(S^{u}/L\big)} = (\tilde{\rho}+\lambda) \frac{S^{u}}{L} - \hat{y} + \tilde{\rho}\phi_{U} - \dot{\phi}_{U} - \frac{\lambda}{\tilde{\rho}} (\tilde{y} +\tilde{\pi}^{L}) - \frac{\dot{L}}{L} \frac{\pi^{L}}{L}.
\end{equation*}

Where $\tilde{y} \equiv \hat{y} - y_{L} = \hat{k}\hat{y}_{\hat{k}}$ and $\tilde{\pi}^{L} = \dot{\pi}^{L}/L - \dot{\pi}_{L}$. The last equation follows by noting that $\tilde{\rho} \pi^{L}/L = \tilde{y} + \tilde{\pi}^{L} + \tilde{\rho} \pi_{L}$. 

Using the sharing rule from the individual bargaining protocol, we have that $\pi_{L}=(1-\Gamma^{n})S^{n}= \xi/q(\theta)$.\footnote{Recall that $\Gamma^{n}$ is the intrinsic bargaining power of labor under individual bargaining and is given by $\Gamma^{n} = \big(T^{w} \Gamma^{nb} + \theta \Gamma^{na}\big)/(T^{w}+\theta)$.} Thus, 

\begin{equation*}
\dot{S}^{n} = \frac{\xi \dot{\theta} }{(1-\Gamma^{n}) q(\theta)} \Bigg[ \frac{\partial \Gamma^{n}/\partial \theta }{(1-\Gamma^{n})} - \frac{q'(\theta)}{q(\theta)}\Bigg].
\end{equation*}

Where $\partial \Gamma^{n}/\partial \theta >0$ by Corollary \ref{coro:barg_power}. Using the previous equation and the second line in \eqref{eq:value_functions_appendix},  it follows that 

\begin{equation}\label{eq:dyn_theta_individual}
\frac{\xi \dot{\theta} }{(1-\Gamma^{n}) q(\theta)} \Bigg[ \frac{\partial \Gamma^{n}/\partial \theta }{(1-\Gamma^{n})} - \frac{q'(\theta)}{q(\theta)}\Bigg] = \frac{\tilde{\rho}+\lambda + \Gamma^{n} f(\theta)}{1-\Gamma^{n}} \frac{\xi}{q(\theta)} - y_{L} + b.
\end{equation}

Repeating a similar excercise for the collective bargaining protocol, we have that $\pi^{L}/L=(1-\Gamma^{u})S^{u}/L=\pi_{L}+ \tilde{\rho}^{-1}(\tilde{y}+\tilde{\pi}^{L}) = \frac{\xi}{q(\theta)} + \tilde{\rho}^{-1}(\tilde{y}+\tilde{\pi}^{L})$. Thus, given that $\hat{k}=\hat{k}^{*}$ by assumption,

\begin{equation*}
\frac{\dot{S}^{u}}{L} - \frac{\dot{L}}{L} \frac{S^{u}}{L} = \frac{-\xi \dot{\theta} q'(\theta)}{(1-\Gamma^{u}) q(\theta)^{2}} 
\end{equation*}

Using the previous expression of $\dot{\big(S^{u}/L\big)}$ and setting $\dot{L}/L = \frac{\partial L}{\partial \theta} \dot{\theta}/L$, we obtain

\begin{equation}\label{eq:dyn_theta_collective}
\begin{split}
\frac{\xi \dot{\theta}}{(1-\Gamma^{u})q(\theta)} & \Big[ \frac{(1-\Gamma^{u})\partial L/\partial \theta}{L}\Big( 1 + \frac{q(\theta)(\tilde{y}+\tilde{\pi}^{L})}{\tilde{\rho} \xi}\Big)  -  \frac{ q'(\theta)}{q(\theta)}\Big]  = \frac{\tilde{\rho} + \lambda + \Gamma^{u} f(\theta)}{1-\Gamma^{u}}\times \\
&  \Big( \frac{\xi}{q(\theta)} + \frac{\tilde{y}+\tilde{\pi}^{L}}{\tilde{\rho}}\Big) - \hat{y} +b 
\end{split}
\end{equation}

Calculating $\dot{\bar{S}}= P^{U} \dot{\big(S^{u}/L\big)} + (1-P^{U}) \dot{S}^{n}$  and assuming that near the equilibrium $\tilde{\pi}^{L} =0$, we arrive at

\begin{equation*}
\begin{split}
 &\Bigg\{\frac{q'(\theta) \xi \big( 1- \Gamma^{u} + P^{U}(\Gamma^{u}-\Gamma^{n})\big)}{q(\theta)^{2} (1-\Gamma^{n})(1-\Gamma^{u})} - \frac{P^{U} \partial L/\partial \theta}{L} \Big( \frac{\xi}{q(\theta)} + \frac{\tilde{y}}{\tilde{\rho}}\Big) - \frac{(1-P^{U})\xi \partial \Gamma^{n}/\partial \theta  }{q(\theta)(1-\Gamma^{n})^2}\Bigg\} \dot{\theta} \\
 & - y_{L} + b + \frac{\xi}{q(\theta)} \frac{(\tilde{\rho}+\lambda + \Gamma^{n} f(\theta))(1-\Gamma^{u}) + P^{U} (\tilde{\rho}+\lambda + f(\theta))(\Gamma^{u}-\Gamma^{n})}{ (1-\Gamma^{n})(1-\Gamma^{u})} \\
 -& P^{U}  \frac{\tilde{y} \big(1-\tilde{\rho} - \lambda - \Gamma^{u}(1+f(\theta))\big))}{1-\Gamma^{u}} = 0  
 \end{split}
\end{equation*}

Multiplying the previous equation by $q(\theta)^2  (1-\Gamma^{n})(1-\Gamma^{u})$ and differentiating with respect to $\dot{\theta}$ and $\theta$ in the neighborhood of $(\dot{\theta}=0, \theta=\theta^{*})$ yields

\begin{equation}\label{eq:dyn_theta_synthesis}
a(\theta^{*}) \mathrm{d}\dot{\theta} + b(\theta^{*})\mathrm{d}\theta =0
\end{equation}

where

\begin{equation*}
\begin{split}
& a(\theta^{*}) = q'(\theta^{*})\xi \big( 1- \Gamma^{u} + P^{U}(\Gamma^{u}-\Gamma^{n})\big)- \frac{ (1-\Gamma^{n})(1-\Gamma^{u}) P^{U} \partial L/\partial \theta}{L} \times \\
& \Big( \xi q(\theta^{*}) + \frac{q(\theta^{*}) \tilde{y}}{\tilde{\rho}}\Big) - \frac{(1-P^{U})\xi q(\theta^{*})(1-\Gamma^{u}) \partial \Gamma^{n}/\partial \theta }{(1-\Gamma^{n})} < 0,\\
& b(\theta^{*}) = -2(1-\Gamma^{u})(1-\Gamma^{n})q(\theta^{*})(y_{L}-b)q'(\theta^{*}) + \xi q(\theta^{*})f'(\theta^{*}) \big[ \Gamma^{n}(1-\Gamma^{u})+\\
&  P^{U}(\Gamma^{u}-\Gamma^{n})\big] + \xi q'(\theta^{*})\Big[(\tilde{\rho}+\lambda + \Gamma^{n} f(\theta^{*}) ) (1-\Gamma^{u})+ P^{U}\big((\tilde{\rho}+\lambda+f(\theta^{*}))(\Gamma^{u}-\Gamma^{n})\big)\Big] \\
& -2 q(\theta^{*})P^{U}\tilde{y}\big[1-\tilde{\rho} -\lambda -\Gamma^{u}(1+f(\theta^{*}))\big] (1-\Gamma^{n})q'(\theta^{*}) + q(\theta^{*})^{2} P^{U} \tilde{y} \Gamma^{u} (1-\Gamma^{b}) f'(\theta^{*}) \\
& +  (1-\Gamma^{u})q(\theta^{*})^{2} (y_{L}-b) \frac{\partial \Gamma^{n}}{\partial \theta} +\xi q(\theta^{*}) \big[f(\theta^{*}) (1-P^{U}) -P^{U} (\tilde{\rho}+\lambda)  \Gamma^{n}\big]\frac{\partial \Gamma^{n}}{\partial \theta}  \\
& + q(\theta^{*})^{2} P^{U} \tilde{y} \big[ 1-\tilde{\rho}+\lambda - \Gamma^{u}(1+f(\theta^{*}))\big] \frac{\partial \Gamma^{n}}{\partial \theta}  > 0. 
\end{split}
\end{equation*}

The sign of $a(\theta^{*})$ is straightforward. Respectively, the sign of $b(\theta^{*}) $ is guaranteed to be positive by  the  labor market equilibrium condition (i.e., by combining \eqref{eq:value_job_firm} and  \eqref{eq:wage_agg}). 

Setting $\mathrm{d}\dot{\theta}=\dot{\theta}$ and $\mathrm{d}\theta = \theta-\theta^{*}$, we have

\begin{equation*}
\theta = \mathcal{B} e^{-c(\theta^{*}) t} + \theta^{*},
\end{equation*} 

where $\mathcal{B}$ is a constant and $c(\theta^{*})=b(\theta^{*})/a(\theta^{*})<0$. Thus, the unique stable path of $\theta$ corresponds to $\mathcal{B}=0$, meaning that $\theta$ will always jump immediately to the stationary value.

Thus, given \eqref{eq:dyn_theta_synthesis}, the system can be characterized near the equilibrium  by the following system of differential equations (for simplicity I am setting $q_{t}=q$). 

\begin{equation}\label{eq:syst_differential_eq}
\begin{split}
&\dot{\theta}_{t} + c(\theta^{*}) (\theta_{t}-\theta^{*})=0\\ 
& \dot{U}_{t} = \lambda_{t} (1-U_{t}) - f(\theta_{t})U_{t}\\
& \dot{\hat{k}}_{t} = q \big[\hat{y}_{t} - \hat{c}_{t} - w_{t} - \xi \frac{V_{t}}{L_{t}} - \frac{\tau}{L_{t}}\big] -(\delta_{t} + g_{t}) \hat{k}_{t} \\ 
& \dot{\hat{c}}_{t}= \frac{\hat{c}_{t}}{\epsilon}\big[ \hat{y}_{\hat{k}_{t}} q - \delta (1+\mu^{*}_{t})\big]
\end{split}
\end{equation}

Analyzing the equations in \eqref{eq:syst_differential_eq} as a recursive block where the first two equations define the equilibrium in the labor market and the remaining two equations derive the process of arbitrage---which takes the equilibrium in the labor market as given---then there is obviously a unique and locally stable solution. This is true because the last two equations are completely analogous to the traditional neoclassical growth model. 

However, the existence of  an equilibrium  BGP with positive cannot be taken as given. To see this we need to solve the equilibrium in the third equation of \eqref{eq:syst_differential_eq}, from which it follows that

\begin{equation*}
(\delta+ g) \frac{\hat{k}}{q} = \hat{y}- \hat{c} - w - \xi \frac{V}{L} - \frac{\tau}{L}
\end{equation*}

Be definition 

\begin{equation*}
\mu = \frac{\text{Profits}}{\text{Costs of production}}= \frac{\hat{y}-\delta \hat{k}/q - w}{\delta \hat{k}/q + w}.
\end{equation*}

Given that the aggregate production function has constant returns to scale,

\begin{equation*}
g \frac{\hat{k}}{q}  = \frac{\mu \hat{y}}{1+\mu} - \hat{c} - \hat{\xi} V - \hat{\tau}.
\end{equation*}

Where $\hat{\xi}=\xi/L$ and $\hat{\tau}=\tau/L$. Dividing both sides by $\hat{k}/q$, we have that

\begin{equation*}
g = r - \frac{\hat{c}q}{\hat{k}} - \chi,
\end{equation*}

since

\begin{equation*}
r =\frac{\text{Profits}}{\text{Value of the capital stock}}= \frac{\mu Y}{P^{k} K} = \frac{\mu \hat{y} q}{(1+\mu) \hat{k}}
\end{equation*}

From this it follows that, if $\hat{c} \geq 0$,  there exist a number $s \in (0,1]$ such that

\begin{equation}\label{eq:harrod_cond_appendix}
g = s(r-\chi),
\end{equation}

which is equation \eqref{eq:harrod_growth} in the main text. Now, an additional problem in the system is that a steady-state growth requires a surplus sufficiently large to sustain the expansion of capital and the payment of vacancy expenses and taxes. The minimum surplus capable of guaranteeing this condition can be found by setting $s=1$ (or $\hat{c}=0$). Using \eqref{eq:harrod_cond_appendix}, it follows that

\begin{equation}\label{eq:min_surplus2}
\mu^{\text{min}} = \Big[1 + \frac{q\hat{y}}{\hat{k}(g + \chi)}\Big]^{-1}.
\end{equation}

The key implication of \eqref{eq:min_surplus2} is that  higher growth rates,   higher vacancy expenses, or higher taxes raise the minimum rate of return of capital, meaning that---unless the system can simultaneously increase the equilibrium surplus---the economy can become unsustainable if $g$ or $\chi$ are sufficiently high. More precisely, we require

\begin{equation*}
\mu > \frac{g}{\delta} > \mu^{\text{min}}
\end{equation*}

to ensure the transversality and sustainability condition.

\subsubsection{Proof of Proposition \ref{prop:comp_stat}}\label{subappendix:comp_stat}

\emph{(Automation)} Suppose that the economy is initially in a BGP with  $\mu > g/\delta >\mu^{\text{min}}$, $|\partial \lambda/\partial \dot{m}|>\alpha$ and $m > \mathrm{max}\{\bar{m}, \tilde{m}\}$. 

\textbf{(Stage 1)} At stage 1, the negative shock on $m$ implies $\dot{m}<0$. Starting with the labor market, we have that (recall that $\tilde{\rho}=\rho + \alpha \dot{m} - g$):

\begin{equation*}
\frac{\partial \Psi^{na}}{\partial \dot{m}}= \frac{\alpha + \partial \lambda/\partial \dot{m}}{\tilde{\rho}+\lambda + \Gamma^{na}f(\theta)} \big[ \Gamma^{na}-\Psi^{na}\big] >0 
\end{equation*}

and

\begin{equation*}
\frac{\partial \Psi^{nb}}{\partial \dot{m}}= \frac{\alpha + \partial \lambda/\partial \dot{m}}{\tilde{\rho}+\lambda + \Gamma^{nb}f(\theta)} \big[ \Gamma^{nb}-\Psi^{nb}\big] >0 
\end{equation*}

implies that $\partial \Psi^{n}/\partial \dot{m} > 0$. Similarly, 

\begin{equation*}
\frac{\partial \Psi^{u}}{\partial \dot{m}}= \frac{\alpha + \partial \lambda/\partial \dot{m}}{\tilde{\rho}+\lambda + \Gamma^{u}f(\theta)} \big[ \Gamma^{u}-\Psi^{u}\big] >0. 
\end{equation*}

Thus, the labor supply equation moves down with $\dot{m}<0$. Using \eqref{eq:value_job_firm}, we have that the labor demand equation changes according to

\begin{equation*}
\frac{\partial w^{d}}{\partial \dot{m}} = - \big(\alpha + \frac{\partial \lambda}{\partial \dot{m}}\big) \frac{\xi}{q(\theta)}>0
\end{equation*}

As consequence, $\dot{m}<0$ implies a lower labor demand schedule. The result is a reduction in the equilibrium stationary real wage, though it is generally not possible to determine the change in $\theta$.  

Using the second line in \eqref{eq:syst_differential_eq}, we have that the steady-state unemployment $U^{*}_{t}=\lambda_{t}/(\lambda_{t}+f(\theta))$ changes in the following direction

\begin{equation*}
\frac{\partial U^{*}}{\partial \dot{m}} = \frac{\partial \lambda/\partial \dot{m}}{\lambda+f(\theta)} \frac{f(\theta)}{\lambda+f(\theta)}<0.
\end{equation*} 

Thus, a lower $\dot{m}$ increases $U^{*}$. A diagrammatic representation of the resulting changes in the labor market are represented in Figure \ref{fig:labor_market_transdyn_appendix}.The decrease in $\dot{m}$ moves $\dot{U}$ to the right and, though it is generally not possible to anticipate how $\theta$ will change (hence the area in teal), the new equilibrium will likely  result in a higher rate of unemployment and vacancy rates. 

Moving now to the conditions of capital arbitrage, we have that $\dot{\hat{c}}=0$  when (see equation \eqref{eq:euler})

\begin{equation*}
\hat{y}_{\hat{k}} = \frac{\delta}{q} (1+\mu)
\end{equation*}

Thus, since $\mu^{*}$ moves up when $w \downarrow$, a lower $\dot{m}$ will lead to a lower equilibrium $\hat{k}^{*}$. Respectively, using the third line in \eqref{eq:syst_differential_eq}, it follows that when $\dot{\hat{k}}=0$, 

\begin{equation*}
\hat{c} = g \frac{\hat{k}}{q} + \frac{\mu \hat{y}}{1+\mu} - \hat{\xi} V - \hat{\tau}. 
\end{equation*} 

The effects on  $\hat{c}$ are ambiguous because a higher $\mu$ raises $\hat{c}$, but a lower $L$ increases $\hat{\tau}$ and $\hat{\xi}$. In all generality, it is most likely that the curve $\dot{\hat{k}}=0$ will remain more or less constant. 

The final effects from stage 1 (as illustrated in Figure \ref{fig:capital_market_trans_automation})  are consequently an initial increase in consumption, followed a decline in $\hat{k}$ and $\hat{c}$.

            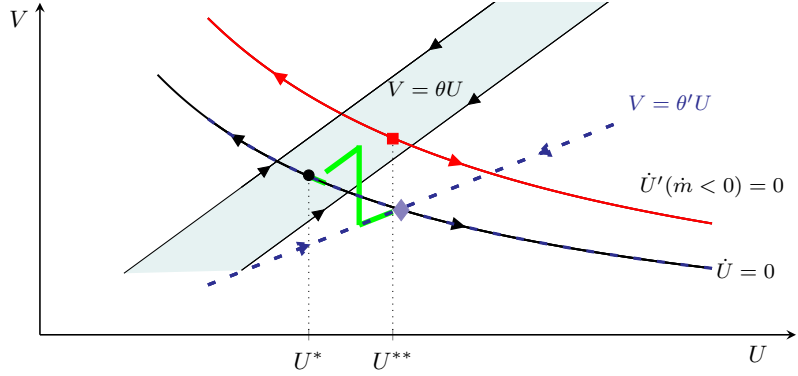
\begin{figure}
\begin{center}
\begin{tikzpicture}
     \begin{axis}[
        axis x line=middle,
                    axis y line=middle,
         name=plot1,
          title={\footnotesize }, 
              title style={at={(axis description cs:0.3,0.98)}, anchor=south} , 
 width=0.8*\textwidth, height=6 cm,
ymin=0.1, ymax = 1.099,
xmin=0, xmax=4.5,
xshift=-0.25cm,
 xlabel = {\scriptsize  $U$},
    ylabel = {\scriptsize $V$},
 ytick = \empty, 
  xtick = {1.6,2.1},
  xticklabels={$U^{*}$,$U^{**}$},
        every axis x label/.style={
    at={(ticklabel* cs:0.95)},
    anchor=north,
},
every axis y label/.style={
    at={(ticklabel* cs:0.95)},
    anchor=east,
}
]

    \addplot[domain=0.7:4,
  line width=0.75pt,
black
]
{1.5/(1+x) }node at (axis cs:4.2,0.3) {\tiny $\dot{U}=0$};

    \addplot[domain=1:4,
  line width=0.75pt,
black
]
{2/(1+0.9*x) }node at (axis cs:4,0.55) {\tiny $\dot{U}'(\dot{m}<0)=0$};

\addplot [name path = A,
    domain = 1.2:3.5] {0.37*x-0.15} node[
    currarrow,
    pos= 0.2, 
    xscale=1,
    sloped,
    scale=1]{};
    
    \addplot [
    domain = 1.2:3.5] {0.37*x-0.15} node[
    currarrow,
    pos= 0.6, 
    xscale=-1,
    sloped,
    scale=1]{};
    
        \addplot [
    domain = 1.2:3.5] {0.37*x+0.1}  node[
    currarrow,
    pos= 0.5, 
    xscale=-1,
    sloped,
    scale=1]{};
    
            \addplot [
    domain = 1.2:3.5] {0.37*x+0.1}  node[
    currarrow,
    pos= 0.06, 
    xscale=1,
    sloped,
    scale=1]{};
    
        \addplot[domain=1:4,
  line width=0.75pt,
red
]
{2/(1+0.9*x) } node[
    currarrow,
    pos= 0.15, 
    xscale=-1,
    sloped,
    color=red,
    scale=1]{};
    
           \addplot[domain=1:4,
  line width=0.75pt,
red
]
{2/(1+0.9*x) } node[
    currarrow,
    pos= 0.5, 
    xscale=1,
    sloped,
    color=red,
    scale=1]{};
    
       \addplot[domain=1:4,
  line width=0.75pt,
black
]
{1.5/(1+x) } node[
    currarrow,
    pos= 0.06, 
    xscale=-1,
    sloped,
    scale=1]{};

     \addplot[domain=1.6:1.7,
  line width=2pt,
green
]
{1.5/(1+x) } node[
    currarrow,
    pos= 0.16, 
    xscale=1,
    sloped,
    scale=0]{};

        \addplot [line width=2pt,
    domain = 1.7:1.9,  green] {0.37*x-0.04}  node[
    currarrow,
    pos= 0.5, 
    xscale=-1,
    sloped,
    scale=0]{};

           \draw[line width=2pt,green] (1.9,0.67)--(1.9,0.43); 
    
        \addplot [line width=2pt,green, 
    domain = 1.9:2.1] {0.2*x+0.05} node[
    currarrow,
    pos= 0.8, 
    xscale=-1,
    sloped,
    color=Green,
    scale=0]{};

           \addplot[domain=1:4,
  line width=0.75pt,
black
]
{1.5/(1+x) } node[
    currarrow,
    pos= 0.5, 
    xscale=1,
    sloped,
    scale=1]{};

         \addplot[loosely dashed, domain=1:4,
  line width=1.2pt,
Blue
]
{1.5/(1+x) }; 

    \addplot [loosely dashed,   Blue,  line width=1.2pt, 
    domain = 1:3.5] {0.2*x+0.05} node[
    currarrow,
    pos= 0.8, 
    xscale=-1,
    sloped,
    color=Blue,
    scale=1]{}; 
    
    \addplot [loosely dashed,   Blue,  line width=1.2pt, 
    domain = 1:3.5] {0.2*x+0.05} node[
    pos= 1.1]{\tiny $V=\theta'U$};

        \addplot [loosely dashed,  Blue,  line width=1.2pt, 
    domain = 1:3.5] {0.2*x+0.05} node[
    currarrow,
    pos= 0.23, 
    xscale=1,
    sloped,
    color=Blue,
    scale=1]{};

\addplot [name path = B,
    domain = 0.5:3.5] {0.37*x+0.1} 
    node [pos=0.5, right] {\tiny $V=\theta U$};

\addplot [teal!10] fill between [of = A and B, soft clip={domain=0:4}] node [pos=0.5, left] {\tiny $V=\theta U$};

\addplot[color = black, dotted, thin] coordinates {(1.6, 0) (1.6, 0.58)};

\addplot[color = black, dotted, thin] coordinates {(2.1, 0) (2.1, 0.69)};

\addplot[color = black, mark = *, only marks, mark size = 2pt] coordinates {(1.6, 0.58)};
\addplot[color = red, mark = square*, only marks, mark size = 2pt] coordinates {(2.1, 0.69)};

\addplot[color = Blue!50, mark =diamond*, only marks, mark size = 4pt] coordinates {(2.15, 0.475)};

\end{axis}
    
\end{tikzpicture}
            \caption{Labor market transitional dynamics with higher automation.\label{fig:labor_market_transdyn_appendix}}
                \end{center}
            \end{figure}

\textbf{(Stage 2)} In stage 2, we have that $\dot{m}=0$ and we reach an automation measure  $m'<m$. Starting again in the labor market equilibrium, we have that the labor supply schedule returns to its initial position but the labor demand stays below it (see subsection \ref{subappendix:proof_lemma_automation} below). Thus, the initial equilibrium results in a lower $\theta$ and a lower $w$. These two effects explain the movements described by the dashed blue lines in  Figure \ref{fig:labor_market_transdyn_appendix}. First,  $\dot{U}=0$ returns to its initial position because $\dot{m}=0$. Second, the vacancy rate curve moves to the right given the lower equilibrium value of $\theta$. Thus, the end result is a lower vacancy rate and a higher rate of unemployment.

        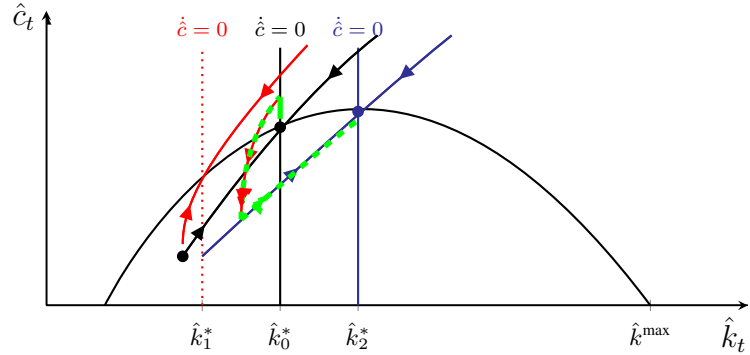
\begin{figure}
\begin{center}
\begin{tikzpicture}
     \begin{axis}[
        axis x line=middle,
                    axis y line=middle,
         name=plotramsey,
          title={}, 
              title style={at={(axis description cs:0.1,0.965)}, anchor=south} , 
 width=0.75\textwidth, height=5.5 cm,
ymin=0, ymax = 12,
xmin=0, xmax=18,
xshift=-0.5cm,
 xlabel = {$\hat{k}_{t}$},
    ylabel = {$\hat{c}_{t}$},
  xtick = {4,6,8,15.5},
  xticklabels={$\hat{k}^{*}_{1}$,$\hat{k}^{*}_{0}$, $\hat{k}^{*}_{2}$,$\hat{k}^{\text{max}}$},
 ytick = \empty, 
        every axis x label/.style={
    at={(ticklabel* cs:0.98)},
    anchor=north,
},
  dot/.style={circle,fill=black,minimum size=4pt,inner sep=0pt,
            outer sep=-0pt},
every axis y label/.style={
    at={(ticklabel* cs:0.98)},
    anchor=east,
}
]

\draw[thick,<->] (0,12) node[left]{$\hat{c}_{t}$}--(0,0) node[below right]{$0$}--(19,0) node[below]{$\hat{k}_{t}$};

\draw[thick,black] (6,0)--(6,10.5) node[above]{\tiny  $\dot{\hat{c}}=0$};

\draw[thick, red, dotted] (4,0)--(4,10.5) node[above]{\tiny  $\dot{\hat{c}}=0$};

\draw[thick, Blue] (8,0)--(8,10.5) node[above,Blue]{\tiny  $\dot{\hat{c}}=0$};

\draw[thick](1.5,0) ..controls (1.5,0) and (7.0,18) .. (15.5,0);    

         \draw[line width=2pt,green] (6,7.6)--(6,8.5);

\draw[thick,black](3.5,2) ..controls (3.5,2) and (6,7.95) .. (8.5,11) node[color=black,
	currarrow,
	pos=0.25, 
	xscale=1,
	sloped,
	scale=1]{};

	\draw[thick,black](3.5,2) ..controls (3.5,2) and (6,7.95) .. (8.5,11)node[
	currarrow,
	color=black,
	pos=0.85, 
	xscale=-1,
	sloped,
	scale=1]{};

	\draw[thick,Blue](4,2) ..controls (4,2) and (8,7.95) .. (10.4,11) node[color=Blue,
	currarrow,
	pos=0.5, 
	xscale=1,
	sloped,
	scale=1]{};
	
	\draw[thick,Blue](4,2) ..controls (4,2) and (8,7.95) .. (10.4,11)node[
	currarrow,
	color=Blue,
	pos=0.85, 
	xscale=-1,
	sloped,
	scale=1]{};

	\draw[thick,red](3.5,2.5) ..controls (3.5,4) and (4,6) .. (6.7,10.6) node[color=red,
	currarrow,
	pos=0.25, 
	xscale=1,
	sloped,
	scale=1]{};
	
	\draw[thick,red](3.5,2.5) ..controls (3.5,4) and (4,6) .. (6.7,10.6) node[color=red,
	currarrow,
	color=red,
	pos=0.85, 
	xscale=-1,
	sloped,
	scale=1]{};

		\draw[thick,red](5,3.5) ..controls (5,7) and (6,8.5) .. (6,8.5) node[color=red,
	currarrow,
	pos=0.3, 
	xscale=-1,
	sloped,
	scale=1]{};
	
		\draw[thick,red](5,3.5) ..controls (5,7) and (6,8.5) .. (6,8.5) node[color=red,
	currarrow,
	color=red,
	pos=0.1, 
	xscale=-1,
	sloped,
	scale=1]{};

		\draw[line width=2pt,dashed, green](5,3.5) ..controls (5,7) and (6,8.5) .. (6,8.5) node[color=red,
	currarrow,
	color=red,
	pos=0.1, 
	xscale=-1,
	sloped,
	scale=1]{};

		\draw[line width=2pt,dashed, green](5,3.5) ..controls (5,3.5) and (8,7.75) .. (8,7.5)node[
	currarrow,
	color=green,
	pos=0.25, 
	xscale=1,
	sloped,
	scale=1]{};

\addplot[color = black, mark = *, only marks, mark size = 2pt] coordinates {(6, 7.26)};
\addplot[color = black, mark = *, only marks, mark size = 2pt] coordinates {(3.5, 2)};
\addplot[color = Blue, mark = *, only marks, mark size = 2pt] coordinates {(8, 7.9)};


\end{axis}
    
\end{tikzpicture}
            \caption{Capital market transitional dynamics with an increase in automation.\label{fig:capital_market_trans_automation}}
                \end{center}
            \end{figure}

On the capital market, we have now that---though $\mu$ increases because $w$ goes down---the reduction in $m$ will end up raising $\hat{k}$ (see equation \eqref{eq:online_marginal} in the main text). Thus,  curve where $\dot{\hat{c}}=0$ moves to the right since the new equilibrium requires a higher $\dot{k}$. Respectively, since $\mu$ and $L$ again move in different directions, it is not clear whether $\dot{\hat{k}}=0$ moves up or down. However, in general, it will most likely stay relatively constant. 

The end results in the labor and capital market are illustrated by the green lines  in Figures \ref{fig:labor_market_transdyn_appendix} and \ref{fig:capital_market_trans_automation}.

\emph{(Labor-augmenting technical change)} Here we consider the effects of an increase in $\dot{M}$, which leads to a rise in $g$. As before, let us start  by studying the effects in the labor market with the assumption that $\sigma \in (0,1)$ and $g>0$. In this case, given Corollary \ref{coro:barg_power}, we have that 

\begin{equation*}
\frac{\partial \Psi^{n}}{\partial \dot{M}} > 0
\end{equation*}

which moves the labor supply equation upwards. The effects on the demand for labor are summarized by

\begin{equation*}
\frac{\partial w^{d}}{\partial \dot{M}} = \alpha - \frac{\partial \lambda}{\partial \dot{M} }>0
\end{equation*}

The sign is positive for the same reason reasons that  $\partial \Psi^{n}/\partial \dot{M} > 0$ (see the proof of Corollary \ref{coro:barg_power} above).  Summing up the effects on the demand and supply of labor, the result is a higher $w$---though with ambiguous effects on $\theta$.

            \begin{figure}
\begin{center}
\begin{tikzpicture}
     \begin{axis}[
        axis x line=middle,
                    axis y line=middle,
         name=plot1,
          title={\footnotesize }, 
              title style={at={(axis description cs:0.3,0.98)}, anchor=south} , 
 width=0.8*\textwidth, height=6 cm,
ymin=0.1, ymax = 1.099,
xmin=0, xmax=4.5,
xshift=-0.25cm,
 xlabel = {\scriptsize  $U$},
    ylabel = {\scriptsize $V$},
 ytick = \empty, 
  xtick = {1.6,2.1},
  xticklabels={$U^{*}$,$U^{**}$},
        every axis x label/.style={
    at={(ticklabel* cs:0.95)},
    anchor=north,
},
every axis y label/.style={
    at={(ticklabel* cs:0.95)},
    anchor=east,
}
]

    \addplot[domain=0.7:4,
  line width=0.75pt,
black
]
{1.5/(1+x) }node at (axis cs:4.2,0.3) {\tiny $\dot{U}=0$};

    \addplot[domain=1:4,
  line width=0.75pt,
red
]
{2.1/(1+0.9*x) }node at (axis cs:4,0.55) {\tiny $\dot{U}'(\uparrow \dot{M})=0$};

\addplot [name path = A,
    domain = 1.2:3.5] {0.37*x-0.15} node[
    currarrow,
    pos= 0.2, 
    xscale=1,
    sloped,
    scale=1]{};
    
    \addplot [
    domain = 1.2:3.5] {0.37*x-0.15} node[
    currarrow,
    pos= 0.6, 
    xscale=-1,
    sloped,
    scale=1]{};
    
        \addplot [
    domain = 1.2:3.5] {0.37*x+0.1}  node[
    currarrow,
    pos= 0.5, 
    xscale=-1,
    sloped,
    scale=1]{};
    
            \addplot [
    domain = 1.2:3.5] {0.37*x+0.1}  node[
    currarrow,
    pos= 0.06, 
    xscale=1,
    sloped,
    scale=1]{};
    
        \addplot[domain=1:4,
  line width=0.75pt,
red
]
{2.1/(1+0.9*x) } node[
    currarrow,
    pos= 0.15, 
    xscale=-1,
    sloped,
    color=red,
    scale=1]{};
    
           \addplot[domain=1:4,
  line width=0.75pt,
red
]
{2.1/(1+0.9*x) } node[
    currarrow,
    pos= 0.5, 
    xscale=1,
    sloped,
    color=red,
    scale=1]{};

       \addplot[domain=1:4,
  line width=0.75pt,
black
]
{1.5/(1+x) } node[
    currarrow,
    pos= 0.06, 
    xscale=-1,
    sloped,
    scale=1]{};

     \addplot[domain=1.6:1.7,
  line width=2pt,
green
]
{1.5/(1+x) } node[
    currarrow,
    pos= 0.16, 
    xscale=1,
    sloped,
    scale=0]{};

        \addplot [line width=2pt,
    domain = 1.6:2.1,  green] {0.37*x-0.04}  node[
    currarrow,
    pos= 0.5, 
    xscale=1,
    sloped,
    color=green, 
    scale=1]{};

           \addplot[domain=1:4,
  line width=0.75pt,
black
]
{1.5/(1+x) } node[
    currarrow,
    pos= 0.5, 
    xscale=1,
    sloped,
    scale=1]{};

\addplot [name path = B,
    domain = 0.5:3.5] {0.37*x+0.1} 
    node [pos=0.5, right] {\tiny $V=\theta U$};

\addplot [teal!10] fill between [of = A and B, soft clip={domain=0:4}] node [pos=0.5, left] {\tiny $V=\theta U$};

\addplot[color = black, dotted, thin] coordinates {(1.6, 0) (1.6, 0.58)};

\addplot[color = black, dotted, thin] coordinates {(2.1, 0) (2.1, 0.69)};

\addplot[color = black, mark = *, only marks, mark size = 2pt] coordinates {(1.6, 0.58)};

\addplot[color = red, mark = square*, only marks, mark size = 2pt] coordinates {(2.1, 0.72)};

\end{axis}
    
\end{tikzpicture}
            \caption{Labor market transitional dynamics with higher $\dot{M}$ and $\sigma<1$.\label{fig:labor_market_laboraugmenting}}
                \end{center}
            \end{figure}
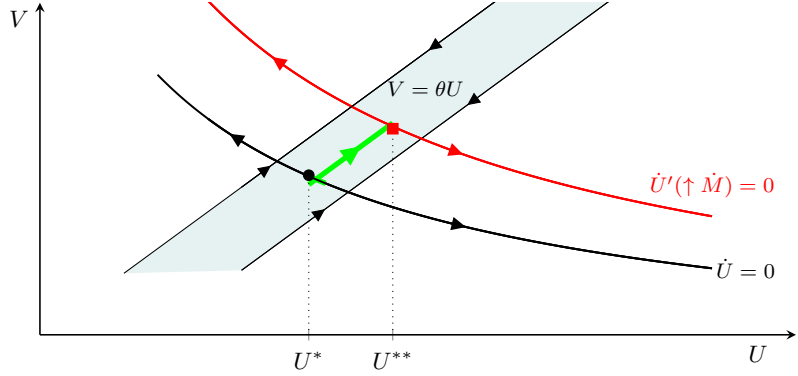

The changes in the equilibrium rate of unemployment are given by

\begin{equation*}
\frac{\partial U}{\partial \dot{M}}= \frac{\partial \lambda/\partial \dot{M} }{\lambda+f(\theta)} \frac{f(\theta)}{\lambda+f(\theta)}>0
\end{equation*}

This is explained by Lemma \ref{lemma:tech_unemployment}.  The end results in the labor market are summarized  in Figure \ref{fig:labor_market_laboraugmenting}.

Moving to the capital market, the higher $w$ leads to a lower $\mu$. From this, using the stationary condition in the Euler equation \eqref{eq:euler} in the main text, we arrive at a new equilibrium with a higher $\hat{k}$. Respectively, using the condition that $\dot{\hat{k}}=0$ in the thrid line of \eqref{eq:syst_differential_eq}, it follows that $\hat{c}$ decreases given the lower $\mu$ and $L$. The resulting transition dynamics are summarized in Figure \ref{fig:capital_market_trans_labtech}.

        \begin{figure}
\begin{center}
\begin{tikzpicture}
     \begin{axis}[
        axis x line=middle,
                    axis y line=middle,
         name=plotramsey,
          title={}, 
              title style={at={(axis description cs:0.1,0.965)}, anchor=south} , 
 width=0.75\textwidth, height=5.5 cm,
ymin=0, ymax = 12,
xmin=0, xmax=18,
xshift=-0.5cm,
 xlabel = {$\hat{k}_{t}$},
    ylabel = {$\hat{c}_{t}$},
  xtick = {6,8,15.5},
  xticklabels={$\hat{k}^{*}_{0}$, $\hat{k}^{*}_{1}$,$\hat{k}^{\text{max}}$},
 ytick = \empty, 
        every axis x label/.style={
    at={(ticklabel* cs:0.98)},
    anchor=north,
},
  dot/.style={circle,fill=black,minimum size=4pt,inner sep=0pt,
            outer sep=-0pt},
every axis y label/.style={
    at={(ticklabel* cs:0.98)},
    anchor=east,
}
]

\draw[thick,<->] (0,12) node[left]{$\hat{c}_{t}$}--(0,0) node[below right]{$0$}--(19,0) node[below]{$\hat{k}_{t}$};

\draw[thick,black] (6,0)--(6,10.5) node[above ]{\tiny  $\dot{\hat{c}}=0$};

\draw[thick, red] (8,0)--(8,10.5) node[above,xshift=1em, red]{\tiny  $\dot{\hat{c}}(\uparrow \dot{M})=0$};

\draw[thick, red](2.5,0) ..controls (2.5,0) and (7.0,15) .. (14,0);    

\draw[thick](1.5,0) ..controls (1.5,0) and (7.0,18) .. (15.5,0);

         \draw[line width=2pt,dashed, green] (6,7.6)--(6,4);

\draw[thick,black](3.5,2) ..controls (3.5,2) and (6,7.95) .. (8.5,11) node[color=black,
	currarrow,
	pos=0.35, 
	xscale=1,
	sloped,
	scale=1]{};

	\draw[thick,black](3.5,2) ..controls (3.5,2) and (6,7.95) .. (8.5,11)node[
	currarrow,
	color=black,
	pos=0.85, 
	xscale=-1,
	sloped,
	scale=1]{};

	\draw[thick,red](4,1) ..controls (4,1) and (7,5.95) .. (10,9) node[color=red,
	currarrow,
	pos=0.65, 
	xscale=1,
	sloped,
	scale=1]{};
	
	\draw[thick,red](4,1) ..controls (4,1) and (7,5.95) .. (10,9) node[color=red,
	currarrow,
	color=red,
	pos=0.85, 
	xscale=-1,
	sloped,
	scale=1]{};

		\draw[line width=2pt,dashed, green](6,4) ..controls (7,5.8) and (8,6.8) .. (8,6.8) node[color=red,
	currarrow,
	color=red,
	pos=0.1, 
	xscale=-1,
	sloped,
	scale=0]{};

\addplot[color = black, mark = *, only marks, mark size = 2pt] coordinates {(6, 7.26)};
\addplot[color = black, mark = *, only marks, mark size = 2pt] coordinates {(3.5, 2)};
\addplot[color = red, mark = *, only marks, mark size = 2pt] coordinates {(8, 6.8)};


\end{axis}
    
\end{tikzpicture}
            \caption{Capital market transitional dynamics with an increase in $\dot{M}$.\label{fig:capital_market_trans_labtech}}
                \end{center}
            \end{figure}
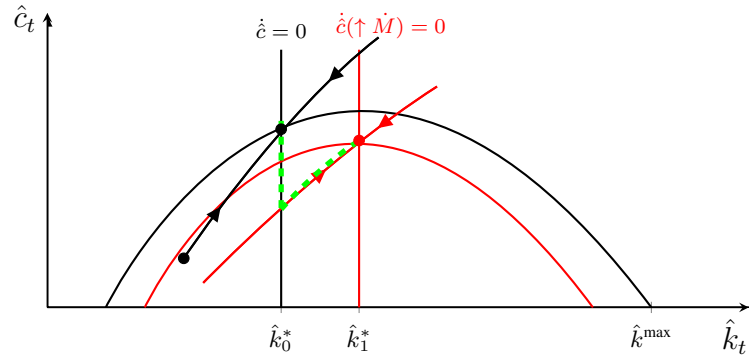

            The case where $\sigma >1$ is a straightforward application of the previous argument.

 \emph{(Labor Institutions)}  Without loss of generality, we can focus on variations in $T^{w}$ as a representation of variations in the institutional support to labor. Remembering that a higher $T^{w}$ can be read as a lower support to labor, we have from Corollary \ref{coro:barg_power} that 
 
 \begin{equation*}
 \frac{\partial \Psi^{n}}{\partial T^{w}} <0
 \end{equation*}
            
            which moves the labor supply schedule downwards. In the paper, given that we provided no link between worker power and labor productivity, there are no changes in the labor supply schedule. Thus, the end result is a lower stationary real wage and a higher value of $\theta$.  The resulting changes in the labor market  are summarized in Figure

                  \begin{figure}
\begin{center}
\begin{tikzpicture}
     \begin{axis}[
        axis x line=middle,
                    axis y line=middle,
         name=plot1,
          title={\footnotesize }, 
              title style={at={(axis description cs:0.3,0.98)}, anchor=south} , 
 width=0.8*\textwidth, height=6 cm,
ymin=0.1, ymax = 1.099,
xmin=0, xmax=4,
xshift=-0.25cm,
 xlabel = {\scriptsize  $U$},
    ylabel = {\scriptsize $V$},
 ytick = \empty, 
  xtick = {1,1.6},
  xticklabels={$U^{*}$,$U^{**}$},
        every axis x label/.style={
    at={(ticklabel* cs:0.95)},
    anchor=north,
},
every axis y label/.style={
    at={(ticklabel* cs:0.95)},
    anchor=east,
}
]

    \addplot[domain=0.5:3.2,
  line width=0.75pt,
black
]
{1.5/(1+x) }node at (axis cs:3.5,0.42) {\tiny $\dot{U}=0$};

\addplot [name path = A,
    domain = 0:3.5] {0.35*x+0.05} node[
    currarrow,
    pos= 0.2, 
    xscale=1,
    sloped,
    scale=1]{};
    
    \addplot [
    domain = 0:3.5] {0.35*x+0.05} node[
    currarrow,
    pos= 0.6, 
    xscale=-1,
    sloped,
    scale=1]{};

        \addplot [red, 
    domain = 0:3.5] {0.7*x+0.05}  node[
    currarrow,
    pos= 0.35, 
    xscale=-1,
    sloped,
    color=red,
    scale=1]{};
    
            \addplot [red,
    domain = 0:3.5] {0.7*x+0.05}  node[
    currarrow,
    pos= 0.16, 
    xscale=1,
    sloped,
        color=red,
    scale=1]{};

       \addplot[domain=1:1.6,
  line width=2pt,
green
]
{1.5/(1+x) } node[
    currarrow,
    pos= 0.5, 
    xscale=-1,
    sloped,
    color=green,
    scale=1]{};

       \addplot[domain=0.5:4,
  line width=0.75pt,
black
]
{1.5/(1+x) } node[
    currarrow,
    pos= 0.06, 
    xscale=-1,
    sloped,
    scale=1]{};

           \addplot[domain=0.5:4,
  line width=0.75pt,
black
]
{1.5/(1+x) } node[
    currarrow,
    pos= 0.5, 
    xscale=1,
    sloped,
    scale=1]{};

\addplot[color = black, dotted, thin] coordinates {(1.6, 0) (1.6, 0.58)};

\addplot[color = black, dotted, thin] coordinates {(1, 0) (1, 0.77)};

\addplot[color = black, mark = *, only marks, mark size = 2pt] coordinates {(1.6, 0.58)};

\addplot[color = red, mark = square*, only marks, mark size = 2pt] coordinates {(1, 0.77)};

\end{axis}
    
\end{tikzpicture}
            \caption{Labor market transitional dynamics with higher $T^{w}$.\label{fig:labor_market_institutions}}
                \end{center}
            \end{figure}
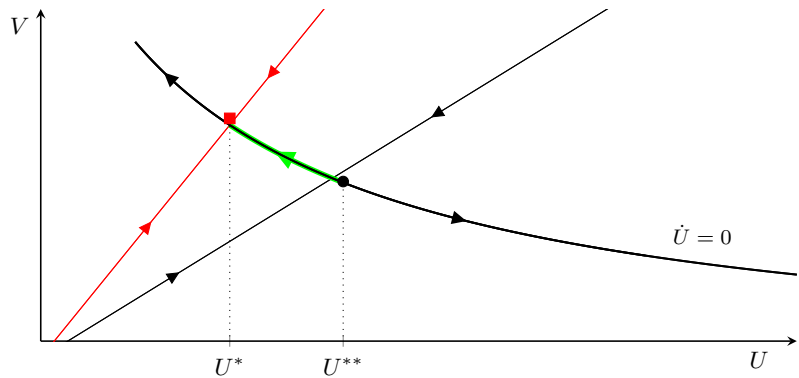

            As a consequence of the lower $w$ resulting from the new equilibrium in the labor market, we have that $\mu$ goes up. Using this result in the stationary solution of the Euler equation \eqref{eq:euler} in the main text, we have that $\dot{\hat{c}}=0$ moves to the left because  $\hat{k}$  decreases with a higher rate of return of capital (see \eqref{eq:online_marginal} in the main Appendix).

            Similarly, using the third line in \eqref{eq:syst_differential_eq}, we have that 
            
            \begin{equation*}
            \hat{c} = g \frac{\hat{k}}{q} + \frac{\mu \hat{y}}{1+\mu} - \hat{xi} V - \hat{\tau}
            \end{equation*}
            
            when $\dot{\hat{k}}=0$. Thus, the isocline $\dot{\hat{k}}=0$   will increase given the rise in $\mu$ and $L$ (which lowers $\hat{\xi}$ and $\hat{\tau}$). The end results are summarized in Figure

                    \begin{figure}
\begin{center}
\begin{tikzpicture}
     \begin{axis}[
        axis x line=middle,
                    axis y line=middle,
         name=plotramsey,
          title={}, 
              title style={at={(axis description cs:0.1,0.965)}, anchor=south} , 
 width=0.75\textwidth, height=5.5 cm,
ymin=0, ymax = 12,
xmin=0, xmax=18,
xshift=-0.5cm,
 xlabel = {$\hat{k}_{t}$},
    ylabel = {$\hat{c}_{t}$},
  xtick = {6,8,15.5},
  xticklabels={$\hat{k}^{*}_{1}$, $\hat{k}^{*}_{0}$,$\hat{k}^{\text{max}}$},
 ytick = \empty, 
        every axis x label/.style={
    at={(ticklabel* cs:0.98)},
    anchor=north,
},
  dot/.style={circle,fill=black,minimum size=4pt,inner sep=0pt,
            outer sep=-0pt},
every axis y label/.style={
    at={(ticklabel* cs:0.98)},
    anchor=east,
}
]

\draw[thick,<->] (0,12) node[left]{$\hat{c}_{t}$}--(0,0) node[below right]{$0$}--(19,0) node[below]{$\hat{k}_{t}$};

\draw[thick,red,dashed] (6,0)--(6,10.5) node[above]{\tiny  $\dot{\hat{c}}(\uparrow T^{w})=0$};

\draw[thick, black] (8,0)--(8,10.75) node[above,xshift=1em, black]{\tiny  $\dot{\hat{c}}=0$};

\draw[thick](2.5,0) ..controls (2.5,0) and (7.0,15) .. (14,0);    

\draw[thick, red, dashed](1.5,0) ..controls (1.5,0) and (7.0,18) .. (15.5,0);

         \draw[line width=2pt,dashed, green] (8,7.3)--(8,10.7);

\draw[thick,red](3.5,2) ..controls (3.5,2) and (6,7.95) .. (8.5,11) node[color=red,
	currarrow,
	pos=0.35, 
	xscale=1,
	sloped,
	scale=1]{};

	\draw[thick,red](3.5,2) ..controls (3.5,2) and (6,7.95) .. (8.5,11)node[
	currarrow,
	color=red,
	pos=0.85, 
	xscale=-1,
	sloped,
	scale=1]{};

	\draw[thick](4,1) ..controls (4,1) and (7,5.95) .. (10,9) node[
	currarrow,
	pos=0.65, 
	xscale=1,
	sloped,
	scale=1]{};

	\draw[thick](4,1) ..controls (4,1) and (7,5.95) .. (10,9) node[
	currarrow,
	pos=0.85, 
	xscale=-1,
	sloped,
	scale=1]{};

		\draw[line width=2pt,dashed, green](6,6.8) ..controls (6,6.8) and (8,10.7) .. (8,10.7) node[color=red,
	currarrow,
	color=red,
	pos=0.1, 
	xscale=-1,
	sloped,
	scale=0]{};

\addplot[color = red, mark = *, only marks, mark size = 2pt] coordinates {(6, 7.26)};

\addplot[color = black, mark = *, only marks, mark size = 2pt] coordinates {(8, 6.8)};

\end{axis}
    
\end{tikzpicture}
            \caption{Capital market transitional dynamics following an increase in $T^{w}$.\label{fig:capital_market_trans_inst}}
                \end{center}
            \end{figure}
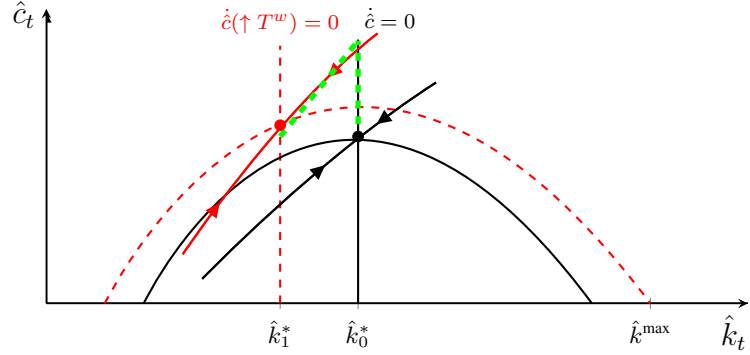

\subsubsection{Proof Lemma \ref{lemma:lemma_2A_AR}}\label{subappendix:proof_lemma_automation}

Unlike in the main text, here $w_{J}(m)= \mathrm{lim}_{t \rightarrow \infty} W_{t}/(P_{t} e^{\alpha J^{*}_{t}} D(h_{t})^{a_{1}})$ and $w_{M}(m)= \mathrm{lim}_{t \rightarrow \infty} W_{t}/(P_{t} e^{\alpha M_{t}} D(h_{t})^{a_{1}})$. Following the steps outlined by \citeasnoun{acemoglu2018race}, $w_{J}(m)$ can be derived from the ideal price index condition as (time arguments are ignored to save notation)

\begin{equation*}
w_{J}(m)^{1-\sigma} = \frac{P^{\sigma-1} - (1-m^{*})\Big[\frac{\delta D(h)^{a_{0}}}{A^{k} q}\Big]^{1-\sigma}}{\int_{0}^{m^{*}} e^{\alpha (\sigma-1)j}\mathrm{d}j }
\end{equation*}

Holding $\mu^{*}$ fixed with  small changes in $m$, it follows that

\begin{equation*}
\frac{w'_{J}(m)}{w_{J}(m)} = \frac{1}{1-\sigma} \Bigg[ \frac{\Big[\frac{\delta D(h)^{a_{0}}}{A^{k} q}\Big]^{1-\sigma} -w^{1-\sigma}_{M}(m) }{w^{1-\sigma}_{M}(m) \int_{0}^{m^{*}} e^{\alpha (1-\sigma) j} \mathrm{d}j }\Bigg]. 
\end{equation*}

Similarly,

\begin{equation*}
\frac{w'_{M}(m)}{w_{M}(m)} = \frac{1}{1-\sigma} \Bigg[ \frac{\Big[\frac{\delta D(h)^{a_{0}}}{A^{k} q}\Big]^{1-\sigma} -w^{1-\sigma}_{J}(m) }{w^{1-\sigma}_{M}(m) \int_{0}^{m^{*}} e^{\alpha (1-\sigma) j} \mathrm{d}j }\Bigg] 
\end{equation*}

since $w_{J}^{1-\sigma}(m)\int_{0}^{m} e^{\alpha(\sigma-1)j }\mathrm{d}j = w_{M}^{1-\sigma}(m)\int_{0}^{m} e^{\alpha(1-\sigma)j} \mathrm{d}j$ and $e^{\alpha(1-\sigma)m}  \int_{0}^{m} e^{\alpha(\sigma-1)j }\mathrm{d}j = \int_{0}^{m} e^{\alpha(1-\sigma)j} \mathrm{d}j $. 

Define $\bar{q}$ as the relative price of capital for which $m=0$. Using \eqref{eq:online_marginal}, it follows that

\begin{equation}
\bar{q}(\mu^{*}) = \frac{\delta (1+\mu^{*}) D(h)^{a_{0}}}{A^{k}}.
\end{equation}

To prove  part (i), let $\delta D(h)^{a_{0}}/(A^{k} q^{\text{min}}) = w_{J}(1)$. The proof that $q^{\text{min}}>\bar{q}$ is a direct application of the steps in \citeasnoun[p. 1532]{acemoglu2018race}. 

At $m=\bar{m}(q)$, the ideal price index condition satisfies

\begin{equation*}
\begin{split}
(1+\mu^{*})^{\sigma-1} &= \Bigg[\frac{\delta D(h)^{a_{0}}}{(A^{k} q)}\Bigg]^{1-\sigma}  \Bigg[ 1-\bar{m}(q) + \frac{e^{\alpha(\sigma-1)\bar{m}(q)}-1}{\alpha(\sigma-1)} \Bigg] \\
& \approx \Bigg[\frac{\delta D(h)^{a_{0}}}{(A^{k} q)}\Bigg]^{1-\sigma} \Big[ 1+ \frac{\alpha (\sigma-1)}{2} \bar{m}(q)^{2} \Big].
\end{split}
\end{equation*}

The second line uses a Taylor expansion, which provides a reasonable approximation because $m \in (0,1)$ and $\sigma$ is most likely a number not very different to 1. From this it follows that 

\begin{equation*}
\bar{m}(q) \approx \sqrt{\frac{2}{\alpha (1-\sigma)} \times \Big[1-\Big(\frac{\delta(1+\mu^{*})D(h)^{a_{0}}}{A^{k}q}}\Big)^{\sigma-1}\Big].
\end{equation*}

which clearly implies that  $\bar{m}'<0$. Note, in addition, that if $q=\bar{q}$, $\bar{m}(q)  =  0$. 

 When $q=q^{\text{min}}$, we have that

\begin{equation*}
\begin{split}
(1+\mu^{*})^{\sigma-1}  &= \Bigg[\frac{\delta D(h)^{a_{0}}}{(A^{k} q^{\text{min}})}\Bigg]^{1-\sigma}  \Bigg[ 1-\bar{m}(q^{\text{min}}) + \frac{e^{\alpha(\sigma-1)\bar{m}(q^{\text{min}})}-1}{\alpha(\sigma-1)} \Bigg] \\
 & = w_{J}(1)^{1-\sigma} \Bigg[ 1-\bar{m}(q^{\text{min}}) + \frac{e^{\alpha(\sigma-1)\bar{m}(q^{\text{min}})}-1}{\alpha(\sigma-1)} \Bigg].
\end{split}
\end{equation*}

Since $w_{J}(1)^{1-\sigma}  \int_{0}^{1} e^{\alpha (\sigma-1)j} \mathrm{d}j = (1+\mu^{*})^{\sigma-1} $, it follows that $ \int_{0}^{1} e^{\alpha (\sigma-1)j} \mathrm{d}j = 1-\bar{m}(q^{\text{min}}) + \frac{e^{\alpha(\sigma-1)\bar{m}(q^{\text{min}})}-1}{\alpha(\sigma-1)}$ only  when $\bar{m}(q^{\text{min}} )=1$. 

 In the region  where $m > \bar{m}(q)$ and $q \in [q^{\text{min}}, \bar{q}]$, we have that $w_{J}(m) > w_{J}(\bar{m}(q)) =  \delta D(h)^{a_{0}}/(A^{k} q)  > w_{M}(\bar{m}(q))) > w_{M}(m)$ because $w_{J}(m)'>0$ and $w_{M}(m)'<0$ when $w_{J}(m) > \delta D(h)^{a_{0}}/(A^{k} q)  >  w_{M}(m)$. On the other hand, when $m < \bar{m}(q)$, it follows that $w_{J}(\bar{m}(q)) <  \delta D(h)^{a_{0}}/(A^{k} q) $, meaning that additional automation would reduce aggregate output, so small changes in $m$ do not affect $m^{*}$ and have no effects on the economic equilibrium.

To prove  part (ii), let $\delta D(h)^{a_{0}}/(A^{k} q^{\text{max}}) = w_{M}(1)$. The proof that $q^{\text{max}}>\bar{q}$ is a direct application of the steps in \citeasnoun[p. 1532]{acemoglu2018race}. Correspondingly, when $m=\tilde{m}(q)$, it follows that

\begin{equation*}
\begin{split}
(1+\mu^{*})^{\sigma-1} &= \Bigg[\frac{\delta D(h)^{a_{0}}}{(A^{k} q)}\Bigg]^{1-\sigma}  \Bigg[ 1-\tilde{m}(q) + \frac{e^{\alpha(1-\sigma)\tilde{m}(q)}-1}{\alpha(1-\sigma)} \Bigg] \\
& \approx \Bigg[\frac{\delta D(h)^{a_{0}}}{(A^{k} q)}\Bigg]^{1-\sigma} \Big[ 1+ \frac{\alpha (1-\sigma)}{2} \tilde{m}(q)^{2} \Big].
\end{split}
\end{equation*}

The second line uses a Taylor approximation of the exponential function, such that

\begin{equation*}
\tilde{m}(q) \approx \sqrt{\frac{2}{\alpha (\sigma-1)} \times \Big[1-\Big(\frac{\delta(1+\mu^{*})D(h)^{a_{0}}}{A^{k}q}}\Big)^{\sigma-1}\Big].
\end{equation*}

Clearly, $\tilde{m}'>0$. Similarly, note that when $q=\bar{q}$, $\tilde{m}(q) =0$. When $q=q^{\text{max}}$, $\delta D(h)^{a_{0}}/(A^{k} q^{\text{max}} )= w_{M}(1)$ implies that 

\begin{equation*}
\begin{split}
w_{M}(1)^{1-\sigma} \int_{0}^{1} e^{\alpha(1-\sigma)j}\mathrm{d}j &= \Big[w_{M}(1)\Big]^{1-\sigma}  \Bigg[ 1-\tilde{m}(q) + \frac{e^{\alpha(1-\sigma)\tilde{m}(q)}-1}{\alpha(1-\sigma)} \Bigg] \\
\frac{e^{\alpha(1-\sigma)}-1}{\alpha(1-\sigma)} &= 1-\tilde{m}(q^{\text{max}}) + \frac{e^{\alpha(1-\sigma)\tilde{m} (q^{\text{max}})}-1}{\alpha(1-\sigma)}.
\end{split}
\end{equation*}

Clearly, the previous equation only holds if $\tilde{m} (q^{\text{max}})=1$. 

In this region, because $w_{J}(m)> \delta D(h)^{a_{0}}/(A^{k} q) > w_{M}(m)$ we have that $w'_{J}(m)>0$ and  $w'_{M}(m)<0$. Thus, for $m > \tilde{m}(q)$, $w_{J}(m) > w_{J}(\tilde{m}(q)) > \delta D(h)^{a_{0}}/(A^{k} q)  = w_{M}(\tilde{m}(q)) >w_{M}(m)$. On the other hand, for  $m < \tilde{m}(q)$, $ \delta D(h)^{a_{0}}/(A^{k} q)  < w_{M}(m)$, which means that new  tasks would reduce aggregate output, so they are not adopted. 

\section{Data Description}\label{appendix:data}

 In the main text, I used the experimental BEA-BLS integrated of \citeasnoun{eldridge2020toward} from 1947 to 2016.  Particularly, given the emphasis of the paper on \emph{production} relations between firms and workers, I focused on sectors which do not require imputing a value-added onto to them in order to make them equal to their respective incomes. This is the case, for example, with Finance, Insurance, and Real Estate (FIRE) sectors, Education and health Services, and Professional and Business Services. Focusing on the non-farming economy, Table \ref{table:bea_codes} summarizes the BEA-BLS industry categories used in the main text.

\begin{table}\caption{BEA and BLS classification codes. }
\begin{center}
\resizebox{0.95\textwidth}{!}{
\begin{tabular}{l   l}
\toprule
 BEA industry category & BLS industry category \\
 \hline 
Utilities & Utilities\\
Construction & Construction\\
Manufacturing & \\
& Wood products\\
& Nonmetallic mineral products\\
& Primary metals\\
& Fabricated metal products\\
& Machinery\\
& Computer and electronic products\\
& Electrical equipment, appliances, and components\\
& Motor vehicles, bodies and trailers, and parts\\
& Other transportation equipment\\
& Furniture and related products\\
& Miscellaneous manufacturing\\
& Food and beverage and tobacco products\\
& Textile mills and textile product mills\\
& Apparel and leather and allied products\\
& Paper products\\
& Printing and related support activities\\
& Petroleum and coal products\\
& Chemical products\\ 
& Plastics and rubber products\\ 
Whole sale trade & Whole sale trade\\
Retail  trade & Retail  trade \\
Transporting and warehousing & \\
&  Air transportation\\ 
& Rail transportation\\
& Water transportation\\
& Truck transportation\\
&Transit and ground passenger transportation\\
& Pipeline transportation\\ 
& Other transportation and support activities\\
& Warehousing and storage\\
 Information  & \\
 & Publishing industries, except internet (includes software)\\
 & Motion picture and sound recording industries\\
& Broadcasting and telecommunications\\
& Data processing, internet publishing, and other information services\\
Administrative and waste management services & \\
& Administrative and support services\\
& Waste management and remediation services\\
Arts, entertainment, and recreation & \\
& Performing arts, spectator sports, museums, and related activities\\
& Amusements, gambling, and recreation industries\\
Accommodation and food services & \\
& Accommodation\\
& Food services and drinking places\\
Other services, except government & Other services, except government\\
\bottomrule
\end{tabular}}
\end{center}
\label{table:bea_codes}
\end{table}

For these sectors, I estimated the rate of return of capital ($\mu^{\text{BEA-BLS}}_{t})$ as follows

\begin{equation*}
\mu^{\text{BEA-BLS}}_{t} = \frac{P_{t}Y_{t}-P^{c}Y_{t}}{P^{c}Y_{t} } = \frac{P_{t} Y_{t} - \delta P^{k}_{t} K_{t}- W_{t} L_{t}}{W_{t} L_{t} +  \delta P^{k}_{t} K_{t}}
\end{equation*}

where $\delta P^{k}_{t} K_{t}$ is the  Current-Cost Depreciation of Private Fixed Assets obtained from Table 3.4ESI from the BEA Fixed Assets Accounts Tables, $P_{t}Y_{t}$ is the nominal gross output minus nominal intermediate output, and $W_{t} L_{t} $ is the sum nominal college labor input and
nominal non-college labor input in \citeasnoun{eldridge2020toward}. 

The capital-output ratio $P_{t} Y_{t}/ P^{k}_{t} K_{t}$ is the   nominal gross output minus nominal intermediate output in \citeasnoun{eldridge2020toward} over the sum of  Current-Cost Net Stock of Private Fixed Assets and the  Current-Cost Depreciation of Private Fixed Assets. 

The wage-premium (BEA-BLS) is measured as

\begin{equation*}
w^{\text{bea-bls}} = \frac{\text{(nominal college labor input)}/\text{(quantity index college labor input)}}{\text{(nominal non-college labor input)}/\text{(quantity index non-college labor input)}}
\end{equation*}

The depreciation rate $\delta^{\text{BEA-BLS}}$ used in equation \eqref{eq:empirical_automation} is obtained is the monthly average of  Current-Cost Depreciation of Private Fixed Assets over the sum of  Current-Cost Net Stock of Private Fixed Assets and the  Current-Cost Depreciation of Private Fixed Assets.  This value is approximately $0.056\%$.  For clarity, $K_{t}/(qY_{t})$ in \eqref{eq:empirical_automation} is also the monthly capital-output ratio, which is the annual value divided by 12.

\begin{figure}
    \begin{center}
    \includegraphics[width=0.8\textwidth]{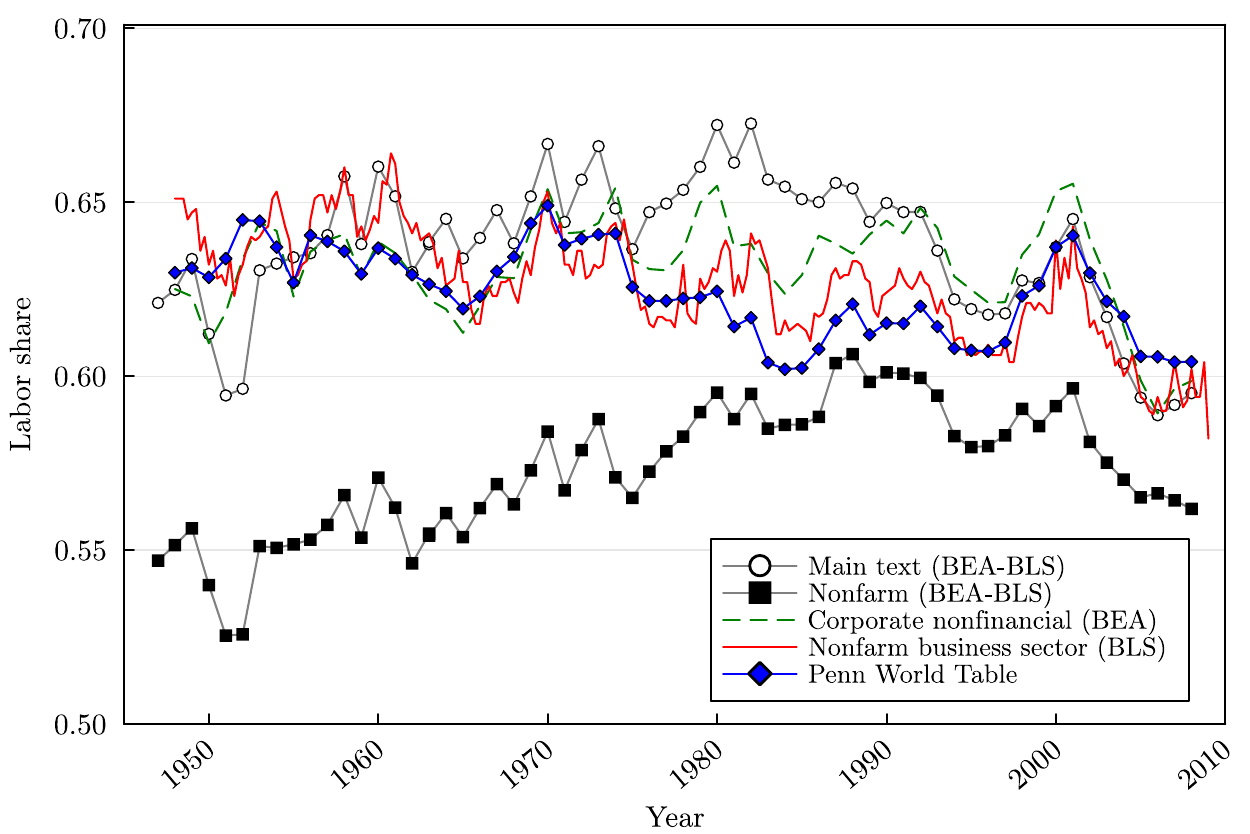}
    \end{center}
          \caption{Labor shares. \emph{Notes--- The black line is the nonfarm  (BEA-BLS) measure, the gray line with white circles is the labor share using the sectors in Table \ref{table:bea_codes}.} \label{fig:labor_shares_US}}
  \end{figure}

\subsection{Measures of the Labor share and  Labor shares by sector}  Figure \ref{fig:labor_shares_US} shows five different measures of the labor share. The  nonfarm  (BEA-BLS) data and the corporate nonfinancial (BEA) data are the only two which  exhibit a clear downward trend of the labor share after the 2000s. In contrast, Figure \ref{fig:predictedLS_several} shows that the remaining three measures of the labor share are broadly consistent with the predicted paths of the model generated by allowing changes in technology and institutions.

What explains the difference  the nonfarm  (BEA-BLS)  and the corporate nonfinancial (BEA) data with the other measures? Part of the answer can be found by analyzing the behavior of individual sectors in the economy. Using the BEA-BLS data,  Figures \ref{fig:labor_shares1_US} and \ref{fig:labor_shares2_US} show that the difference between the nonfarm (BEA-BLS) labor share and that following the industry categories in Table \ref{table:bea_codes}   can be explained by the data in sectors with questionable value added imputations. Even if we exclude the Finance and Insurance sector, there is a number of service sector with questionable imputations.  The data associated with Service sectors (2) in Figure \ref{fig:labor_shares2_US}, for instance, exhibit a sharp increase in the labor share even after the 1970s when the institutional support to labor started to decrease.

\begin{figure}
    \begin{center}
    \includegraphics[width=0.8\textwidth]{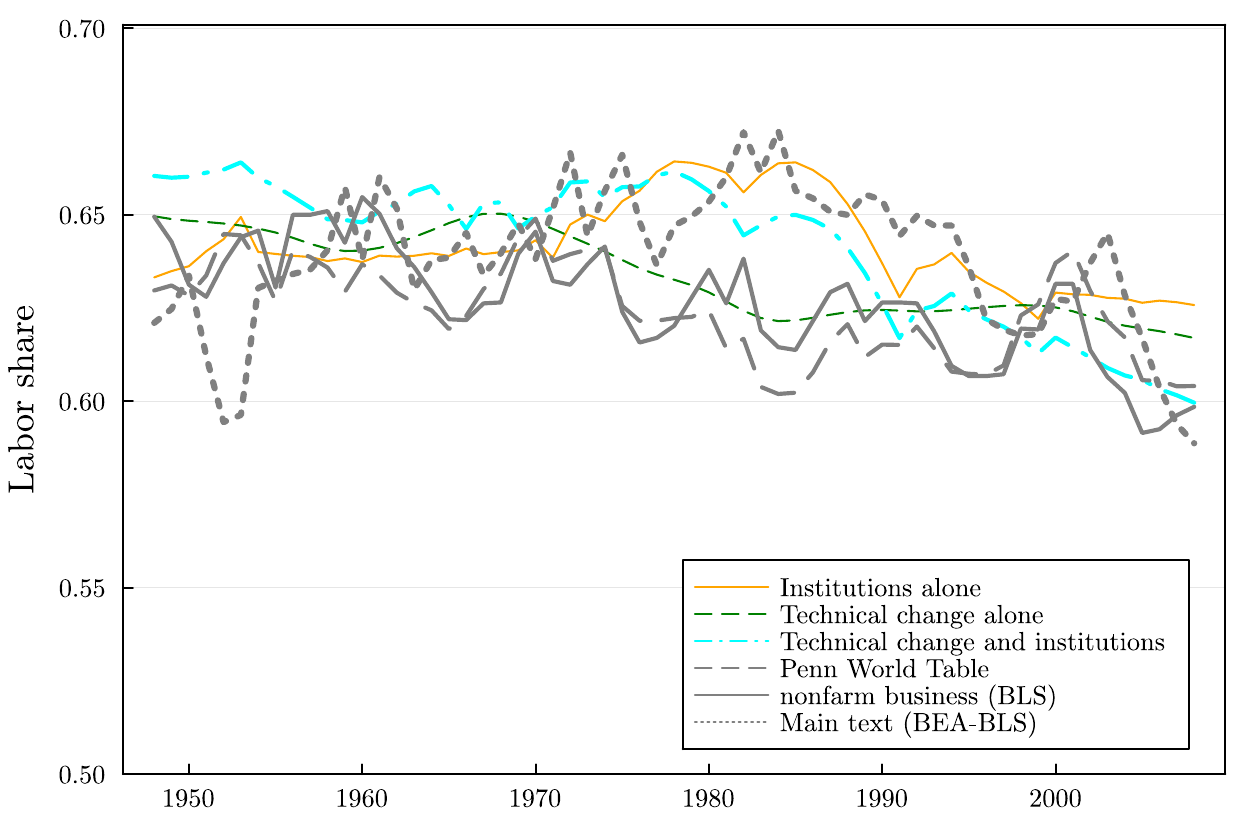}
    \end{center}
          \caption{Predicted and real labor shares. \label{fig:predictedLS_several}}
  \end{figure}

\begin{figure}
    \begin{center}
    \includegraphics[width=0.85\textwidth]{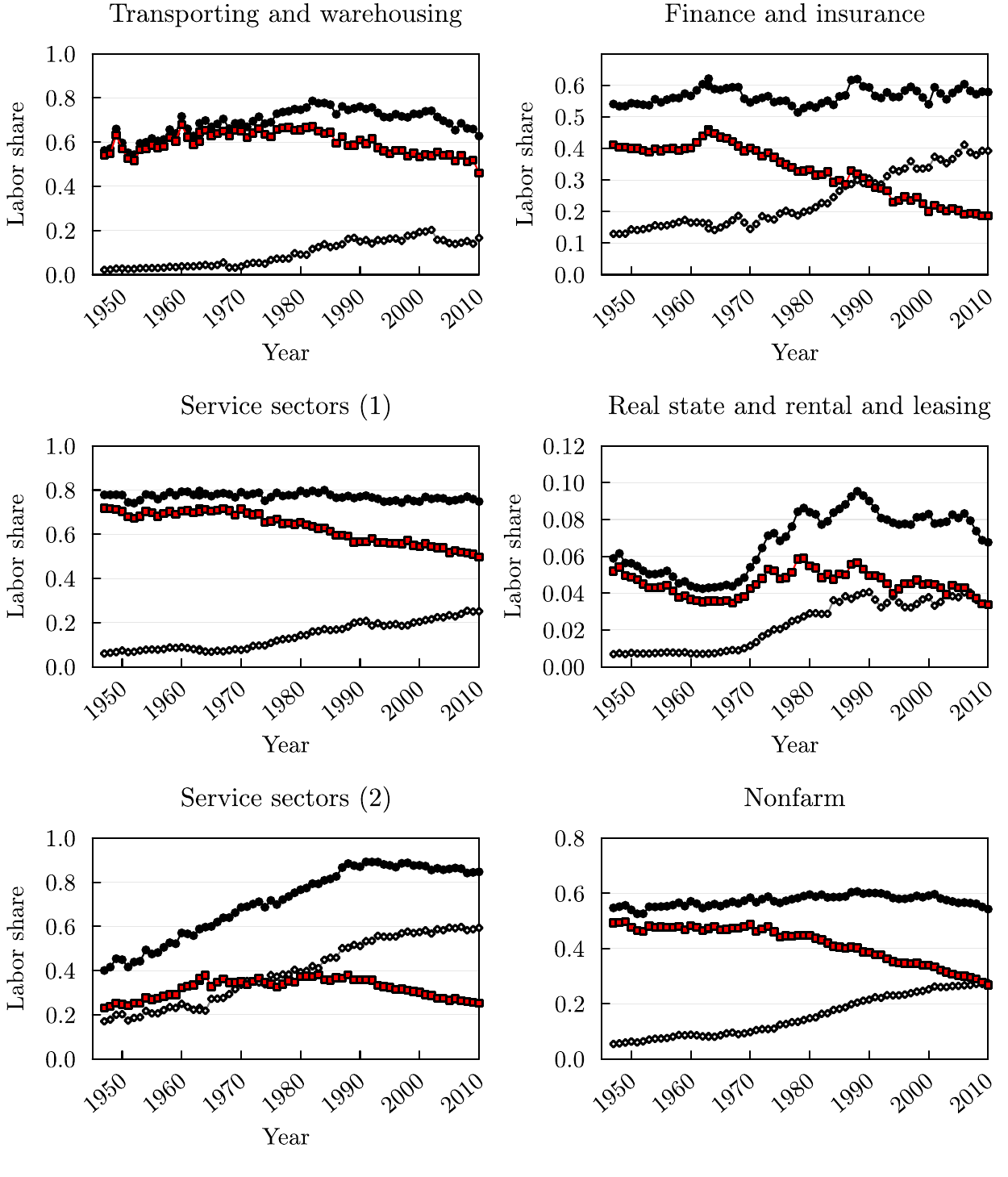}
    \end{center}
          \caption{Labor share by sectors.    \emph{Notes--- Services (1) include:Administrative and waste management services, Arts, entertainment, and recreation, Accommodation, Food services and drinking places, and Other services, except government. Services (2) include: Professional, scientific, and technical services, Management of companies and enterprises,  Educational services,  and Health care and social assistance. The black lines with circles is the total labor share, the red lines with squares is the share on non-college labor, and the gray lines with diamonds  is the share of college labor.} \label{fig:labor_shares1_US}}
  \end{figure}

\begin{figure}
    \begin{center}
    \includegraphics[width=0.85\textwidth]{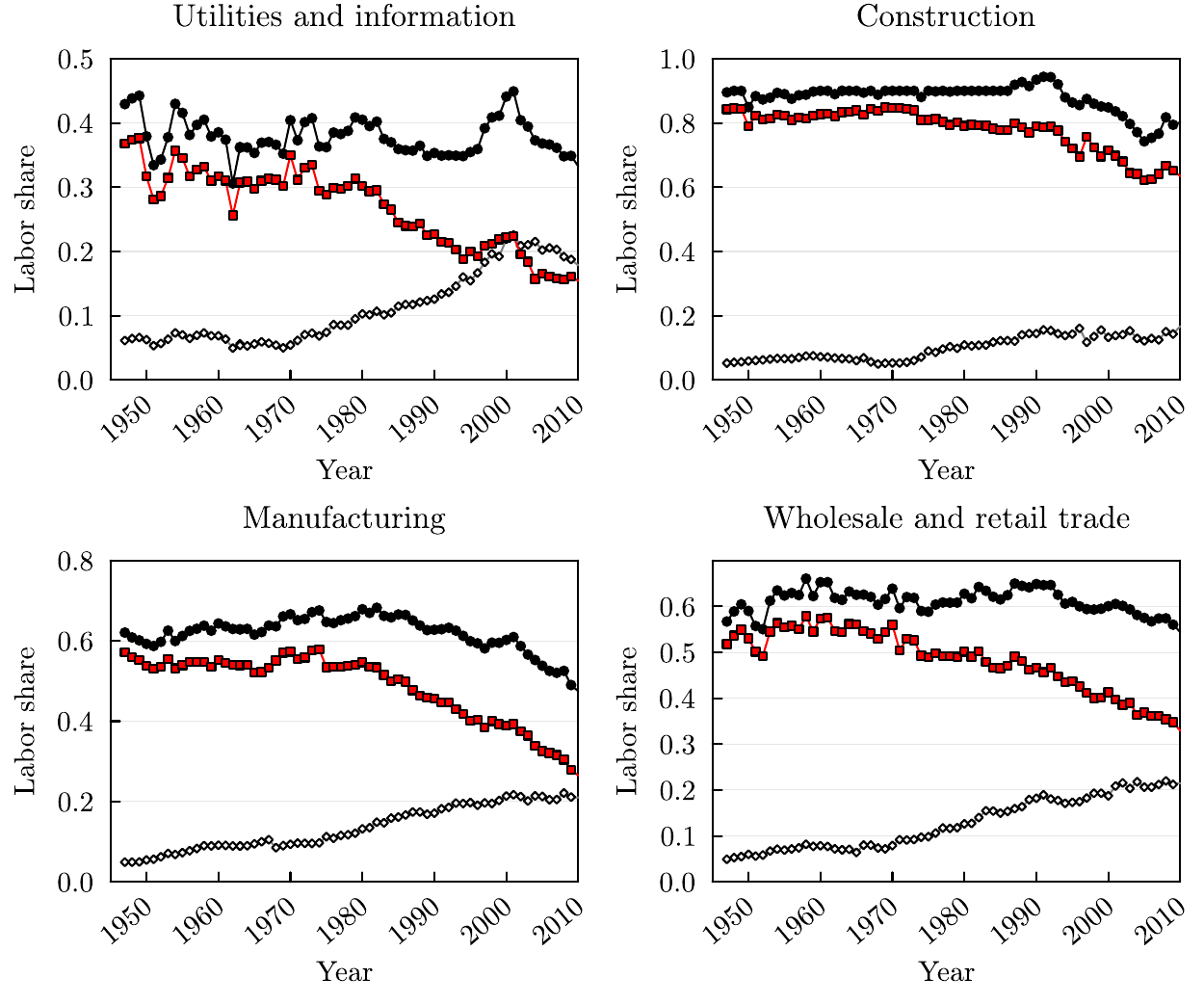}
    \end{center}
          \caption{Labor share by sectors.  \emph{Notes---The black lines with circles is the total labor share, the red lines with squares is the share on non-college labor, and the gray lines with diamonds  is the share of college labor.}\label{fig:labor_shares2_US}}
  \end{figure}

Thus, although the nonfarm business sector (BLS) and the Penn World Table data are popular measures, they may underestimate the fall of the labor share since the late 1970s given some questionable imputations of labor income in specific service sectors and because they include proprietor income as a component of the labor share. This  problem is depicted in Figure \ref{fig:profitability_US} below. The BEA-BLS integrated data, however, is still experimental and is subject to important measurement.

\subsection{Capital-output ratio} Figure \ref{fig:capital_outputUS} depicts three different series of the annual capital-output ratio. The (BEA) Corporate time series is obtained by summing the value of net stocks with the depreciation of capital in the corporate sector from Tables 6.4 and 6.1 of the Fixed Assets Accounts Tables over the gross value added of corporate business from Table 1.14 of the National Income and Product Accounts. 

There are two important conclusions that can be drawn from Figure \ref{fig:capital_outputUS}.   First,  the capital-output ratio used in the main text is very similar to the capital-output ratio of the corporate sector obtained from the BEA. This means that the automation measure obtained from \eqref{eq:empirical_automation} is robust to alternatives measures  of the capital-output ratio.  The second conclusion is that, though all measures follow a similar pattern, the nonfarm BEA-BLS capital-output ratio   is much higher  than the other two : about 1.6 times greater than the BEA capital-output ratio and close to 1.66 times higher than the measure used in the main text.  This suggests that the main problems in the BEA-BLS experimental data are probably found in the sectors excluded from Table \ref{table:bea_codes}.

\begin{figure}
    \begin{center}
    \includegraphics[width=0.8\textwidth]{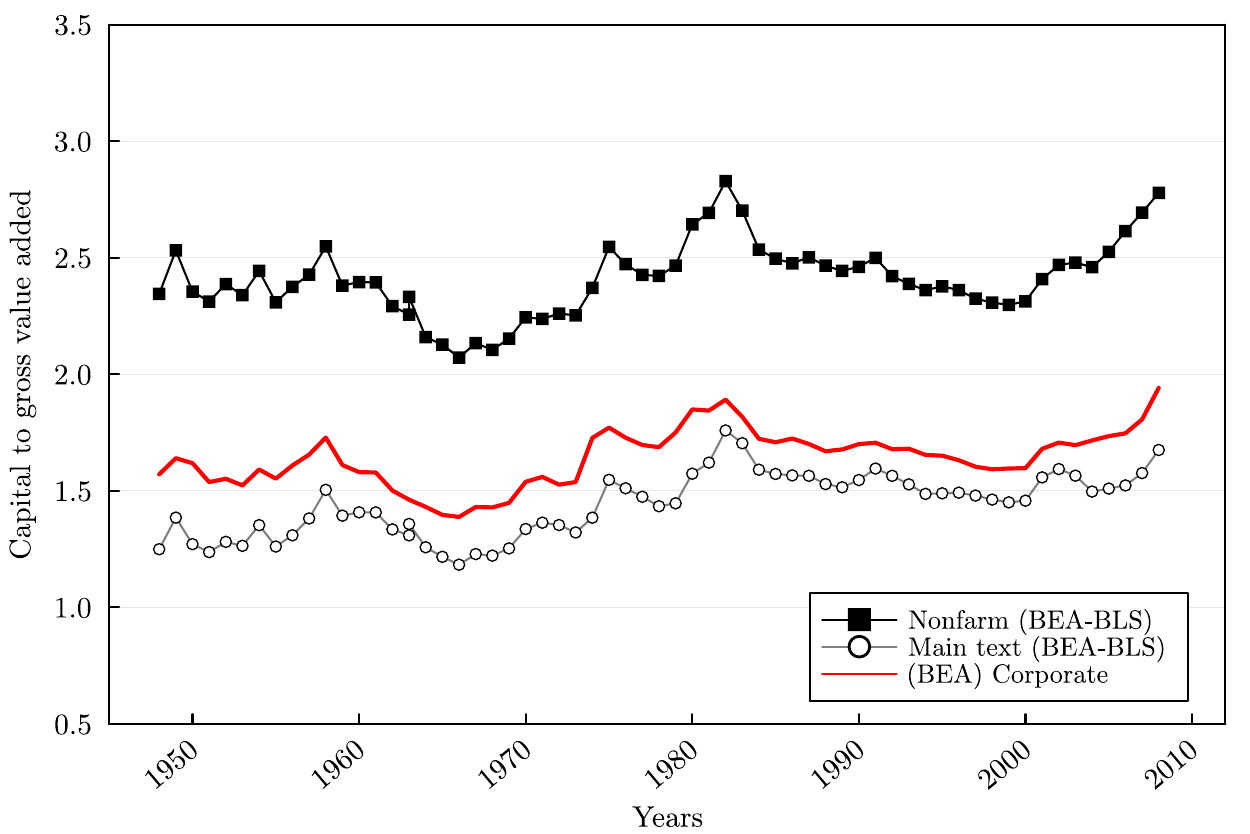}
    \end{center}
          \caption{Capital-output ratio. \label{fig:capital_outputUS}}
  \end{figure}

  \subsection{Profitability} Figure \ref{fig:profitability_US} depicts three different measures of the rate of return of capital and compares them with the proprietor's labor compensation share, which depends on the proprietor's return to capital.\footnote{Here ``proprietors'' is taken to mean ``unincorporated proprietors''. See  \url{https://www.bls.gov/opub/mlr/2017/article/estimating-the-us-labor-share.htm} for additional details. }  The data shows that the behavior of proprietor's labor compensation share is remarkably similar with the rate of return of capital in the main text. Particularly, both present a sharp decline before the 1980s and a strong recovery afterwards. A similar behavior is shared by the other two measures of business profitability, though in a lesser degree. 
  
  There are two important implications of this result. First,  the  nonfarm business sector (BLS) labor share time series---in spite of already exhibiting a a downward trend---is probably underestimating the fall of the participation of workers on gross aggregate income. Second, the measure of the rate of return of capital in the main text provides a credible measure of business profitability in the postwar US economy.

    \begin{figure}
    \begin{center}
    \includegraphics[width=0.85\textwidth]{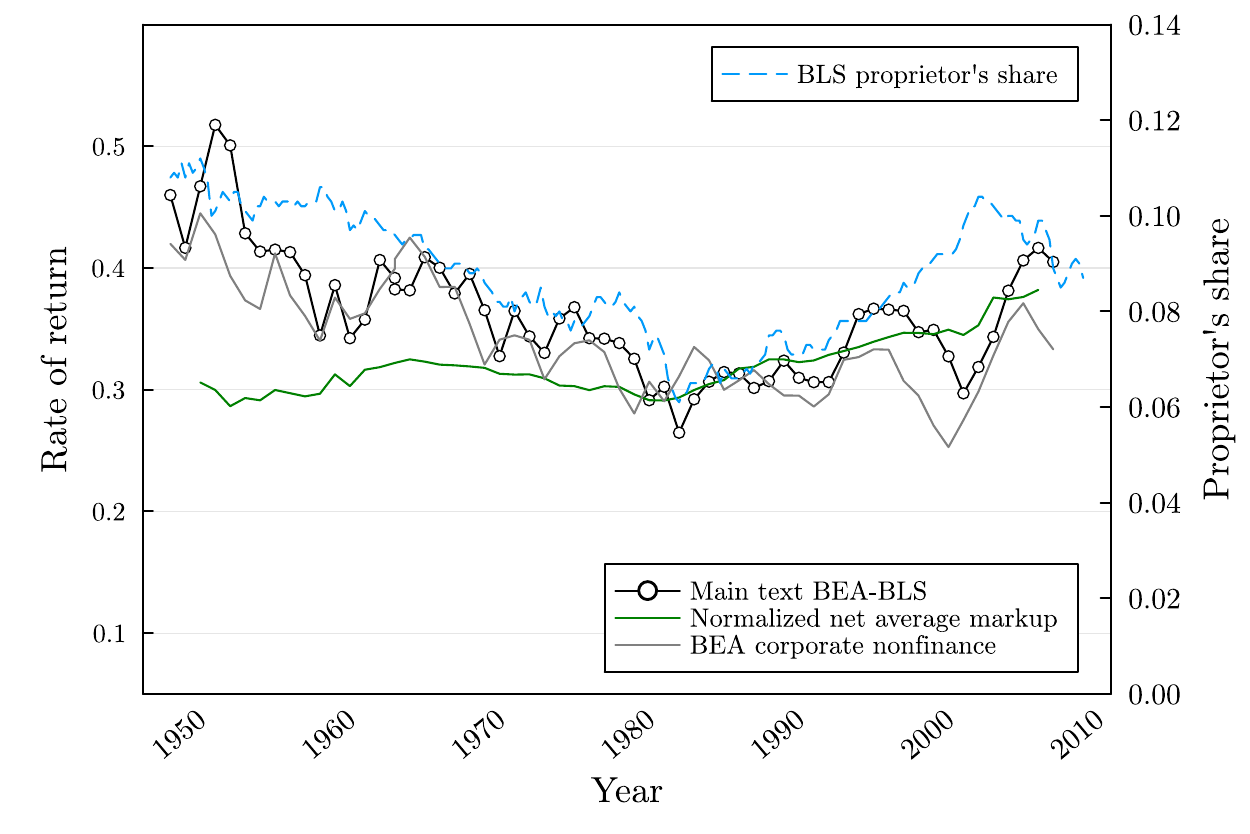}
    \end{center}
          \caption{Rates of return. \emph{Notes--- The  net average markup of \protect \citeasnoun{de2020rise} is normalized so that in 1980 is equal to the rate of return of capital in the main text.} \label{fig:profitability_US}}
  \end{figure}

    Figure \ref{fig:profitability_wp} complements the results in Figure \ref{fig:inst_prof_wagep} by showing that, particularly since the mid-1960s, there is a clear correlation between the different measures of the wage-premium and the rate of return of capital. Thus, given the  connection between the rate of return of capital and the variations in worker power, it is plausible that an important part of the behavior of the wage-premium can be accounted for by  changes in labor institutions. These results do not contradict the evidence of the skill-biased technical change. Rather, it adds an additional latter to the analysis by noting that there is probably a link between the demand for high-skilled labor and business profitability, such that the fundamental question to answer is what determines the  rate of return of capital. This is a potentially fruitful area for future research since it can shed new light on how the market-driven and institution-driven are  connected in relation to the problem of wage inequality in the US.

  Figure \ref{fig:profitability_conc_US} compares the four normalized measures of corporate profitability with three different measures of market concentration. As already noted in Figure \ref{fig:concentration_hyp}, there is only a clear link between the concentration of market among large firms and the rise of business profitability after the early 1980s. If anything, the relation is negative between the  1940s and the late 1970s.

    \begin{figure}
    \begin{center}
    \includegraphics[width=0.85\textwidth]{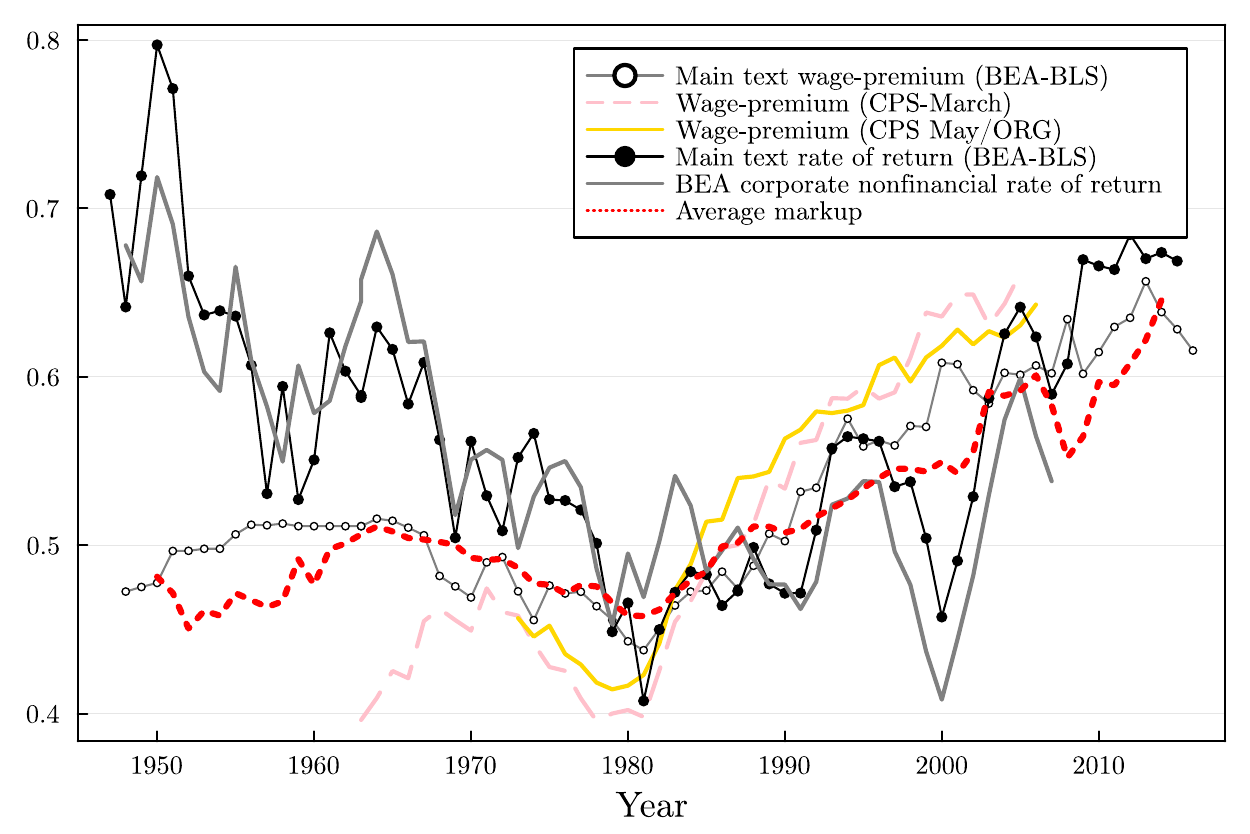}
    \end{center}
          \caption{Rates of return and wage-premium. \emph{Notes--- All measures are normalized to the 1985 value  of the wage-premium (CPS/March) in \protect \citeasnoun{autor2008trends}.} \label{fig:profitability_wp}}
  \end{figure}

    \begin{figure}
    \begin{center}
    \includegraphics[width=0.95\textwidth]{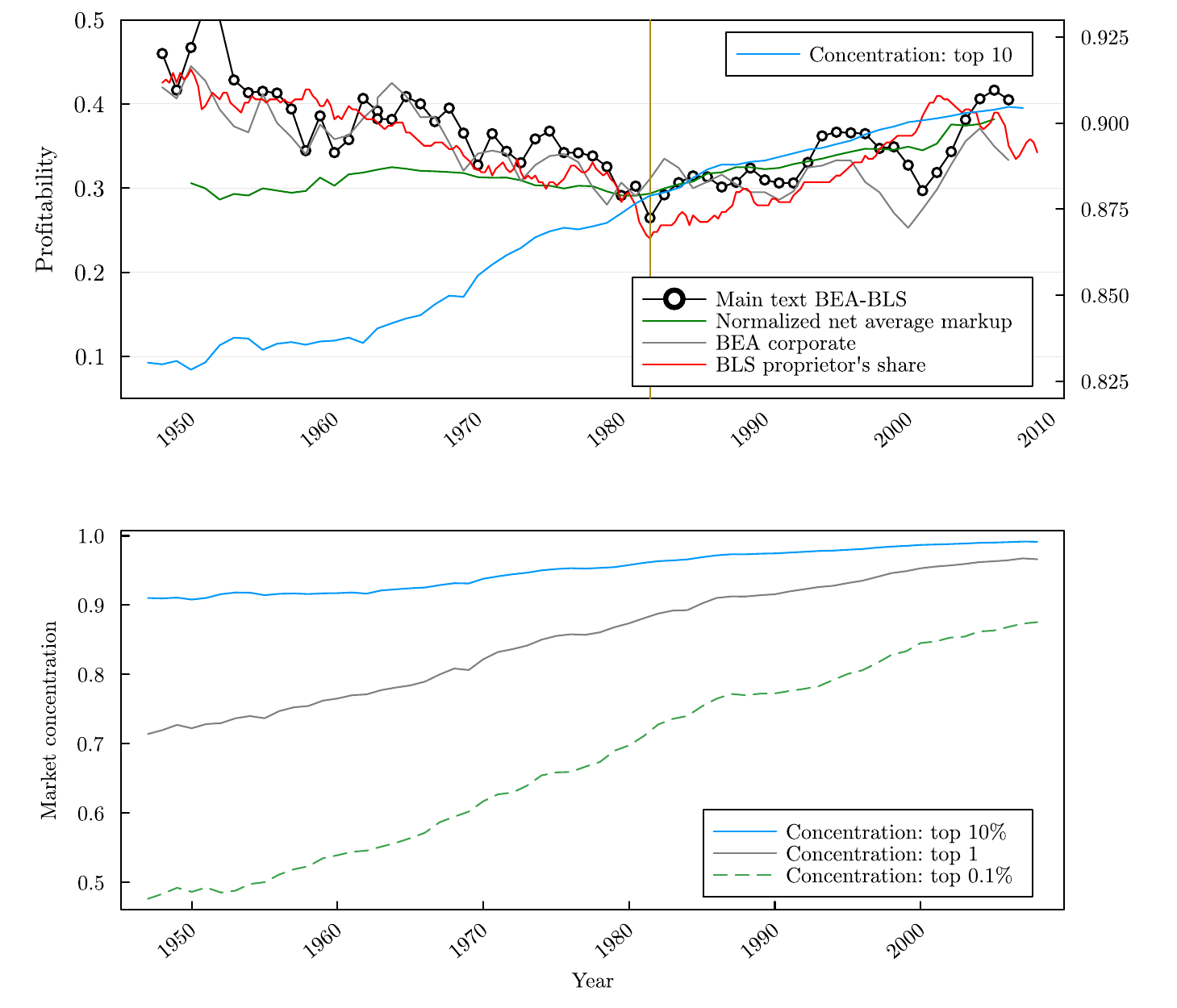}
    \end{center}
          \caption{Rates of return. \emph{Notes--- All profitability measures in the top panel are normalized so that in 1980 they are equal to the rate of return in the main text. The data in the lower panel measures the share of corporate assets accounted for by the top 10$\%$, 1$\%$, and $0.1\%$  \protect \cite{kwon2023}.} \label{fig:profitability_conc_US}}
  \end{figure}

To sum up, this appendix shows that the data used in the main text is robust to alternative measurements of the labor share,  the capital-output ratio, and business profitability.

\section{Robustness}\label{appendix:robustness} Here I present some additional results complementing Section \ref{sec:empirics} in the main text and show two robustness tests which strengthen the conclusions in the main text.  

\subsection{Additional Results} Figure \ref{fig:stability_condition_euler} shows that the stability condition of Proposition \ref{prop:gen_equilibrium} is plausible in light of the time series of the US. Particularly, we find that---with the exception of the early 1980s---the postwar US economy was probably in a condition to fund a steady growth rate of about 2$\%$ and maintain at the same time the available funds for financing capitalist consumption, taxes, and vacancy expenses. Moreover, Figure \ref{fig:stability_condition_euler} also tells an important story about the sustainability of the New Deal Order  given that, by the late 1970s, the difference between the rate of return and $g/\delta$ was coming to a minimum. It also says that the weakening power of labor probably contributed to the economic  sustainability of the system given the expansion of  $\mu$ over $g/\delta$.

   \begin{figure}
    \begin{center}
    \includegraphics[width=0.5\textwidth]{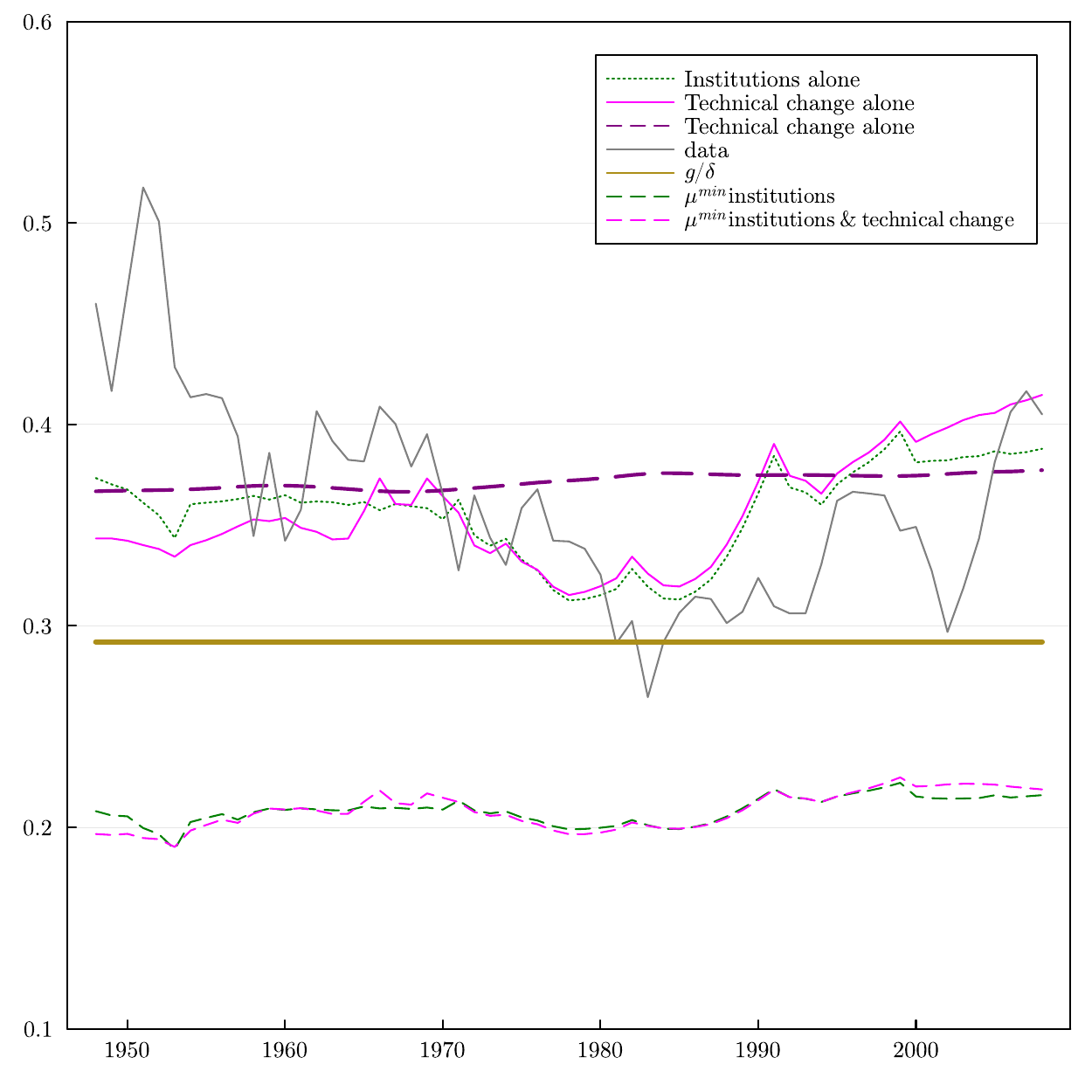}
    \end{center}
          \caption{Stability of BGP with positive growth. \emph{Notes---The growth rate $g$ is the average rate over the entire sample. See Table \ref{table:calibration_steady_state1}}. \label{fig:stability_condition_euler}}
  \end{figure}

  Figure \ref{fig:automation_regions_empirical} shows that, given the calibration in Table \ref{table:calibration_steady_state1}, the economy is always operating in region 2 of Figure \ref{fig:automation_reg} in the main text and under the condition that automated tasks raise aggregate output and are inmediately produced with capital. 
  
  \begin{figure}
    \begin{center}
    \includegraphics[width=0.9\textwidth]{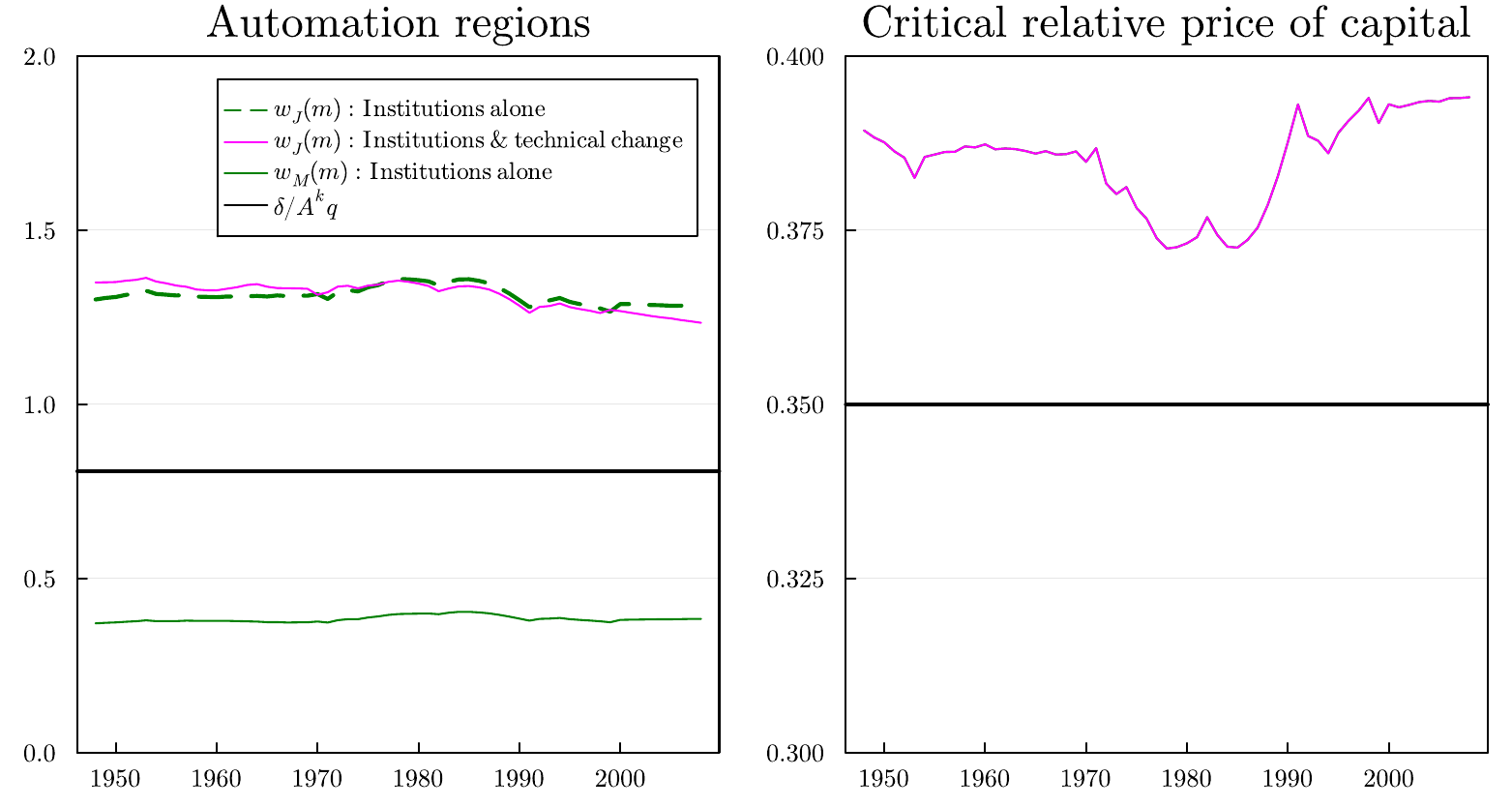}
    \end{center}
          \caption{Empirical automation regions. \label{fig:automation_regions_empirical}}
  \end{figure}

  \subsection{NAIRU Calibration} Here we consider how the model predictions change if we target the NAIRU instead of the efficient unemployment rate of \citeasnoun{michaillat2021beveridgean} employed in the main text. Given that the NAIRU is always above the efficient unemployment rate, I increase $\lambda_{0}$ from 0.02 in the main text to 0.025. Similarly, given that a higher $\lambda_{0}$ lowers the labor share by increasing the rate of unemployment, I now set $\alpha=1.7$. The remaining parameters are as in Table \ref{table:calibration_steady_state1}. 
  
  The results in Figure \ref{fig:results_NAIRU} support the results in the main text and show that there are no significant changes to the conclusions of Section \ref{sec:empirics}. 
  
  \begin{figure}
    \begin{center}
    \includegraphics[width=0.95\textwidth]{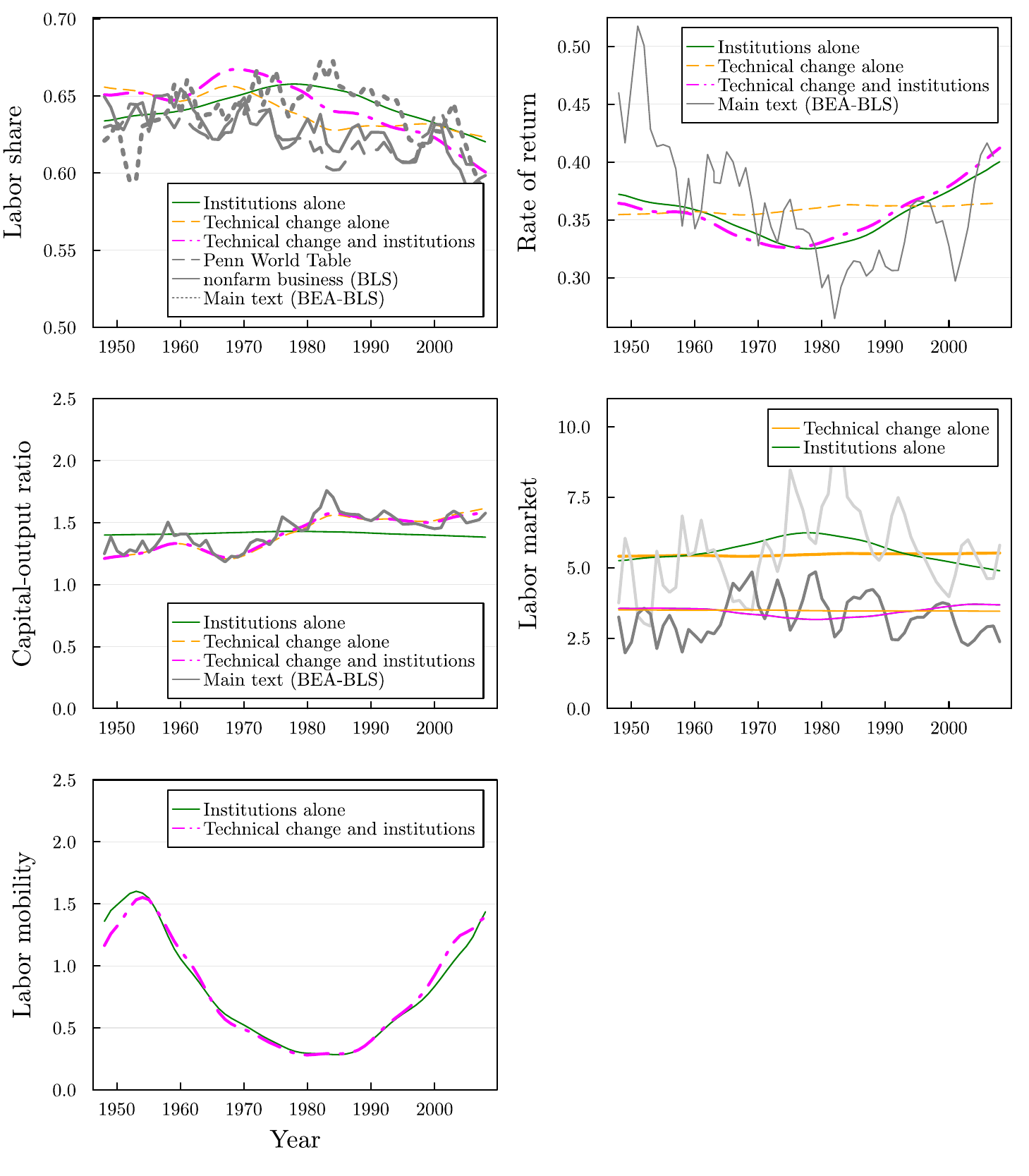}
    \end{center}
          \caption{Calibration targeting the NAIRU. \label{fig:results_NAIRU}}
  \end{figure}

\subsection{Gross substitution: $\sigma>1$} Let us now consider the case where capital and labor are gross substitutes. Targeting the efficient unemployment rate, I change $\sigma$ to 1.2 and set $\alpha=1.3$ to target an average labor share of about 0.63.

 Figure \ref{fig:results_G1} again shows that there are no considerable changes to the results in the main text by modifying the value of $\sigma$.

   \begin{figure}
    \begin{center}
    \includegraphics[width=0.95\textwidth]{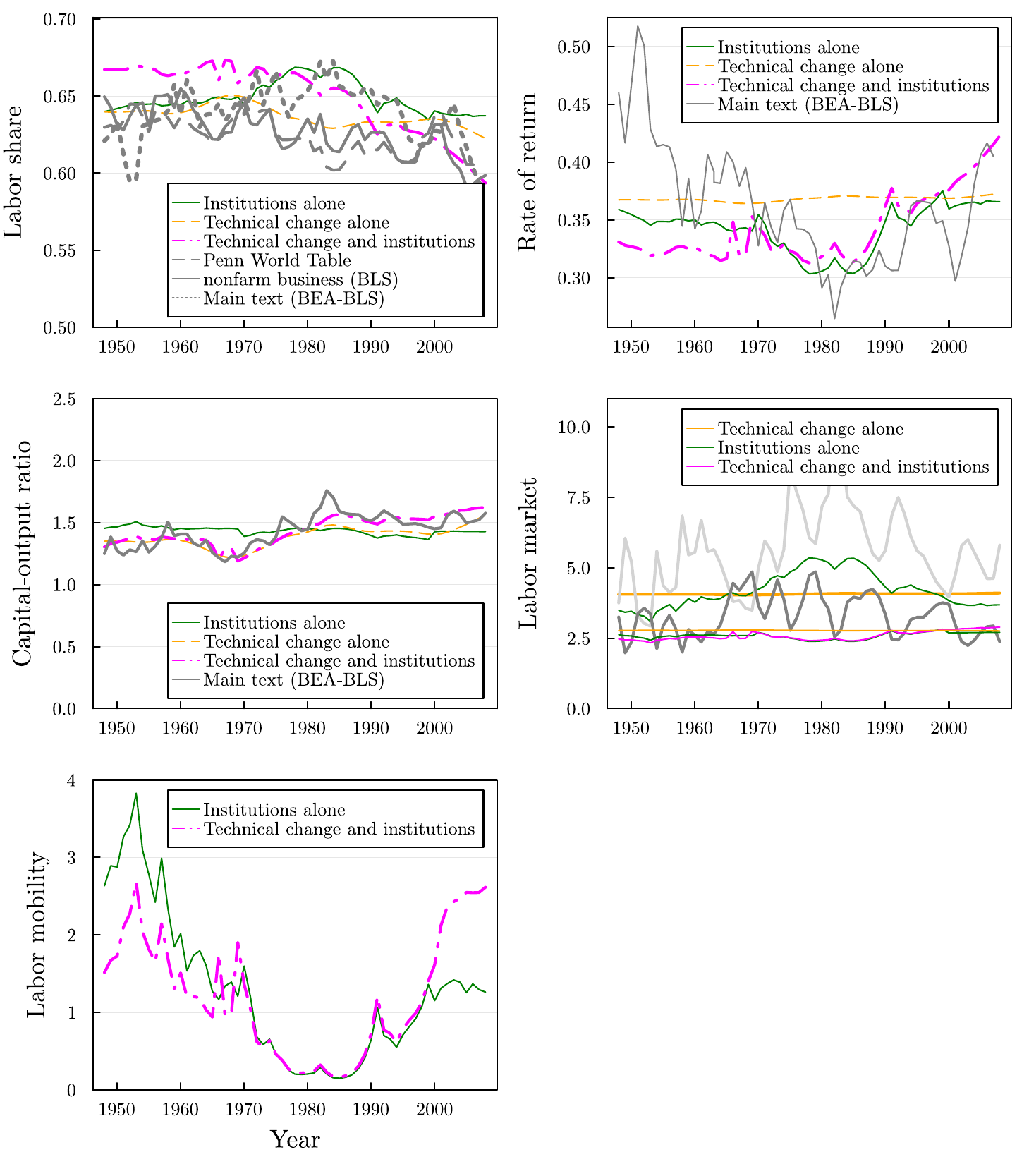}
    \end{center}
          \caption{Calibration with $\sigma>1$.  \label{fig:results_G1}}
  \end{figure}

\end{document}